\documentclass[11pt,a4paper]{article}
\pdfoutput=1 
\usepackage[utf8]{inputenc}

\usepackage[margin=1in]{geometry}
\usepackage{amsfonts}
\usepackage{amsmath}
\usepackage{amssymb}
\usepackage{amsthm}
\usepackage{float}
\usepackage[font=small,labelfont=bf]{caption} 
\allowdisplaybreaks
\usepackage{bbold}

\usepackage{physics}
\usepackage{xspace}

\usepackage{multirow}

\usepackage{bold-extra} 

\usepackage{hyperref}
\usepackage[svgnames]{xcolor}
\hypersetup{colorlinks={true},urlcolor={blue},linkcolor={DarkBlue},citecolor=[named]{DarkGreen}}

\usepackage{microtype}
\usepackage[capitalise,nameinlink,noabbrev]{cleveref}

\usepackage{tikz}
\usetikzlibrary{arrows.meta}

\usepackage{doi}

\usepackage{graphicx}

\def\dosth#1{\ifx###1##\else\dofirst#1\anytoken\fi}
\def\doagain#1\anytoken{\dosth{#1}}
\def\payoffpairs#1#2#3{\m=#1\multiply\m by 4 \advance\m by -1 \n=1
  \def\dofirst##1{\put(\n,-\m){\makebox(0,0){\strut##1}}\advance\n by 4 \doagain}%
  \dosth{#2\strut}%
  \m=#1\multiply\m by 4 \advance\m by -3 \n=3 \dosth{#3\strut}}
\def\singlepayoffs#1#2{\m=#1\multiply\m by 4 \advance\m by -2 \n=2
  \def\dofirst##1{\put(\n,-\m){\makebox(0,0){\strut##1}}\advance\n by 4 \doagain}%
  {\large\dosth{#2\strut}}}
\newcommand{\bimatrixgame}[8]{%
\setlength{\unitlength}{#1}%
\newcount\rows
\newcount\cols
\rows=#2
\cols=#3
\newcount\rowcoord
\newcount\colcoord
\rowcoord=\rows
\colcoord=\cols
\multiply\rowcoord by 4
\multiply\colcoord by 4
\newcount\m
\newcount\n
\m=\rowcoord
\n=\colcoord
\advance\m by 2 
\advance\n by 2 
\begin{picture}(\n,\m)(-2,-\rowcoord)
\m=\rows
\n=\cols
\advance\m by 1
\advance\n by 1 
\thinlines
\multiput(0,0)(0,-4){\m}{\line(1,0){\colcoord}}
\multiput(0,0)(4,0){\n}{\line(0,-1){\rowcoord}}
\put(0,0){\line(-1,1){2}}
\put(-1.5,0.5){\makebox(0,0)[r]{#4}}  
\put(-.7,1.7){\makebox(0,0)[l]{#5}}   
\n=2
\def\dofirst##1{\put(-0.8,-\n){\makebox(0,0)[r]{\strut##1}}\advance\n by 4
   \doagain}
\dosth{#6\strut} 
\n=2
\def\dofirst##1{\put(\n,1.0){\makebox(0,0){\strut##1}}\advance\n by 4
   \doagain}
\dosth{#7\strut}#8%
\end{picture}}
\usepackage{enumerate}

\theoremstyle{definition}
\newtheorem{definition}{Definition}

\theoremstyle{plain}
\newtheorem{theorem}{Theorem}[section]
\newtheorem{lemma}[theorem]{Lemma}
\newtheorem{corollary}[theorem]{Corollary}

\theoremstyle{remark}
\newtheorem{remark}{Remark}

\Crefname{claim}{Claim}{Claims}

\def\fp/{\textup{\textsf{FP}}}
\def\p/{\textup{\textsf{P}}}
\def\np/{\textup{\textsf{NP}}}
\def\conp/{\textup{\textsf{co-NP}}}
\def\fnp/{\textup{\textsf{FNP}}}
\def\tfnp/{\textup{\textsf{TFNP}}}
\def\ptfnp/{\textup{\textsf{PTFNP}}}
\def\ppa/{\textup{\textsf{PPA}}}
\def\ppad/{\textup{\textsf{PPAD}}}
\def\ppads/{\textup{\textsf{PPADS}}}
\def\ppp/{\textup{\textsf{PPP}}}
\def\pwpp/{\textup{\textsf{PWPP}}}
\def\pls/{\textup{\textsf{PLS}}}
\def\cls/{\textup{\textsf{CLS}}}
\def\ppadpls/{\textup{$\textsf{PPAD} \cap \textsf{PLS}$}}
\def\ppapls/{\textup{$\textsf{PPA} \cap \textsf{PLS}$}}
\def\eopl/{\textup{\textsf{EOPL}}}
\def\sopl/{\textup{\textsf{SOPL}}}
\def\ueopl/{\textup{\textsf{UEOPL}}}
\def\fixp/{\textup{\textsf{FIXP}}}
\def\bu/{\textup{\textsf{BU}}}
\def\bbu/{\textup{\textsf{BBU}}}
\def\linearfixp/{\textup{\textsf{Linear-FIXP}}}
\def\pspace/{\textup{\textsf{PSPACE}}}

\def\ssperner{\ensuremath{\textup{\textsc{StrongSperner}}}\xspace}
\def\gcircuit{\textup{\textsc{GCircuit}}\xspace}
\def\eol{\textup{\textsc{End-of-Line}}\xspace}

\def\pcircuit{\textup{\textsc{Pure-Circuit}}\xspace}

\def\graphical{\textup{\textsc{GraphicalWSNE}}\xspace}
\def\graphicalne{\textup{\textsc{GraphicalNE}}\xspace}
\def\polymatrix{\textup{\textsc{PolymatrixWSNE}}\xspace}
\def\polymatrixne{\textup{\textsc{PolymatrixNE}}\xspace}
\def\threshold{\textup{\textsc{ThresholdGameNE}}\xspace}

\newcommand{\garbo}{\ensuremath{\bot}\xspace}

\newcommand{\val}[1]{\boldsymbol{\mathrm{x}}[#1]}
\newcommand{\valonly}{\boldsymbol{\mathrm{x}}}
\newcommand{\valtwo}[1]{\boldsymbol{\mathrm{x'}}[#1]}
\newcommand{\valtwoonly}{\boldsymbol{\mathrm{x'}}}

\newcommand{\PURE}{\textup{\textsf{PURIFY}}\xspace}

\newcommand{\NOR}{\textup{\textsf{NOR}}\xspace}
\newcommand{\NAND}{\textup{\textsf{NAND}}\xspace}
\newcommand{\NOT}{\textup{\textsf{NOT}}\xspace}
\newcommand{\OR}{\textup{\textsf{OR}}\xspace}
\newcommand{\AND}{\textup{\textsf{AND}}\xspace}
\newcommand{\COPY}{\textup{\textsf{COPY}}\xspace}

\newcommand{\X}{\textup{\textsf{X}}\xspace}
\newcommand{\Y}{\textup{\textsf{Y}}\xspace}

\newcommand{\supp}{\textup{\textsf{supp}}\xspace}
\newcommand{\nash}[1]{\ensuremath{#1}-NE\xspace}
\newcommand{\wn}[1]{\ensuremath{#1}-WSNE\xspace}
\newcommand{\relwn}[1]{relative \ensuremath{#1}-WSNE\xspace}

\newcommand{\zero}{\textup{\textsf{zero}}\xspace}
\newcommand{\one}{\textup{\textsf{one}}\xspace}

\newcommand{\discr}{cutoff\xspace}

\newcommand{\innei}[1]{\ensuremath{N^{-}(#1)}\xspace}
\newcommand{\outnei}[1]{\ensuremath{N^{+}(#1)}\xspace}
\newcommand{\nei}[1]{\ensuremath{N(#1)}\xspace}

\newcommand{\wid}{M}
\newcommand{\dims}{N}
\newcommand{\cop}{K}
\newcommand{\vinsymbol}{u}
\newcommand{\vin}[2]{\vinsymbol_{#1,#2}}
\newcommand{\vinc}[3]{\vin{#1}{#2}^{(#3)}}
\newcommand{\vincdimension}[2]{\vinsymbol_{#1}^{(#2)}}
\newcommand{\vinctotal}[1]{\vinsymbol^{(#1)}}
\newcommand{\vout}[2]{v_{#1}^{(#2)}}
\newcommand{\vsort}[2]{w_{#1}^{(#2)}}
\newcommand{\goodcopies}{G}

\newcommand{\eps}{\ensuremath{\varepsilon}\xspace}
\newcommand{\del}{\ensuremath{\delta}\xspace}

\newcommand{\vba}{\ensuremath{\vb{a}}\xspace}
\newcommand{\vbai}{\ensuremath{\vb{a}_{-i}}\xspace}
\newcommand{\vbs}{\ensuremath{\vb{s}}\xspace}
\newcommand{\vbsi}{\ensuremath{\vb{s}_{-i}}\xspace}

\newcommand{\br}{\ensuremath{\text{br}}\xspace}

\title{\pcircuit:\\ Tight Inapproximability for \ppad/ \footnotetext{This paper combines the results of two prior conference papers:  
\textit{Pure-Circuit: Strong Inapproximability for PPAD} published in FOCS 2022 \cite{DFHM22}, and
\textit{Tight Inapproximability for Graphical Games} published in AAAI 2023~\cite{DFHM22b}.
}}
\author{
\begin{tabular}{cc}
& \\
\textbf{Argyrios Deligkas} & \textbf{John Fearnley}\\
\small{Royal Holloway, United Kingdom} & \small{University of Liverpool, United Kingdom}\\
\href{mailto:argyrios.deligkas@rhul.ac.uk}{\small{\texttt{argyrios.deligkas@rhul.ac.uk}}} & \href{mailto:john.fearnley@liverpool.ac.uk}{\small{\texttt{john.fearnley@liverpool.ac.uk}}}\\
& \\
\textbf{Alexandros Hollender} & \textbf{Themistoklis Melissourgos}\\
\small{University of Oxford, United Kingdom} & \small{University of Essex, United Kingdom}\\
\href{mailto:alexandros.hollender@cs.ox.ac.uk}{\small{\texttt{alexandros.hollender@cs.ox.ac.uk}}} & \href{mailto:themistoklis.melissourgos@essex.ac.uk}{\small{\texttt{themistoklis.melissourgos@essex.ac.uk}}}\\
& \\
\end{tabular}
}

\date{}

\begin{document}

\maketitle
\thispagestyle{empty}

\begin{abstract}
The current state-of-the-art methods for showing inapproximability in \ppad/ arise from the $\eps$-Generalized-Circuit ($\eps$-\gcircuit) problem. Rubinstein (2018) showed that there exists a small unknown constant $\eps$ for which $\eps$-\gcircuit is \ppad/-hard, and subsequent work has shown hardness results for other problems in \ppad/ by using $\eps$-\gcircuit as an intermediate problem.

We introduce \pcircuit, a new intermediate problem for \ppad/, which can be
thought of as $\eps$-\gcircuit pushed to the limit as $\eps \rightarrow 1$, and
we show that the problem is \ppad/-complete. We then prove that
$\eps$-\gcircuit is \ppad/-hard for all $\eps < 1/10$ by a reduction from \pcircuit, and thus strengthen all prior work that has used \gcircuit as an intermediate problem from the existential-constant regime to the large-constant regime.

We show that stronger inapproximability results can be derived by reducing
directly from \pcircuit. In particular, we prove tight inapproximability
results for computing approximate Nash equilibria and approximate
well-supported Nash equilibria in graphical games, for finding approximate well-supported Nash equilibria in polymatrix games, and for finding approximate equilibria in threshold games. 
\end{abstract}

\newpage
\pagenumbering{arabic}

\tableofcontents

\section{Introduction}

The complexity class \ppad/ has played a central role in determining the
computational complexity of many problems arising in game theory and
economics~\cite{Papadimitriou94-TFNP-subclasses}. The celebrated
results of Daskalakis, Goldberg, and Papadimitriou~\cite{DaskalakisGP09-Nash}
and Chen, Deng, and Teng~\cite{ChenDT09-Nash} established that finding a Nash
equilibrium in a strategic form game is \ppad/-complete, and subsequent to this
breakthrough many other \ppad/-completeness results have been
shown~\cite{CodenottiSVY08-economies-games,ChenDDT09-Arrow-Debreu,VaziraniY11-market,DengQS12-cake,Daskalakis13-approximate-Nash,KintaliPRST13-fractional-PPAD,ChenDO15-anonymous-games,ChenPY17-non-monotone-markets,SchuldenzuckerSB17-financial-PPAD,Mehta2018-constant-rank,Rubinstein18-Nash-inapproximability,DeligkasFS20-tree-polymatrix,ChenKK21-pacing,ChenKK21-throttling,PapadimitriouP21-threshold-games,GoldbergH21-hairy-ball,FRGHLP23-FPA,DaskalakisSZ21-min-max,ChenCPY22-HZ-hardness}.

These celebrated results not only showed that it is \ppad/-hard to find an exact
equilibrium, but also that finding approximate solutions is \ppad/-hard. The
result of Daskalakis, Goldberg, and Papadimitriou~\cite{DaskalakisGP09-Nash}
showed that finding an $\eps$-Nash equilibrium is \ppad/-complete when $\eps$ is
exponentially small, while the result of Chen, Deng, and
Teng~\cite{ChenDT09-Nash} improved this to show hardness for polynomially small
$\eps$. This lower bound is strong enough to rule out the existence of an FPTAS
for the problem.

The main open question following these results was whether equilibrium computation problems in \ppad/ were
hard for \emph{constant}~$\eps$, which would also rule out the existence of a
PTAS. Here one must be careful, because some problems do in fact admit
approximation schemes. For example, in the case of two-player strategic-form
games, a quasipolynomial-time approximation scheme is known~\cite{LiptonMM03-Nash},
meaning that the problem cannot be hard for a constant $\eps$ unless every
problem in \ppad/ can be solved in quasipolynomial time.
But for other types of game such results are not known. This includes
\emph{polymatrix games}, which are $n$-player games with succinct representation~\cite{Janovskaja1968-polymatrix}.

In another breakthrough result,
Rubinstein~\cite{Rubinstein18-Nash-inapproximability} developed techniques for
showing constant inapproximability within \ppad/, by proving that there exists a
constant $\eps$ such that finding an $\eps$-well-supported Nash equilibrium in a
polymatrix game is \ppad/-complete. This lower bound is obtained by first showing
constant inapproximability for the \emph{$\eps$-Generalized-Circuit} ($\eps$-\gcircuit) problem introduced by 
Chen, Deng, and Teng~\cite{ChenDT09-Nash}, and then utilizing the known reduction from \gcircuit to
polymatrix games~\cite{DaskalakisGP09-Nash}. 

Rubinstein's lower bound has since been used to show constant inapproximability
for other problems. Rubinstein himself showed constant
inapproximability for finding Bayesian Nash equilibria, relative approximate
Nash equilibria, and approximate Nash equilibria in non-monotone
markets~\cite{Rubinstein18-Nash-inapproximability}. Subsequent work has shown
constant inapproximability for finding clearing payments in financial networks
with credit default swaps~\cite{SchuldenzuckerSB17-financial-PPAD}, finding equilibria in first-price auctions with subjective priors~\cite{FRGHLP23-FPA}, finding
throttling equilibria in auction markets~\cite{ChenKK21-throttling}, 
finding equilibria in public goods games on directed
networks~\cite{PapadimitriouP21-threshold-games}, and finding consensus-halving solutions in fair division~\cite{FRFGZ18-consensus-hardness,GoldbergHIMS22-consensus-items}.

Rubinstein's lower bound is an \emph{existential-constant} result, meaning that
it shows that there exists
some constant $\eps$ below which the problem becomes \ppad/-hard. The fact that
such a constant exists is important, since it rules out a PTAS. On the other hand, Rubinstein does not give any concrete lower bound on the value of the constant (understandably so, since this was not the purpose of his work). One could of course deduce such a lower bound by a careful examination of his reduction, but it is clear that this would yield an extremely small constant.
Due to this, all of the other results that have 
utilized Rubinstein's lower bound are likewise 
existential-constant results, which rule out PTASs but do not give any
concrete lower bounds.

Ultimately, this means that existing work does not rule out an efficient
algorithm that finds, say, a $0.001$-approximate solution for any of these
problems, which would likely be more than enough for most practical needs. For
example, in a game where utilities are normalized to lie between $0$
and $1$, it is likely that a player would be more than happy to know that her
strategy is an optimal best-response, up to an additive loss of at most $0.001$ in her
utility value. Moreover, the existing work gives us no clue as to where the
threshold for hardness may actually lie. To address these questions one would
need to prove a \emph{large-constant} inapproximability result, giving hardness
for a known substantial constant.

Rubinstein's lower bound is the ultimate source of all of the recent
existential-constant lower bounds, so if one seeks a large-constant
lower bound, then Rubinstein's result is the bottleneck. 
Attempting to directly strengthen or optimize Rubinstein's result does not seem like a
promising direction. His proof, while ingenious, is very involved, and does not
lend itself to easy optimization. Furthermore, it consists of many moving parts,
so that even if one was able to optimize each module, the resulting
constant would still be very small.

\paragraph{\bf Our contribution.} In this paper we introduce the techniques
needed to show 
large-constant inapproximability results for problems in \ppad/.
Our key technical innovation is the introduction of a new intermediate problem,
called \pcircuit, which we show to be \ppad/-complete. 

Then, by reducing onwards from \pcircuit, we show
large-constant inapproximability results for a variety of problems in
\ppad/, some of which can be shown to be tight. In this sense, \pcircuit now
takes on the role that $\eps$-\gcircuit has taken in the past, as an important
intermediate problem from which all other results of this type are derived. 

The \pcircuit problem itself can be thought of as a version of $\eps$-\gcircuit that
is taken to its limits, and also dramatically simplified. In fact, the problem
has only two gates (or, in a different formulation, three gates), which should
be compared to $\eps$-\gcircuit, which has nine distinct gates. 
Perhaps more importantly, the gates in  
\pcircuit have very weak constraints on their outputs: the gates can be thought of as taking inputs
in $[0, 1]$, and producing outputs in $[0, 1]$, but 
the gates themselves 
essentially only
care about the values $0$ and $1$, with all other values being considered to be
``bad'' or ``garbage'' values (which we will later simply denote by ``$\garbo$'', instead of using values in $(0,1)$). This should be compared to
$\eps$-\gcircuit gates, where, for example, one must output a value in $[0,1]$ that is within
$\eps$ of the sum of two inputs. 

Combined, these properties make \pcircuit a very attractive problem to reduce from when
showing a hardness result, since one only has to implement two (or three) gates,
and the constraints that one must simulate are very loose, making them easy to
implement. We formally introduce \pcircuit, and compare it to $\eps$-\gcircuit, in
\cref{sec:pcircuitdef}.

Our main result is to show that the \pcircuit problem is \ppad/-complete. It
is worth noting that there is no $\eps$ in this result, and in fact the
\pcircuit problem does not even take a parameter~$\eps$ in its definition. This is because, in
some sense, the \pcircuit problem can be viewed as a variant of $\eps$-\gcircuit in
which we have taken the limit $\eps \to 1$. We give further justification of this idea in
\cref{sec:pcircuitdef}, but at a high level, this means that there is
no loss of $\eps$ in our main hardness result, with the only losses coming when
one reduces onwards from \pcircuit. The proof of our main result is presented in \cref{sec:main-proof}, but we present a brief exposition of the main ideas in \cref{sec:intro-proof-overview}.

Finally, in \cref{sec:applications} we present a number of new large-constant hardness results for problems in \ppad/, all of which are shown via
reductions from \pcircuit. We begin by showing that $\eps$-\gcircuit is
\ppad/-hard for all $\eps < 1/10$, giving a direct strengthening of Rubinstein's
lower bound. This also implies large-constant inapproximability results for all
of the problems that currently have existential-constant lower bounds proved via
\gcircuit. However, to determine the constant, one would need to determine the
amount of $\eps$ that is lost in each of the onward reductions, and these
reductions often did not optimize this, since they were proving existential-constant lower bounds.

We argue that the way forward now is providing \emph{direct} reductions from
\pcircuit in order to get the best possible hardness results. As evidence of
this, we present the first \emph{tight} inapproximability results for additive
approximate equilibria. We show the following results.
\begin{itemize}
\item It is \ppad/-hard to compute an $\eps$-well-supported Nash equilibrium
in a two-action polymatrix game for all $\eps < 1/3$.

\item It is \ppad/-hard to compute an $\eps$-well-supported Nash equilibrium in a two-action graphical game for all $\eps < 1$. 

\item It is \ppad/-hard to compute an $\eps$-Nash equilibrium
in a two-action graphical game for all $\eps < 1/2$.

\item It is \ppad/-hard to compute an $\eps$-approximate equilibrium
in a threshold game for all $\eps < 1/6$.
\end{itemize}
Each of these results completely characterizes the computational complexity of
the respective problem: in each case we provide a polynomial-time algorithm
that can find an $\eps$-approximate solution of the problem for the values of $\eps$ not covered by the hardness result.

Furthermore, each of the hardness results are shown via a direct, and usually
relatively straightforward, reduction from \pcircuit. We view this as evidence
that \pcircuit is the ``correct'' intermediate problem for showing hardness
results within \ppad/: once one has shown that it is \ppad/-hard to solve
\pcircuit, it is then relatively easy to obtain tight lower bounds for a
variety of approximation problems within \ppad/, and we expect further such
results to be shown via \pcircuit in the future.

We summarize the hardness results obtained through a direct reduction from \pcircuit in the table below.
\begin{center}
\begin{tabular}{l|r}
Problem & Hardness Threshold \\ \hline
\gcircuit & 1/10 \\
\polymatrix & 1/3 \\
\polymatrixne & 0.088 \\
\graphical & 1 \\
\graphicalne & 1/2 \\
\threshold & 1/6
\end{tabular}
\end{center}
We note that \polymatrix and \threshold
have themselves both been used as intermediate problems for showing
other constant inapproximability results in
\ppad/~\cite{Rubinstein18-Nash-inapproximability,PapadimitriouP21-threshold-games,ChenKK21-throttling,ChenCPY22-HZ-hardness}, and thus our lower bounds potentially strengthen those results too. We provide an example in \cref{app:sec:bimatrix}, where we show that computing a relative $\eps$-WSNE in a bimatrix game with non-negative payoffs is \ppad/-complete for any $\eps \leq 1/57$, by using an improved version of a reduction from \polymatrix due to Rubinstein~\cite{Rubinstein18-Nash-inapproximability}.

\paragraph{\bf Open Questions.}
Below we identify two immediate research questions that arise from our work, which deserve further research.
\begin{itemize}
	\item What is the intractability threshold for \eps-NE in graphical games with more than two actions? We show that $1/2$ is the threshold for two-action games, and 
	we conjecture that $1/2$ is the correct answer for the multi-action case as well. Hence, we view the main open problem as finding a polynomial-time algorithm that can compute a $1/2$-NE in any graphical game. 
	We note that such an algorithm is already known for the special case of polymatrix games~\cite{DeligkasFSS17-eps-Nash-polymatrix}.
	\item What is the intractability threshold for \eps-NE and \eps-WSNE in polymatrix games? 
	For \eps-NE our understanding is far from complete, even in the two-action case, since there is a substantial gap
	between the $0.088$ lower bound and the $1/3$ upper bound.
	For \eps-WSNE, although the problem is completely resolved for the two-action case, the gap in multi-action polymatrix games is still large, and it seems that improving either the lower bound of $1/3$, or the trivial upper bound of $1$, would require significantly new ideas.
\end{itemize}

\subsection{Proof Overview for Our Main Result}\label{sec:intro-proof-overview}

We begin this proof overview by defining a very weak version of \pcircuit. An instance of the problem consists of a Boolean circuit using the standard gates \NOT, \AND, and \OR, but with the following tweak: the circuit is allowed to have \emph{cycles}. A solution to the problem is an assignment of values to each node of the circuit, so that all gates are satisfied. If we are only allowed to assign values in $\{0,1\}$ to the nodes, then it is easy to see that the problem is not a total search problem, i.e., some instances do not have a solution. For example, there is no way to assign consistent values to a cycle of three consecutive \NOT gates.

In order to ensure that the problem is total (and can thus be used to prove \ppad/-hardness results), we make the value space continuous by extending it to $[0,1]$. We extend the definition of the logical gates \NOT, \AND, and \OR to non-Boolean inputs in the most permissive way: if at least one input to the gate is not a pure bit (i.e., not in $\{0,1\}$), then the gate is allowed to output any value in $[0,1]$. The attentive reader might observe that this problem is now trivial to solve: just assign arbitrary values in $(0,1)$ to all the gates.

It is thus clear that the definition of the problem needs to be extended, by adding extra gates or by strengthening existing gates, so that the problem becomes \ppad/-hard. However, in order to discover the least amount of additional structure needed to make the problem hard, it is instructive to proceed with this definition for now, and attempt to prove hardness.

In order to prove the \ppad/-hardness of the problem, we cannot follow Rubinstein's approach, which goes through the construction of a continuous Brouwer function, because \pcircuit only offers very weak gates. Instead, we proceed via a direct reduction from the \ssperner problem, a discrete problem that is a computational version of Sperner's Lemma. The problem was shown to be \ppad/-hard by Daskalakis, Skoulakis, and
Zampetakis~\cite{DaskalakisSZ21-min-max} (who called it the \textsc{HighD-BiSperner} problem), and is the ``\ppad/-analogue'' of the \textsc{StrongTucker} problem which was recently used to prove \ppa/-hardness results~\cite{DeligkasFHM2022-CH-constant-eps}. This approach completely bypasses the continuous aspect of all such existing hardness reductions and enables us to work with the very weak gates that \pcircuit offers.

At a high level, our hardness construction works as follows: the \pcircuit
instance implements the evaluation of the \ssperner labeling on some input point
$x$ (represented in unary by multiple nodes) and then uses a feedback mechanism to ensure that the circuit is only
satisfied if $x$ is a solution to the \ssperner instance. The full reduction is
presented in \cref{sec:main-proof}, but we mention here the two main obstacles
when trying to implement this idea, and how to overcome them.
\begin{enumerate}
    \item \textbf{The input point $\boldsymbol{x}$ might not be represented by a valid bitstring.} Indeed, since the gates take values in $[0,1]$ (and values in $(0,1)$ essentially do not carry any information), there is no guarantee that the input $x$ will be represented by bits $\{0,1\}$. But then the implementation of the \ssperner labeling (which is given as a Boolean circuit) will also fail. To resolve this issue, we introduce a new gate, the \PURE gate, which, on any input, outputs two values, with the guarantee that at least one of them is a ``pure'' bit, i.e., $0$ or $1$. If the input is already a pure bit, then both outputs are guaranteed to be copies of the input. Using a binary tree of \PURE gates, we can now create many copies of $x$, such that most of them consist only of pure bits, and then use the logical gates to compute the \ssperner labeling correctly on these good copies.
    
    \item \textbf{How to implement the feedback mechanism?} Given the outputs of the \ssperner labeling at all the copies of $x$, we now need to provide some kind of feedback to $x$, so that $x$ is forced to change if it is not a solution of \ssperner. It turns out that this step can be performed if we have access to \emph{sorting}: given a list of values in $[0,1]$, sort them from smallest to largest. Unfortunately, this is impossible to achieve with the gates at our disposal, namely standard logical gates and the \PURE gate. We circumvent this obstacle by observing that: (i) it is sufficient to be able to perform some kind of ``weak sorting'' (essentially, we only care about pure bits being sorted correctly), and (ii) this weak sorting can be achieved if we make our logical gates \emph{robust}. For example, the robust version of the \AND gate outputs $0$, whenever at least one of its inputs is $0$, irrespective of whether the other input is a pure bit or not.
\end{enumerate}
With these two extensions in hand---namely, the \PURE gate and the robustness of the logical gates--- it is now possible to prove \ppad/-hardness of the problem. A very natural question to ask is: Is it really necessary to add both extensions for the problem to be hard? In \cref{app:sec:pcircuit-discussion} we show that any attempt to weaken the gate-constraints makes the problem polynomial-time solvable. In particular, the introduction of the \PURE gate is not enough by itself to make the problem \ppad/-hard; the robustness of the logical gates is also needed.

The robustness of, say, the \AND gate seems like a very natural constraint to impose. It is consistent with the meaning of the logical \AND operation, but we also observe in our applications that this ``robustness'' seems to always be automatically satisfied by the gadgets that we construct to simulate the \AND gate in our reductions from \pcircuit. On the other hand, the \PURE gate, which might look a bit unnatural or artificial at first, actually corresponds to the simplest possible version of a bit decoder, a crucial tool in all prior works. As mentioned above, we show in \cref{app:sec:pcircuit-discussion} that these are the minimal gate-constraints that are needed for the problem to be \ppad/-hard. In that sense, we argue that \pcircuit captures the essence of \ppad/-hardness: it consists of the minimal set of ingredients that are needed for a problem to be \ppad/-hard.

The attentive reader might have noticed that our gates do not distinguish between different values in $(0,1)$. For this reason, it will be more convenient to use a single symbol to denote such values in the definition of \pcircuit (\cref{sec:pcircuitdef}) and in the rest of this paper. As explained in more detail in \cref{sec:pcircuitdef}, the symbol ``$\garbo$'' will be used to denote these ``garbage'' values. In other words, the nodes of the circuit will take values in $\{0,1,\garbo\}$ instead of $[0,1]$.

\section{The \pcircuit Problem}
\label{sec:pcircuitdef}

In this section we define our new problem \pcircuit and state our main result, namely its \ppad/-completeness. Before defining \pcircuit, we begin by explaining the intuition behind its definition, and how it relates to the Generalized-Circuit (\gcircuit) problem.

\paragraph{\bf The Generalized-Circuit problem.}
In the Generalized-Circuit (\gcircuit) problem (formally defined in \cref{sec:gcircuit}) we are given a circuit and the goal is to assign a value to each node of the circuit so that each gate is computed correctly. Importantly, the circuit is a \emph{generalized} circuit, meaning that cycles are allowed. If cycles were not allowed, then it would be easy to find values satisfying all gates: just pick arbitrary values for the input gates, and then evaluate the circuit on those inputs.

Every node of \gcircuit must be assigned a value in $[0,1]$, and the gates are arithmetic gates, such as addition, subtraction, multiplication by a
constant (with output truncated to lie in $[0,1]$), and suitably defined logical gates. Reducing from \gcircuit is very useful for obtaining hardness of approximation results, because the problem remains \ppad/-hard, even when we allow some error at every gate. In the $\eps$-\gcircuit problem, the goal is to assign a value in $[0,1]$ to each node of the circuit, so that each gate is computed correctly, up to an additive error of $\pm \eps$.

The problem was first defined by Chen et al.~\cite{ChenDT09-Nash}, who proved
that it is \ppad/-hard for inverse polynomial $\eps$, and who used it to prove
\ppad/-hardness of finding Nash equilibria in bimatrix games. Prior to that, Daskalakis et al.~\cite{DaskalakisGP09-Nash} had implicitly proved that it is \ppad/-hard for inverse exponential $\eps$. Rubinstein's~\cite{Rubinstein18-Nash-inapproximability} breakthrough result proved that there exists some constant $\eps > 0$ such that $\eps$-\gcircuit remains \ppad/-hard.

\paragraph{\bf Taking the limit $\boldsymbol{\eps \to 1}$.}
In order to get strong inapproximability results, it seems necessary to prove hardness of $\eps$-\gcircuit for large, explicit, values of $\eps$. Ideally, we would like to obtain hardness for the largest possible $\eps$. While it is unclear what that value is for \gcircuit, in theory, as long as $\eps < 1$ the output of a gate still carries some information. Namely a gate whose actual output should be $0$ cannot take the value $1$.

This observation leads us to define a problem to essentially capture the setting
$\eps \to 1$. In that case, a node carries information only if its value is $0$
or $1$. Otherwise, its value is irrelevant. As a result, the natural operations
to consider in this setting are simple Boolean operations, such as \NOT, \AND,
\OR, \NAND, and \NOR. We only require these gates to output the correct result when their input is relevant, i.e., $0$ or $1$. For example, the \NOT gate should output $1$ on input $0$, and output $0$ on input $1$, but there is no constraint on its output when the input lies in $(0,1)$.

Since values in $(0,1)$ do not carry any information, and are as such interchangeable (e.g., a value $1/2$ can be replaced by $1/3$ without impacting any of the gates), we will instead use the symbol ``$\garbo$'' to denote any and all values in $(0,1)$. In other words, instead of assigning a value in $[0,1]$ to each gate, we will assign a value in $\{0,1,\garbo\}$, where $\garbo$ is interpreted as a ``garbage'' value, i.e., not corresponding to a pure bit value $0$ or $1$. With this new notation, the updated description of the \NOT gate would be that it must output $1$ on input $0$, it must output $0$ on input $1$, and it can output anything (namely, $0$, $1$, or $\garbo$) on input $\garbo$.

Unfortunately, if we only allow these logical gates, then the problem is
trivial to solve: assigning the ``garbage'' value $\garbo$ (or any value in $(0,1)$ if we use the old notation) to every node will satisfy all gates. Thus, we
need a gate that makes this impossible.

\paragraph{\bf The \PURE gate.}
To achieve this, we introduce the \PURE gate: a gate with one input and two outputs, which, intuitively, ``purifies'' its input. When fed with an actual pure bit, the \PURE gate outputs two copies of the input bit. However, when the input is not a pure bit, the gate still ensures that at least one of its two outputs is a pure bit. In more detail:
\begin{itemize}
    \item If the input is $0$, then both outputs are $0$.
    \item If the input is $1$, then both outputs are $1$.
    \item If the input is $\garbo$, then at least one of the outputs is a pure bit, i.e., $0$ or $1$.
\end{itemize}
Note that the gate is quite ``under-defined''. For example, we do not specify
which pure bit the gate should output when the input is $\garbo$, nor do we
specify the output on which this bit appears. This is actually an advantage,
because it makes it easier to reduce from the problem, since the less constrained the gates are, the easier it is to simulate them in the target application problem.

\paragraph{\bf Robustness of the logical gates.}
The introduction of the \PURE gate makes the problem non-trivial: if a \PURE gate appears in the circuit, then assigning the ``garbage'' value $\garbo$ to all nodes is no longer a solution. However, it turns out that one more modification is needed to make the problem \ppad/-hard: we have to make the logical gates \emph{robust}. For the \AND gate, this means the following: if one of its two inputs is $0$, then the output is $0$, no matter what the other input is (even if it is not a pure bit, i.e., if it is $\garbo$). Similarly, for the \OR gate we require that the output be $1$ when at least one of the two inputs is $1$. Robustness is defined analogously for \NAND and \NOR.

We show that introducing the \PURE gate and making the logical gates robust is enough to make the problem \ppad/-complete. Next, we define the problem formally and state our main result.

\paragraph{\bf Formal definition.}
In the definition below, we use the \PURE and \NOR gates, because these two gates are enough for the problem to already be \ppad/-complete. However, the problem remains hard for various other combinations of gates and restrictions on the interactions between nodes, as we detail in \cref{cor:pcircuit-gates,cor:pcircuit-restricted}. In \cref{app:sec:pcircuit-discussion} we discuss the definition in more detail, and explain why any attempt at relaxing the definition (in particular, removing the robustness) makes the problem polynomial-time solvable.

\begin{definition}[\pcircuit]\label{def:pcircuit}
An instance of \pcircuit is given by a vertex set $V=[n]$ and a set $G$ of gate-constraints (or just \emph{gates}). Each gate $g \in G$ is of the form $g = (T,u,v,w)$ where $u,v,w \in V$ are distinct nodes and $T \in \{\NOR, \PURE\}$ is the type of the gate, with the following interpretation.
\begin{itemize}
    \item If $T=\NOR$, then $u$ and $v$ are the inputs of the gate, and $w$ is its output.
    \item If $T=\PURE$, then $u$ is the input of the gate, and $v$ and $w$ are its outputs.
\end{itemize}
We require that each node is the output of exactly one gate.

A solution to instance $(V,G)$ is an assignment $\valonly: V \to \{0,1,\garbo\}$ that satisfies all the gates, i.e., for each gate $g=(T,u,v,w) \in G$ we have:
\begin{itemize}
    \item if $T=\NOR$, then $\valonly$ satisfies (left: mathematically; right: truth table)
    
    \begin{minipage}{0.45\textwidth}
    \begin{center}
        \begin{align*}
        \val{u} = \val{v} = 0 \implies \val{w} = 1\\
        (\val{u} = 1) \lor (\val{v} = 1) \implies \val{w} = 0
        \end{align*}
        \end{center}
    \end{minipage}
    \begin{minipage}{0.45\textwidth}
        \begin{center}
            \begin{tabular}{c|c||c}
                $u$ & $v$ & $w$ \\ \hline
                0 & 0 & 1 \\
                1 & $\{0,1,\garbo\}$ & 0 \\
                $\{0,1,\garbo\}$ & 1 & 0 \\
                \multicolumn{2}{c||}{Else} & $\{0,1,\garbo\}$
            \end{tabular}
        \end{center}
    \end{minipage}
    
    \item if $T=\PURE$, then $\valonly$ satisfies
    
    \begin{minipage}{0.45\textwidth}
    \begin{center}
        \begin{align*}
        \{\val{v}, \val{w}\} \cap \{0,1\} \neq \emptyset\\
        \val{u} \in \{0,1\} \implies \val{v} = \val{w} = \val{u}
        \end{align*}
        \end{center}
    \end{minipage}
    \begin{minipage}{0.45\textwidth}
    \begin{center}
        \begin{tabular}{c||c|c}
            $u$ & \phantom{xx}$v$\phantom{xx}  & $w$ \\ \hline
            $0$ & $0$ & $0$ \\
            $1$ & $1$ & $1$ \\
            \multirow{2}{*}{$\garbo$} & \multicolumn{2}{c}{At least one} \\
            & \multicolumn{2}{c}{output in $\{0,1\}$}
        \end{tabular}
        \end{center}
    \end{minipage}
    
\end{itemize}
\end{definition}

\bigskip

\noindent The following theorem is our main technical result and is proved in \cref{sec:main-proof}.

\begin{theorem}\label{thm:pcircuit-hard}
The \pcircuit problem is \ppad/-complete.
\end{theorem}

The most important part of this statement is of course the \ppad/-hardness of \pcircuit, but let us briefly discuss the other part, namely the \ppad/-membership. This is obtained as a byproduct of our results in \cref{sec:applications}, where we reduce \pcircuit to various problems that are known to lie in \ppad/. However, there is also a more direct way to prove membership in \ppad/, and in particular to establish the existence of a solution, and we briefly sketch it here. Indeed, the \pcircuit problem can be reduced to the problem of finding a Brouwer fixed point of a continuous function $F$, a problem known to lie in \ppad/~\cite{Papadimitriou94-TFNP-subclasses,EtessamiY10-FIXP}. Given an instance of \pcircuit with $n$ nodes, the function $F: [0,1]^n \to [0,1]^n$ is constructed by letting $x \in [0,1]^n$ represent an assignment of values to the $n$ nodes, and by defining $F_i(x) \in [0,1]$ as a continuous function that outputs a valid value for the $i$th node, given that the other nodes have values according to assignment $x$ (where any value in $(0,1)$ is interpreted as ``$\garbo$''). For every type of gate, it is not hard to construct a continuous piecewise-linear function $F_i$ that satisfies the constraints of that type of gate. For \PURE gates, which have two outputs, we require two such functions $F_i$ and $F_j$, which can for example be taken to be $F_i(x) = \max\{0,2x_k-1\}$ and $F_j(x) = \min\{1,2x_k\}$, where node $k$ is the input of the gate.

\begin{remark}
Note that the definition of our logical gates essentially follows Kleene's strong logic of indeterminacy~\cite{Kleene52-metamathematics}, except that undetermined outputs are not required to take value $\garbo$, but instead any value in $\{0,1,\garbo\}$ can be used. This makes it easier to argue about reductions from \pcircuit, since the gadgets implementing the gates have to enforce fewer constraints. As a result of this connection to Kleene's logic, these gates can in particular be used to implement hazard-free circuits~\cite{IkenmeyerKLLMS19-hazard-free}.
\end{remark}

\subsection{Alternative Gates and Further Restrictions}\label{sec:pcircuit-more-gates}

In this section, we present various versions of the problem that remain \ppad/-complete, in particular versions that use alternative gates and have additional restrictions.

\paragraph{\bf More gates.}
We define the following additional gates.
\begin{itemize}
    \item If $T=\COPY$ in $g=(T,u,v)$, then $\valonly$ satisfies
    
    \begin{minipage}{0.45\textwidth}
    \begin{center}
    \begin{align*}
        \val{u} = 0 \implies \val{v} = 0\\
        \val{u} = 1 \implies \val{v} = 1
    \end{align*}
        \end{center}
    \end{minipage}
    \begin{minipage}{0.45\textwidth}
        \begin{center}
            \begin{tabular}{c||c}
                $u$ & $v$ \\ \hline
                0 & 0 \\
                1 & 1 \\
                $\garbo$ & $\{0,1,\garbo\}$
            \end{tabular}
        \end{center}
    \end{minipage}

    \item If $T=\NOT$ in $g=(T,u,v)$, then $\valonly$ satisfies
    
    \begin{minipage}{0.45\textwidth}
    \begin{center}
    \begin{align*}
        \val{u} = 0 \implies \val{v} = 1\\
        \val{u} = 1 \implies \val{v} = 0
    \end{align*}
        \end{center}
    \end{minipage}
    \begin{minipage}{0.45\textwidth}
        \begin{center}
            \begin{tabular}{c||c}
                $u$ & $v$ \\ \hline
                0 & 1 \\
                1 & 0 \\
                $\garbo$ & $\{0,1,\garbo\}$
            \end{tabular}
        \end{center}
    \end{minipage}

    \item If $T=\OR$ in $g = (T,u,v,w)$, then $\valonly$ satisfies
    
    \begin{minipage}{0.45\textwidth}
    \begin{center}
    \begin{align*}
        \val{u} = \val{v} = 0 \implies \val{w} = 0\\
        (\val{u} = 1) \lor (\val{v} = 1) \implies \val{w} = 1
    \end{align*}
        \end{center}
    \end{minipage}
    \begin{minipage}{0.45\textwidth}
        \begin{center}
            \begin{tabular}{c|c||c}
                $u$ & $v$ & $w$ \\ \hline
                0 & 0 & 0 \\
                1 & $\{0,1,\garbo\}$ & 1 \\
                $\{0,1,\garbo\}$ & 1 & 1 \\
                \multicolumn{2}{c||}{Else} & $\{0,1,\garbo\}$
            \end{tabular}
        \end{center}
    \end{minipage}

    \item If $T=\AND$ in $g = (T,u,v,w)$, then $\valonly$ satisfies
    
    \begin{minipage}{0.45\textwidth}
    \begin{center}
    \begin{align*}
        \val{u} = \val{v} = 1 \implies \val{w} = 1\\
        (\val{u} = 0) \lor (\val{v} = 0) \implies \val{w} = 0
    \end{align*}
        \end{center}
    \end{minipage}
    \begin{minipage}{0.45\textwidth}
        \begin{center}
            \begin{tabular}{c|c||c}
                $u$ & $v$ & $w$ \\ \hline
                1 & 1 & 1 \\
                0 & $\{0,1,\garbo\}$ & 0 \\
                $\{0,1,\garbo\}$ & 0 & 0 \\
                \multicolumn{2}{c||}{Else} & $\{0,1,\garbo\}$
            \end{tabular}
        \end{center}
    \end{minipage}

    \item If $T=\NAND$ in $g = (T,u,v,w)$, then $\valonly$ satisfies
    
    \begin{minipage}{0.45\textwidth}
    \begin{center}
    \begin{align*}
        \val{u} = \val{v} = 1 \implies \val{w} = 0\\
        (\val{u} = 0) \lor (\val{v} = 0) \implies \val{w} = 1
    \end{align*}
        \end{center}
    \end{minipage}
    \begin{minipage}{0.45\textwidth}
        \begin{center}
            \begin{tabular}{c|c||c}
                $u$ & $v$ & $w$ \\ \hline
                1 & 1 & 0 \\
                0 & $\{0,1,\garbo\}$ & 1 \\
                $\{0,1,\garbo\}$ & 0 & 1 \\
                \multicolumn{2}{c||}{Else} & $\{0,1,\garbo\}$
            \end{tabular}
        \end{center}
    \end{minipage}
    
\end{itemize}

\bigskip

\begin{corollary}\label{cor:pcircuit-gates}
The \pcircuit problem is \ppad/-complete, for any of the following choices of gate types:
\begin{itemize}
    \item \PURE and at least one of $\{\NOR, \NAND\}$;
    \item \PURE, \NOT, and at least one of $\{\OR, \AND\}$.
\end{itemize}
\end{corollary}

\begin{proof}
This follows from \cref{thm:pcircuit-hard} by observing that a \NOR gate can always be simulated with the given set of gates. Clearly, \NOR can be simulated by first using an \OR gate and then a \NOT gate. Furthermore, \OR can be simulated by \NOT and \AND by applying De Morgan's laws. Finally, \AND can easily be obtained from \NOT and \NAND, and \NOT can be obtained from \NAND and \PURE as follows: first apply a \PURE gate, and then use its two outputs as the two inputs to a \NAND gate.
\end{proof}

\paragraph{\bf More structure.}
The hardness result is also robust with respect to restrictions applied to the \emph{interaction graph}. This graph is constructed on the vertex set $V = [n]$ by adding a directed edge from node $u$ to node $v$ whenever $v$ is the output of a gate with input $u$. For example, a \NOR gate with inputs $u,v$ and output $w$ yields the two edges $(u,w)$ and $(v,w)$. On the other hand, a \PURE gate with input $u$ and outputs $v,w$ gives the edges $(u,v)$ and $(u,w)$. Since any given node is the output of at most one gate, it immediately follows that the in-degree of every node is at most $2$. However, the out-degree of a node can \emph{a priori} be arbitrarily large. It is quite easy to show that the problem remains \ppad/-complete, even if we severely restrict the interaction graph. The proof of the following corollary can be found in \cref{sec:app:proof-pcircuit-restricted}.

\begin{corollary}\label{cor:pcircuit-restricted}
The \pcircuit problem remains \ppad/-complete, for any choice of gates $\{\PURE, \X, \Y\}$, where $(\X,\Y) \in \{\NOT\} \times \{\OR, \AND, \NOR, \NAND\}$ or $(\X,\Y) \in \{\COPY\} \times \{\NOR, \NAND\}$ and even if we also \emph{simultaneously} have all of the following restrictions.
\begin{enumerate}
    \item Every node is the input of exactly one gate.
    \item In the interaction graph, the total degree of every node is at most 3. More specifically, for every node, the in- and out-degrees, $d_{in}$ and $d_{out}$, satisfy $(d_{in},d_{out}) \in \{(1,1), (2,1), (1,2)\}$.
    \item The interaction graph is bipartite.
\end{enumerate}
\end{corollary}

\begin{remark}
Using \cref{cor:pcircuit-restricted}, it is also possible to show that \pcircuit with only two gates (namely, \PURE and one of $\{\NOR, \NAND\}$) remains \ppad/-complete even if the total degree of every node is at most 4 in the interaction graph. Indeed, \NOT gates can be implemented by first using a \PURE gate and then a \NOR/\NAND gate. The structural properties of \cref{cor:pcircuit-restricted} ensure that this yields an interaction graph where the total degree is at most 4 for each node. This can be further reduced to degree 3, if one modifies the definition of \pcircuit (\cref{def:pcircuit}) so that the two inputs to a \NOR/\NAND gate are no longer required to be two \emph{distinct} nodes $u$ and $v$, but can possibly be the same node $u=v$. However, if the definition is modified in that way, then one must be careful when reducing from \pcircuit to make sure to take into account the possibility that $u=v$ when constructing the gadget for a \NOR/\NAND gate.
\end{remark}

\section{\ppad/-completeness of \pcircuit}\label{sec:main-proof}

This section proves our main technical result, namely that \pcircuit is \ppad/-complete (\cref{thm:pcircuit-hard}). We note that membership in \ppad/ follows immediately from the reduction of the problem to \gcircuit in \cref{sec:gcircuit}. In order to establish the \ppad/-hardness, we present a polynomial-time reduction from a \ppad/-complete problem to \pcircuit. The canonical \ppad/-complete problem is the \eol problem, but, as is usually the case, we do not reduce directly from \eol, but from a problem with topological structure instead, which we introduce next.

\subsection{The \ssperner Problem}

We will reduce from the \ssperner problem, which is based on a variant of Sperner's lemma~\cite{Sperner28-lemma}. This problem is in essence the same as the \textsc{HighD-BiSperner} problem defined by Daskalakis et al.~\cite{DaskalakisSZ21-min-max} and used to prove \ppad/-hardness of a problem related to constrained min-max optimization. Furthermore, the corresponding ``strong'' variant of Tucker's lemma was used by Deligkas et al.~\cite{DeligkasFHM2022-CH-constant-eps} to provide improved \ppa/-hardness results for the consensus-halving problem in fair division.

\begin{definition}
The \ssperner problem:
\begin{description}
    \item[Input:] A Boolean circuit computing a labeling $\lambda: [\wid]^\dims \to \{-1,+1\}^\dims$ satisfying the following boundary conditions for every $i \in [\dims]$:
\begin{itemize}
    \item if $x_i = 1$, then $[\lambda(x)]_i = +1$;
    \item if $x_i = \wid$, then $[\lambda(x)]_i = -1$.
\end{itemize}
\item[Output:] Points $x^{(1)}, \dots, x^{(\dims)} \in [\wid]^\dims$ that satisfy $\|x^{(i)} - x^{(j)}\|_\infty \leq 1$ for all $i,j \in [\dims]$, and such that $\lambda(x^{(1)}), \dots, \lambda(x^{(\dims)})$ \emph{cover all labels}, i.e., for all $i \in [\dims]$ and $\ell \in \{-1,+1\}$ there exists $j \in [\dims]$ with $[\lambda(x^{(j)})]_i = \ell$.
\end{description}
\end{definition}

Note that the requirement that a solution should consist of exactly $\dims$ points is without loss of generality. If we find less than $\dims$ points that cover all labels, then we can simply re-use the same points multiple times to obtain a list of $\dims$ points that cover all labels (there is no requirement on them being distinct). If we find more than $\dims$ points that cover all labels, then it is easy to see that we can extract a subset of $\dims$ points that still cover all labels in polynomial time~\cite[Lemma~3.1]{DeligkasFHM2022-CH-constant-eps}.

\begin{theorem}[\cite{DaskalakisSZ21-min-max}]
\ssperner is \ppad/-hard, even when $\wid$ is only polynomially large (or equivalently, it is given in unary in the input).
\end{theorem}

\begin{remark}
This was proven by Daskalakis et al.~\cite{DaskalakisSZ21-min-max} by reducing from the \textsc{SuccinctBrouwer} problem, which had been proven \ppad/-hard by Rubinstein~\cite{Rubinstein16-two-player}. The \ppad/-hardness can also be proved by a more direct reduction from \eol. Indeed, \eol can be reduced to \ssperner with $\dims = 2$ and exponentially large $\wid$ by using the techniques of Chen and Deng~\cite{ChenD09-2D-Brouwer}. Then, a snake embedding technique~\cite{ChenDT09-Nash,DeligkasFHM2022-CH-constant-eps} can be used to obtain hardness for the high-dimensional version with small $\wid$, in fact, even for constant $\wid$. For our purposes, the hardness for polynomially large $\wid$ is sufficient.
\end{remark}

\subsection{Reduction from \ssperner to \pcircuit}

Consider an instance $\lambda: [\wid]^\dims \to \{-1,+1\}^\dims$ of \ssperner, where $\lambda$ is given as a Boolean circuit and $\wid$ is only polynomially large (i.e., given in unary). We will now show how to construct an instance of \pcircuit in polynomial time such that from any correct assignment to the nodes, we can extract a solution to the \ssperner instance in polynomial time. We will make use of the gates \PURE, \AND, \OR, \NOT, \COPY. All these gates can easily be simulated using the two gates \PURE and \NOR, by the arguments in the proof of \cref{cor:pcircuit-gates}.

We begin the construction of the \pcircuit instance by creating nodes $\vin{i}{1}, \dots, \vin{i}{\wid}$ for each $i \in [\dims]$. We call these nodes the \emph{original inputs}, and we think of $\vin{i}{1}, \dots, \vin{i}{\wid}$ as being a unary representation of an element in $[\wid]$. Namely, when all these nodes happen to have a pure value, i.e., $0$ or $1$, we can interpret the number of $1$'s as representing a number in $[\wid]$. (For convenience, we use the convention that if all the nodes have value $0$, then this corresponds to $1 \in [\wid]$.) Note that these unary representations are not unique, since we do not care about the order of the bits. Of course, all of this only makes sense when all these nodes are indeed assigned pure bit values. In general, this will not be the case, but this encoding will be useful in other parts of the construction. The rest of the instance can be divided into four parts: the {\em purification} stage, the {\em circuit} stage, the {\em sorting} stage, and the {\em selection} stage.

We begin with a brief overview of the purpose of each stage and how they interact with each other.
The purification stage uses the \PURE gate to create multiple ``copies'' of the original $\vin{i}{j}$ nodes, while ensuring that most of the copies have pure bit values (Purification Lemma, \cref{lem:purification}). Then, the circuit stage evaluates the circuit $\lambda$ on these copies of the original inputs. Since the purification stage ensures that most copies have pure bits, the circuit stage outputs the correct labels for most copies, and, in particular, most outputs are pure bits (Circuit Lemma, \cref{lem:circuit}). Next, for each $i \in [\dims],$ the sorting stage ``sorts'' the list of all $i$th output values computed in the circuit stage (Sorting Lemma, \cref{lem:sorting}). Since most of the $i$th output values are pure bits, the sorting stage ensures that all non-pure values are close to each other in the sorted list. The selection stage then proceeds to select $\wid$ values from the sorted list, but in a careful way, namely such that they are all far away from each other. This ensures that at most one of the $\wid$ selected values is not a pure bit. The $\wid$ selected values are then fed back into the original inputs $\vin{i}{1}, \dots, \vin{i}{\wid}$. As a result, the original input $\vin{i}{1}, \dots, \vin{i}{\wid}$ contains at most one non-pure bit, and thus the purification stage ensures that all the (pure) copies of $\vin{i}{1}, \dots, \vin{i}{\wid}$ correspond to unary numbers that differ by at most $1$. This means that these copies represent points in the \ssperner domain that are within $\ell_\infty$-distance $1$ (Selection Lemma, \cref{lem:selection}). Finally, using the boundary conditions of the \ssperner instance, we argue that these points must cover all the labels (Solution Lemma, \cref{lem:solution}).

We now describe each of the stages in more detail. We let $\cop$ denote the number of copies that we make. It will be enough to pick $\cop = 3 \dims \wid^2$. In what follows, $\valonly$ always denotes an arbitrary solution to the \pcircuit instance we construct.

\paragraph{\bf Step 1: Purification stage.}
For each $(i,j) \in [\dims] \times [\wid]$, we construct a binary tree of \PURE gates that is rooted at $\vin{i}{j}$ and has leaves $\vinc{i}{j}{1}, \dots, \vinc{i}{j}{\cop}$.

We say that $k \in [\cop]$ is a good copy, if $\val{\vinc{i}{j}{k}}$ is a pure bit for all $(i,j) \in [\dims] \times [\wid]$. We denote the set of all good copies by $\goodcopies$, i.e.,
\[\goodcopies := \left\{k \in [\cop] : \val{\vinc{i}{j}{k}} \in \{0,1\} \quad \forall (i,j) \in [\dims] \times [\wid]\right\}.\]
For a good copy $k \in \goodcopies$ and any $i \in [\dims]$, we can interpret the bitstring $(\val{\vinc{i}{j}{k}})_{j \in [\wid]} \in \{0,1\}^{\wid}$ as representing a number in $[\wid]$ in unary, which we denote by $\vincdimension{i}{k} \in [\wid]$. In other words, $\vincdimension{i}{k}$ corresponds to the number of $1$'s in the bit-string $(\val{\vinc{i}{j}{k}})_{j \in [\wid]}$. As already mentioned above, for notational convenience we let the all-zero bit string $0^{\wid}$ correspond to $1$ as well, i.e., if $\val{\vinc{i}{j}{k}} = 0$ for all $j \in [\wid]$, then $\vincdimension{i}{k} = 1$. (Alternatively, we could also have used $\wid-1$ bits instead of $\wid$.) Finally, we let $\vinctotal{k} \in [\wid]^{\dims}$ denote the vector $(\vincdimension{1}{k}, \dots, \vincdimension{\dims}{k})$.

\begin{lemma}[Purification Lemma]\label{lem:purification}
The following hold.
\begin{enumerate}
    \item There are at least $\cop - \dims \wid$ good copies, i.e., $|\goodcopies| \geq \cop - \dims \wid$.
    \item If for some $(i,j) \in [\dims] \times [\wid]$ the original input $\val{\vin{i}{j}}$ is a pure bit, then all copies have that same bit, i.e., $\val{\vinc{i}{j}{k}} = \val{\vin{i}{j}}$ for all $k \in [\cop]$.
\end{enumerate}
\end{lemma}

\begin{proof}
Observe that in a binary tree of \PURE gates, if some node has some pure value $b \in \{0,1\}$, then all nodes in the subtree rooted at this node also have value $b$. Applying this observation at the root of the tree rooted at $\vin{i}{j}$, we immediately obtain part 2 of the statement.

For part 1, note that for any $(i,j) \in [\dims] \times [\wid]$, in the binary tree of \PURE gates rooted at $\vin{i}{j}$, all leaves, except at most one, have a pure bit value. This follows from the definition of the \PURE gate and the observation about subtrees made in the previous paragraph. As a result, all values $(\val{\vinc{i}{j}{k}})_{(i,j,k) \in [\dims] \times [\wid] \times [\cop]}$ are pure bits, except for at most $\dims \wid$ of them. But this means that there are at least $\cop - \dims \wid$ values of $k \in [\cop]$ such that $\val{\vinc{i}{j}{k}} \in \{0,1\}$ for all $\forall (i,j) \in [\dims] \times [\wid]$. In other words, $|\goodcopies| \geq \cop - \dims \wid$.
\end{proof}

\paragraph{\bf Step 2: Circuit stage.}
We assume, without loss of generality, that $\lambda$ is given as a Boolean circuit $C: (\{0,1\}^{\wid})^{\dims} \to \{0,1\}^{\dims}$ using gates \AND, \OR, \NOT, and, on input $z \in (\{0,1\}^{\wid})^{\dims}$:
\begin{itemize}
    \item For each $i \in [\dims]$, the $i$th block of input bits $z_i \in \{0,1\}^{\wid}$ is interpreted by $C$ as representing a number $\overline{z_i} \in [\wid]$ in unary, in the exact same way as $\vincdimension{i}{k} \in [\wid]$ is obtained from the bitstring $(\val{\vinc{i}{j}{k}})_{j \in [\wid]}$.
    \item The circuit $C$ outputs $\lambda(\overline{z_1}, \dots, \overline{z_{\dims}}) \in \{-1,+1\}^{\dims}$, where a $-1$ output is represented by a $0$, and a $+1$ output by a $1$. To keep things simple, in the rest of this exposition we will abuse notation and think of $\lambda$ as outputting labels in $\{0,1\}^{\dims}$.
\end{itemize}
If the circuit is not originally in this form, then it can be brought in this form in polynomial time.

In the circuit stage, we construct $\cop$ separate copies of the circuit $C$, using the \AND, \OR, and \NOT gates. For each $k \in [\cop]$, the $k$th copy $C_k$ takes as input the nodes $(\vinc{i}{j}{k})_{(i,j) \in [N] \times [M]}$ and we denote its output nodes by $\vout{1}{k}, \dots, \vout{\dims}{k}$. Since the gates always have correct output when the inputs are pure bits, we immediately obtain the following.

\begin{lemma}[Circuit Lemma]\label{lem:circuit}
For all good copies $k \in \goodcopies$, the output of circuit $C_k$ is correct, i.e., $\val{\vout{i}{k}} = [\lambda(\vinctotal{k})]_i$ for all $i \in [\dims]$.
\end{lemma}

\paragraph{\bf Step 3: Sorting stage.} In this stage, for each $i \in [\dims]$, we would like to have a gadget that takes as input the list of nodes $\vout{i}{1}, \dots, \vout{i}{\cop}$ (namely, the list of $i$th outputs of the circuits $C_1, \dots, C_K$) and outputs the nodes $\vsort{i}{1}, \dots, \vsort{i}{\cop}$, such that these output nodes are a sorted list of the values of the input nodes (where we think of the values as being ordered $0 < \garbo < 1$). Unfortunately, this is not possible given the gates we have at our disposal. However, it turns out that we can do some kind of ``weak'' sorting by using the robustness of the \AND and \OR gates (i.e., the fact that \AND on input $0$ and $s$, always outputs $0$, no matter what $s \in \{0,1,\garbo\}$ is).

For now assume that we consider values in $[0,1]$ (instead of $\{0,1,\garbo\}$) and that we have access to a comparator gate that takes two inputs $s_1$ and $s_2$ and outputs $t_1$ and $t_2$, such that $t_1, t_2$ is the sorted list $s_1, s_2$. Formally, we can write this as $t_1 := \min\{s_1,s_2\}$ and $t_2 := \max\{s_1,s_2\}$. Using comparator gates, it is easy to construct a circuit that takes $\cop$ inputs and outputs them in sorted order. Indeed, we can directly implement a sorting network~\cite{Knuth98-book-vol3-sorting}, for example. Even a very naive approach will yield such a circuit of polynomial size, which is all we need. We implement this circuit with inputs $\vout{i}{1}, \dots, \vout{i}{\cop}$ and outputs $\vsort{i}{1}, \dots, \vsort{i}{\cop}$ in our \pcircuit instance, by replacing every comparator gate by \AND and \OR gates. Namely, to implement a comparator gate with inputs $s_1, s_2$ and outputs $t_1, t_2$, we use an \AND gate with inputs $s_1, s_2$ and output $t_1$, and an \OR gate with inputs $s_1, s_2$ and output $t_2$. The robustness of the \AND and \OR gates allows us to prove that this sorting gadget sorts the pure bit values correctly, in the following sense.

\begin{lemma}[Sorting Lemma]\label{lem:sorting}
Let $\cop_0$ and $\cop_1$ denote the number of zeroes and ones that the $i$th sorting gadget gets as input, i.e., $\cop_b := |\{k \in [\cop] : \val{\vout{i}{k}} = b\}|$. Then, the first $\cop_0$ outputs of the gadget are zeroes, and the last $K_1$ outputs are ones. Formally, $\val{\vsort{i}{k}} = 0$ for all $k \in [\cop_0]$, and $\val{\vsort{i}{\cop + 1 - k}} = 1$ for all $k \in [\cop_1]$.
\end{lemma}

\begin{proof}
Consider the ideal sorting circuit (that uses comparator gates) with input $f(\val{\vout{i}{1}}),\allowbreak \dots,\allowbreak f(\val{\vout{i}{\cop}})$, where $f$ maps $0$ to $0$, $1$ to $1$, and $\garbo$ to $1/2$. In other words, we imagine running the comparator circuit on our list of values, except that the ``garbage'' value $\garbo$ is replaced by $1/2$. Since the comparator circuit correctly sorts the list, its output satisfies the desired property: the first $\cop_0$ outputs are $0$, and the last $\cop_1$ outputs are $1$. Thus, in order to prove the lemma, it suffices to prove the following claim: \emph{if a node in the ideal circuit has a pure bit value $b \in \{0,1\}$, then the corresponding node in our \pcircuit instance must also have value $b$.}

We prove the claim by induction. Clearly, all input nodes satisfy the claim. Now consider some node $t_1$ that is the $\min$-output of a comparator gate with inputs $s_1$ and $s_2$, that both satisfy the claim. Recall that this gate will be implemented in the \pcircuit by an \AND gate with inputs $s_1, s_2$ and output $t_1$. If the ideal circuit assigns value $1/2$ to $t_1$, then the claim trivially holds for $t_1$. If the ideal circuit assigns value $1$ to $t_1$, then both $s_1$ and $s_2$ must have value $1$ in the ideal circuit. Since the claim holds for $s_1$ and $s_2$, they also have value $1$ in \pcircuit, and so the \AND gate will ensure that $t_1$ also has value $1$, thus satisfying the claim. Finally, if the ideal circuit assigns value $0$ to $t_1$, then it must also have assigned value $0$ to at least one of $s_1$ or $s_2$. But then, by the claim, \pcircuit also assigns value $0$ to at least one of $s_1$ or $s_2$, and the robustness of the \AND gate ensures that $t_1$ also has value $0$. The same argument also works with $\max$ and \OR instead.
\end{proof}

Note that \cref{lem:sorting} only guarantees a ``weak'' type of sorting: some parts of the output list might not be correctly ordered, and the list of output values might not be a permutation of the input values (namely, it can happen that there are more $0$'s and/or $1$'s in the output list than in the input list). However, this ``weak'' sorting will be enough for our needs as we will see below.

\paragraph{\bf Step 4: Selection stage.}
Since the list $\vsort{i}{1}, \dots, \vsort{i}{\cop}$, is now ``sorted'', we can select $\wid$ nodes from it in such a way that at most one node does not have a pure value. Indeed, this can be achieved by selecting nodes that are sufficiently far apart from each other. We thus select the nodes $(\vsort{i}{j \cdot 2 \dims \wid})_{j \in [\wid]}$ and copy their values onto the original input nodes $(\vin{i}{j})_{j \in [\wid]}$. Namely, for each $j \in [\wid]$ we introduce a \COPY gate with input $\vsort{i}{j \cdot 2 \dims \wid}$ and output $\vin{i}{j}$. Recall that $\cop = 3 \dims \wid^2 \geq \wid \cdot 2 \dims \wid$, so this is well defined.

This selection procedure ensures that two nice properties hold. First, as mentioned above, at most one of the selected nodes does not have a pure value. This is due to the fact that we select nodes that are far apart from each other in the ``sorted'' list, and all nodes with non-pure values lie close together in that list. Second, if all good copies agree that the $i$th output is, say, $1$, then all of the selected nodes will have value $1$. This is due to the fact that the ``sorted'' list will contain very few values that are not $1$, and these values will lie at the very beginning of the list. Since we do not select any nodes from the beginning of the list (the first node that is selected is at index $2NM$), we will thus only select nodes with value $1$. A similar argument also applies to the case where all good copies agree that the $i$th output is $0$. In that case, we use the fact that we do not select any nodes that lie close to the end of the list. We formalize these arguments in the proof of the following lemma.

\begin{lemma}[Selection Lemma]\label{lem:selection}
We have:
\begin{enumerate}
    \item All good copies are close to each other: $\|\vinctotal{k} - \vinctotal{k'}\|_\infty \leq 1$ for all $k, k' \in \goodcopies$.
    \item If all good copies agree that the $i$th output is $b \in \{0,1\}$, i.e., $\val{\vout{i}{k}} = b$ for all $k \in \goodcopies$, then the original input satisfies $\val{\vin{i}{j}} = b$ for all $j \in [\wid]$.
\end{enumerate}
\end{lemma}

\begin{proof}
We begin by proving part 1 of the statement. First of all, note that it suffices to show that, for all $i \in [\dims]$, at most one of the original inputs $\val{\vin{i}{1}}, \dots, \val{\vin{i}{\wid}}$ is not a pure bit. Indeed, in that case, by part 2 of the Purification Lemma (\cref{lem:purification}), it follows that for any $i \in [\dims]$ and any $k,k' \in \goodcopies$, the bitstrings $(\val{\vinc{i}{j}{k}})_{j \in [\wid]}$ and $(\val{\vinc{i}{j}{k'}})_{j \in [\wid]}$ differ in at most one bit. But this implies that $|\vincdimension{i}{k} - \vincdimension{i}{k'}| \leq 1$ for all $i \in [\dims]$, and thus $\|\vinctotal{k} - \vinctotal{k'}\|_\infty \leq 1$.

Now consider any $i \in [\dims]$. By part 1 of the Purification Lemma (\cref{lem:purification}), we have that the number of good copies $|\goodcopies| \geq \cop - \dims \wid$. By the Circuit Lemma (\cref{lem:circuit}) we know that the corresponding outputs are pure bits, i.e., $\val{\vout{i}{k}} \in \{0,1\}$ for all $k \in \goodcopies$. Thus, the list $\val{\vout{i}{1}}, \dots, \val{\vout{i}{\cop}}$ contains at least $\cop - \dims \wid$ pure bits. Using the notation from the Sorting Lemma (\cref{lem:sorting}), this means that $\cop_0 + \cop_1 \geq \cop - \dims \wid$. As a result, by applying the Sorting Lemma, it follows that in the list obtained after sorting, $\val{\vsort{i}{1}}, \dots, \val{\vsort{i}{\cop}}$, all the non-pure bits are contained in an interval of length $\dims \wid$. Since the selected nodes $(\vsort{i}{j \cdot 2 \dims \wid})_{j \in [\wid]}$ are sufficiently far apart (namely $2 \dims \wid$), at most one such node can fall in the ``bad'' interval. This means that all selected nodes, except at most one, are pure bits, and since the selected nodes are copied into the original inputs, this also holds for them.

It remains to prove part 2 of the statement. If for some $i \in [\dims]$, $\val{\vout{i}{k}} = 0$ for all $k \in \goodcopies$, then this means that $\cop_0 \geq |\goodcopies| \geq \cop - \dims \wid$. By the Sorting Lemma (\cref{lem:sorting}), it follows that $\val{\vsort{i}{k}} = 0$ for all $k \in [\cop - \dims \wid]$. In particular, since $\cop - \dims \wid = 3 \dims \wid^2 - \dims \wid \geq \wid \cdot 2 \dims \wid$, this means that for all $j \in [\wid]$, $\val{\vsort{i}{j \cdot 2 \dims \wid}} = 0$. But these are the nodes we select, and we copy their value into the original inputs, so the statement follows. The case where $\val{\vout{i}{k}} = 1$ for all $k \in \goodcopies$ is handled analogously.
\end{proof}

\paragraph{\bf Correctness.}
The description of the reduction is now complete. We have constructed a valid instance of \pcircuit in polynomial time. In particular, note that every node is the output of exactly one gate. To complete the proof, it remains to prove that from any solution of the \pcircuit instance we can extract a solution to \ssperner in polynomial time. We do this in the following final lemma.

\begin{lemma}[Solution Lemma]\label{lem:solution}
The points $\{\vinctotal{k} : k \in \goodcopies\} \subseteq [\wid]^{\dims}$ yield a solution to the \ssperner instance $\lambda$.
\end{lemma}

\begin{proof}
First of all, by part 1 of the Selection Lemma (\cref{lem:selection}), we know that the points $\vinctotal{k}$, $k \in \goodcopies$, are all within $\ell_\infty$-distance $1$. Thus, it suffices to prove that they cover all labels with respect to $\lambda$. Note that we can efficiently extract these points from a solution $\valonly$, since, for each $k \in [K]$, we can easily decide whether $k$ lies in $\goodcopies$ or not, and then extract the point $\vinctotal{k}$.

We prove by contradiction that the points $\vinctotal{k}$, $k \in \goodcopies$, must cover all labels. Assume, on the contrary, that there exists $i \in [\dims]$ and $b \in \{0,1\}$ such that $[\lambda(\vinctotal{k})]_i = b$ for all $k \in \goodcopies$. Then, by the Circuit Lemma (\cref{lem:circuit}), all the corresponding circuits must output $b$, i.e., $\val{\vout{i}{k}} = b$ for all $k \in \goodcopies$. By part 2 of the Selection Lemma (\cref{lem:selection}), it follows that the original input satisfies $\val{\vin{i}{j}} = b$ for all $j \in [\wid]$. Finally, by part 2 of the Purification Lemma (\cref{lem:purification}), it must be that for all copies $k \in \goodcopies$, all $\wid$ bits $\val{\vinc{i}{1}{k}}, \dots, \val{\vinc{i}{\wid}{k}}$ are equal to $b$. Now, if $b=1$, then this means that $\vincdimension{i}{k} = \wid$, and thus $[\lambda(\vinctotal{k})]_i = 0 \neq b$ by the \ssperner boundary conditions, which contradicts the original assumption. Similarly, if $b=0$, then $\vincdimension{i}{k} = 1$, and thus $[\lambda(\vinctotal{k})]_i = 1 \neq b$ by the \ssperner boundary conditions, which is again a contradiction. (We recall here that we have renamed the labels $\{-1,+1\}$ to $\{0,1\}$, respectively, for the purpose of this proof.)
\end{proof}

\section{Applications}
\label{sec:applications}

In this section, we derive strong inapproximability lower
bounds for \ppad/-complete problems, by reducing from \pcircuit.

\subsection{Generalized Circuit}\label{sec:gcircuit}

The $\eps$-\gcircuit problem was introduced by Chen, Deng, and
Teng~\cite{ChenDT09-Nash}.  In this section, we show that $\eps$-\gcircuit is
\ppad/-hard for all $\eps < 1/10$.

\paragraph{\bf The $\boldsymbol{\eps}$-\gcircuit problem.}
The problem is defined as follows.

\begin{definition}[Generalized Circuit~\cite{ChenDT09-Nash}]
A generalized circuit is a tuple $(V, T)$, where $V$ is a set of nodes, and $T$
is a set of gates. Each gate $t \in T$ is a five-tuple $(G, u, v, w, c)$,
where $G$ is a gate type from the set $\{G_c, G_{\times c}, G_=, G_+, G_-, G_<,
G_\lor, G_\land, G_\lnot\}$, $u, v \in V \cup \{\textsf{nil}\}$ are input
variables, $w \in V$ is an output variable, and $c \in [0, 1]
\cup \{\textsf{nil}\}$ is a rational constant.

The following requirements must be satisfied for each gate 
$(G, u, v, w, c) \in T$.
\begin{itemize}
\item 
$G_c$ gates take no input variables and use a constant in $[0, 1]$. So $u = v = \textsf{nil}$
and $c \in [0, 1]$ whenever $G = G_c$.

\item

$G_{\times c}$ gates take one input variable and a constant. So $u \in V$, $v =
\textsf{nil}$, and $c \in [0, 1]$ whenever $G = G_{\times c}$. 

\item 
$G_{=}$ and $G_{\lnot}$ gates take one input variable and do not use a constant. So $u \in V$,
$v = c = \textsf{nil}$, whenever $G \in \{G_{=}, G_{\lnot}\}$. 

\item All other gates take two input variables and do not use a constant. So 
$u \in
V$, $v \in V$, and $c = \textsf{nil}$ whenever $G \notin \{G_c, G_{\times c},
G_{=}, G_{\lnot}\}$. 

\item Every variable in $V$ is the output variable for exactly one gate. More
formally, for each variable $w \in V$, there is exactly one gate $t \in T$ such
that $t = (G, u, v, w, c)$.  
\end{itemize}
\end{definition}

The $\eps$-\gcircuit problem is defined as follows. Given a generalized circuit
$(V, T)$, find a vector $\valonly \in [0, 1]^{|V|}$ such that for each gate 
in $T$ the following constraints are satisfied.

\begin{center}
\begin{tabular}{l|l}
Gate & Constraint \\ \hline
$(G_c, \textsf{nil}, \textsf{nil}, w, c)$ & $\val{w} = c \pm \eps$ \\
$(G_{\times c}, u, \textsf{nil}, w, c)$ & $\val{w} = \val{u} \cdot c \pm \eps$ \\
$(G_{=}, u, \textsf{nil}, w, \textsf{nil})$ & $\val{w} = \val{u} \pm \eps$ \\
$(G_{+}, u, v, w, \textsf{nil})$ & $\val{w} = \min(\val{u} + \val{v}, 1) \pm \eps$ \\
$(G_{-}, u, v, w, \textsf{nil})$ & $\val{w} = \max(\val{u} - \val{v}, 0) \pm \eps$ \\
$(G_{<}, u, v, w, \textsf{nil})$ & $\val{w} = 
    \begin{cases} 
    1 \pm \eps & \text{if } \val{u} < \val{v} - \eps \\ 
    0 \pm \eps & \text{if } \val{u} > \val{v} + \eps 
    \end{cases}$ \\
$(G_{\lor}, u, v, w, \textsf{nil})$ & $\val{w} = 
    \begin{cases} 
    1 \pm \eps & \text{if } \val{u} \geq 1 - \eps \text{ or } \val{v} \geq 1 - \eps \\ 
    0 \pm \eps & \text{if } \val{u} \leq \eps \text{ and } \val{v} \leq \eps
    \end{cases}$ \\
$(G_{\land}, u, v, w, \textsf{nil})$ & $\val{w} = 
    \begin{cases} 
    1 \pm \eps & \text{if } \val{u} \geq 1 - \eps \text{ and } \val{v} \geq 1 - \eps \\ 
    0 \pm \eps & \text{if } \val{u} \leq \eps \text{ or } \val{v} \leq \eps
    \end{cases}$ \\
$(G_{\lnot}, u, \textsf{nil}, w, \textsf{nil})$ & $\val{w} = 
    \begin{cases} 
    1 \pm \eps & \text{if } \val{u} \leq \eps \\
    0 \pm \eps & \text{if } \val{u} \geq 1 - \eps
    \end{cases}$ \\
\end{tabular}
\end{center}
Here the notation $a = b \pm \eps$ is used as a shorthand for $a \in [b - \eps, b + \eps]$.
We will also make use of gates of type $(G_{>}, u, v, w, \textsf{nil})$ which enforce the constraint
\begin{equation*}
\val{w} = \begin{cases} 
    0 \pm \eps & \text{if } \val{u} < \val{v} - \eps \\ 
    1 \pm \eps & \text{if } \val{u} > \val{v} + \eps 
    \end{cases}
\end{equation*}
Gates of type $G_{>}$ can be easily built by using a $G_{\lnot}$ gate to negate
the output of a $G_{<}$ gate.

In the remainder of this section, we prove the following result.

\begin{theorem}
\label{thm:gcircuithard}
\eps-\gcircuit is \ppad/-hard for every $\eps < 1/10$.
\end{theorem}

\begin{proof}
We will reduce from the \pcircuit problem that uses the gates \NOR and
\PURE, which we showed to be \ppad/-hard in \cref{thm:pcircuit-hard}. 
We will encode 0 values in the \pcircuit problem as values in the range
$[0, \eps]$ in \gcircuit, while 1 values will be encoded as values in the range
$[1 - \eps, 1]$. Then, each gate from
\pcircuit will be simulated by a combination of gates in the \gcircuit
instance.

\paragraph{\bf \NOR gates.}
A \NOR gate $(\NOR, u, v, w)$ will be simulated by \gcircuit gates that compute
$$\val{u} + \val{v} < 5/9,$$ 
which requires us to use $G_+$, $G_<$ and $G_c$. We claim that this gate works
for any $\eps < 1/9$. 

The idea is that if both inputs lie in the range $[0, \eps]$, and thus both inputs
encode zeros in \pcircuit, we will have 
$\val{u} + \val{v} \le 3 \eps < 3/9$, where the extra \eps error is introduced by $G_+$. 
The $G_c$ gate outputting 
$5/9$ may output a value as small as $4/9$ after the error is taken into
account. Thus, the $G_<$ gate will compare these two values, and provided that
$\eps < 1/9$, the comparison will succeed, so the gate will output a value
greater than or equal to $1 - \eps$, which corresponds to a 1 in the \pcircuit
instance, as required.

If, on the other hand, at least one input is in the range $[1-\eps, 1]$, then we
will have $\val{u} + \val{v} \ge 1 - 2\eps > 7/9$, where again $G_+$ introduces an extra
\eps error. The $G_c$ gate outputting $5/9$ may output a value as large as $6/9$
after the error is taken into account. Thus, the $G_<$ will compare these two
values and provided that $\eps < 1/9$, the comparison will succeed, so the
gate will output a value less than or equal to $\eps$, which corresponds to a 0
in the \pcircuit instance, which is again as required.

\paragraph{\bf \PURE gates.}
A \PURE gate $(\PURE, u, v, w)$ will be simulated by two \gcircuit gadgets. The
first will set $\val{v}$ equal to $\val{u} > 3/10$, while the second will set
$\val{w}$ equal
to $\val{u} > 7/10$, where $G_c$ and $G_>$ gates are used to implement these
operations. We claim that this construction works for any $\eps < 1/10$. 

We begin by considering the first gadget, which sets $\val{v}$ equal to
$\val{u} > 3/10$. 
If $\val{u} \le \eps < 1/10$, then note that the $G_c$ gate outputting 3/10 may output
a value as small as $2/10$ once errors are taken into account. Thus, since $\eps
< 1/10$, the comparison made by the
$G_>$ gate will succeed, and so $\val{v}$ will be set to a value less than or
equal to
$\eps$. On the other hand, if $\val{u} \ge 1/2$, then note that the 
$G_c$ gate outputting 3/10 may output
a value as large as $4/10$ once errors are taken into account. So, the 
comparison made by the
$G_>$ gate will again succeed, and $\val{v}$ will be set to a value greater
than or equal to $1 - \eps$.

One can repeat the analysis above for the second gadget to conclude that
$\val{w}$
will be set to a value less than or equal to $\eps$ when $\val{u} \le 1/2$, and a
value greater than or equal to $1-\eps$ when $\val{u} \ge 1 - \eps$. So we can
verify that the conditions of the \PURE gate are correctly simulated.

\begin{itemize}
\item If $\val{u} \le \eps$, meaning that the input encodes a zero, then
$\val{v}$ and
$\val{w}$ will be set to values that are less than or equal to $\eps$, and so both
outputs encode zeros.

\item If $\val{u} \ge 1- \eps$, meaning that the input encodes a one, then $\val{v}$ and
$\val{w}$ will be set to values that are greater than or equal to $1 -\eps$, and so
both outputs encode ones.

\item No matter what value $\val{u}$ takes, at least one output will be set to a value
that encodes a zero or a one. Specifically, if $\val{u} \le 1/2$, then the second
comparison gate will set $\val{w} \le \eps$, while if $\val{u} \ge 1/2$, then the first
comparison gate will set $\val{v} \ge 1 - \eps$.  
\end{itemize}
Thus, the construction correctly simulates a \PURE gate. Note that $\eps$ cannot
be increased beyond $1/10$ in this construction, since then there would be no guarantee that an
encoding of a zero or a one would be produced when $\val{u} = 1/2$.

\paragraph{\bf The lower bound.}
Given a \pcircuit instance defined over variables $V$,  we produce a \gcircuit
instance by replacing each gate in the \pcircuit with the constructions given
above. Then, given a solution $\valonly$ to the $1/10$-\gcircuit instance, we can produce
a solution $\valtwoonly$ to \pcircuit in the following way. For each $v \in V$
\begin{itemize}
\item if $\val{v} \le \eps$, then we set $\valtwo{v} = 0$,
\item if $\val{v} \ge 1 - \eps$, then we set $\valtwo{v} = 1$, and
\item if $\val{v} > \eps$ and $\val{v} < 1 - \eps$, then we set $\valtwo{v} =
\garbo$. 
\end{itemize}
By the arguments given above, we have that $\valtwoonly$ is indeed a solution to
\pcircuit, and thus \cref{thm:gcircuithard} is proved.
\end{proof}

\subsection{Definitions for Polymatrix and Graphical Games}
\label{sec:gamedefs}

In Sections~\ref{sec:polymatrix} and~\ref{sec:graphical} we give results for
polymatrix games and graphical games, respectively. These two types of game are
defined in very similar ways, so in this section we give formal definitions for
both polymatrix and graphical games. 

\paragraph{\bf Graphical games.}

An $n$-player $m$-action \emph{graphical game} is defined by a directed graph $G = (V,E)$, where
$|V|=n$, each node of $G$ corresponds to a player, and the maximum number of actions per player is $m$.
We define the in-neighbourhood of node $i$ to be $\innei{i} := \{ j \in V  :  (j,i) \in E \}$, and similarly its out-neighbourhood to be $\outnei{i} := \{ j \in V : (i,j) \in E \}$. We also define the neighbourhood of $i$ as $\nei{i} := \innei{i} \cup \outnei{i} \cup \{i\}$. We include $i$ in its own neighbourhood for notational convenience. 

In a graphical game, each player $i$ participates in a normal form game $G_i$ whose player set is $\nei{i}$, but she affects only the payoffs of her out-neighbours. Player $i$ has $m_i$ \emph{actions}, or \emph{pure strategies}, and her payoffs are represented by a function $R_{i} : [m_i] \times \prod_{j \in \innei{i}} [m_j] \mapsto [0,1]$ which will be referred to as the \emph{payoff tensor} of $i$. If the codomain of $R_i$ is $\{0,1\}$ for all $i \in V$, we have a {\em win-lose} graphical game.

\paragraph{\bf Polymatrix games.}

\emph{Polymatrix games} are
defined by an undirected graph $(V, E)$, where each vertex represents a player,
and we use $n = |V|$ to denote the number of players in the game. Like
graphical games, each player $i \in V$ has a fixed number of actions given by $m_i$. 

The difference between polymatrix games and graphical games comes in how the
payoffs for each player are defined. In a polymatrix game, for each edge $(i,
j) \in E$, there is an $m_i \times m_j$ matrix $A_{ij}$ giving the payoffs that
player $i$ obtains from their interaction with player $j$, and likewise there
is an $m_j \times m_i$ matrix $A_{ji}$ giving payoffs for player $j$'s
interaction with player $i$.

\paragraph{\bf Payoffs.}

From now on, our definitions will apply to both graphical and polymatrix games.
A \emph{mixed strategy} $s_i$ for player $i$ specifies a probability distribution over
player $i$'s actions. Thus, the set of mixed strategies for player $i$ corresponds to the $(m_i-1)$-dimensional simplex $\Delta^{m_i-1}$. 
The \emph{support}
of a mixed strategy $s_i = (s_{i}(1), s_{i}(2), \dots, s_{i}(m_i)) \in \Delta^{m_i-1}$ is given by $\supp(s_i) = \{j \in [m_i] \; : \; s_i(j) > 0\}$. In other words, it is the set of pure strategies that are played with non-zero probability in strategy $s_i$.

An action profile $\vba(H) := (a_i)_{i \in H}$ over a player set $H$ is a tuple of actions, one for each player in $H$, and so the set of these action profiles is given by $A(H) = \prod_{i \in H} [m_i]$.
Similarly, a \emph{strategy profile} $\vbs(H)$ over the same set is a tuple of mixed strategies, and so the set of these strategy profiles is given by $\prod_{i \in H} \Delta^{m_i-1}$.
We define the \emph{partial action profile} \vbai to be the tuple of all players' actions except $i$'s action, and similarly we define the \emph{partial strategy profile} \vbsi. 

In a graphical game, the expected payoff of $i$ when she plays action $k \in
[m_i]$, and all other players play according to $\vbsi$ is 
\begin{equation*}
	u_i(k, \vbsi) := \sum_{\vba \in A(\innei{i})} R_{i}(k; \vba) \cdot \prod_{{j \in \innei{i}}} s_j(a_j).
\end{equation*}
Notice that this depends only on the in-neighbours of $i$.
The expected payoff to player $i$ under $\vbs$ is therefore
\begin{equation*}
	u_i(\vbs) := \sum_{k \in [m_i]} u_i(k, \vbsi) \cdot s_i(k).
\end{equation*}

In contrast, in a polymatrix game, for each strategy profile $\vbs = (s_1, s_2,
\dots, s_n) \in \Sigma$, the payoff to player $i$ is
\begin{equation*}
u_i(\vbs) := s_i^{T} \cdot \sum_{j \; :\; (i, j) \in E} A_{ij} \cdot s_j.
\end{equation*}
In other words, the payoff obtained by player $i$ in a polymatrix game is the \emph{sum} of the payoffs obtained from the interaction of $i$ with every neighbouring player $j$.

All further definitions will apply to both graphical and polymatrix games,
with the respective definitions of $u_i$ given above. 

\paragraph{\bf Best responses.}

A pure strategy $k$ is a \emph{best response} for player $i$ against a partial strategy profile \vbsi
if it achieves the maximum payoff over all her pure strategies, that is,
\begin{equation*}
	u_i(k, \vbsi) = \max_{\ell \in [m_i]} u_i(\ell, \vbsi).
\end{equation*}
Pure strategy $k$ is an \emph{$\eps$-best response} if the payoff it yields is within $\eps$ of a
best response, meaning that 
\begin{equation*}
	u_i(k, \vbsi) \ge \max_{\ell \in [m_i]} u_i(\ell, \vbsi) - \eps.
\end{equation*}
Finally, the \emph{best response payoff} for player $i$ is 
\begin{equation*}
	\br_i(\vbs_{-i}) := \max_{\ell \in [m_i]} u_i(\ell, \vbsi).
\end{equation*}

\paragraph{\bf Approximate Nash equilibria.}

A strategy profile \vbs is a \emph{Nash equilibrium} if every player achieves 
their best response payoff, meaning that $\br_i(\vbsi) = u_i(\vbs)$ for all
players $i$. The approximation notions that will be of interest to us for
graphical games and polymatrix games are defined as follows.

\begin{definition}[\eps-NE]

A strategy profile \vbs is an \emph{$\eps$-Nash equilibrium} (\eps-NE) if every
player's payoff is within $\eps$ of their best response payoff. Formally, 
    \begin{align*}
        \forall i \in [n], \quad u_i(\vbs) \ge \br_i(\vbsi) - \eps .
    \end{align*}
\end{definition}

\begin{definition}[\eps-WSNE]
A strategy profile \vbs is an \emph{$\eps$-well-supported Nash equilibrium}
(\eps-WSNE) if every
player only plays strategies that are $\eps$-best responses, meaning that
for all $i$ we have that $\supp(s_i)$ contains only 
$\eps$-best response strategies. Formally, 
    \begin{align*}
        \forall i \in [n], \forall k \in \supp(s_i), \quad u_i(k, \vbsi) \geq \br_i(\vbsi) - \eps.
    \end{align*}
\end{definition}

\paragraph{\bf Payoff normalization in polymatrix games.}

Note that $\eps$-NE and \wn{\eps} are both additive notions of
approximation, and so values of $\eps$ cannot normally be compared across
games, since doubling all payoffs in the game would double the value of $\eps$,
for example. To deal with this, it is common to normalize the payoffs of the
game so that the maximum possible payoff is~1 and the minimum possible payoff is~0, which allows us to meaningfully compare values of~$\eps$ across
games.

We specifically need to normalize payoffs in polymatrix games. 
Payoff normalization is not needed for graphical games, because in those games all payoffs
came directly from the payoff tensors, and the tensors always contain values in
$[0, 1]$. In a polymatrix game
the overall payoff to a player is the sum of individual payoffs from different
matrices, which necessitates normalization.

We adopt the normalization scheme for polymatrix games given by Deligkas et al.\
\cite{DeligkasFSS17-eps-Nash-polymatrix}, which is a particularly restrictive scheme, but doing so will
allow our lower bounds to be compared to the $1/2 + \del$ upper bound 
for finding $\eps$-NEs given in that paper.

For each player $i$ we compute
\begin{align*}
U_i &= \max_{1 \le p \le m_i} \left( \sum_{j \; :\; (i, j) \in E} \max_{1 \le q \le m_j}
(A_{ij})_{(p, q)} \right) \\
L_i &= \min_{1 \le p \le m_i} \left( \sum_{j \; :\; (i, j) \in E} \min_{1 \le q \le m_j}
(A_{ij})_{(p, q)} \right). 
\end{align*}
Then, if $d(i)$ denotes the degree of player $i$ in $(V, E)$, we apply the following transformation to each payoff $z$ in each payoff matrix $A_{ij}$
\begin{equation*}
T(z) = \frac{1}{U_i - L_i} \cdot \left(z - \frac{L_i}{d(i)}\right).
\end{equation*}
If $U_i = L_i$, then we simply set $T(z) = 0$.
This ensures that each player's maximum possible payoff is $1$, and their minimum possible payoff is $0$. 

In our hardness reductions, we will directly build games that are
already normalized, meaning that $U_i = 1$ and $L_i = 0$ for all players $i$, so no extra normalization step is required.

\subsection{Polymatrix Games}
\label{sec:polymatrix}

In this section, we show a number of results for polymatrix games (defined in
Section~\ref{sec:gamedefs}). We first provide a simple algorithm that computes a $1/3$-well-supported Nash equilibrium (WSNE) in polynomial time when there are two actions per player (\cref{thm:polymatrix-algo}). We then show that this algorithm is in fact optimal, by proving hardness of finding $\eps$-WSNE for all $\eps < 1/3$ (\cref{thm: polymatrix wsne}). Next, we also prove hardness of computing $\eps$-Nash equilibria for all $\eps <  2\sqrt{73} - 17 \approx 0.088$ (\cref{thm: polymatrix ne}). Both of these hardness results hold for degree 3 bipartite games with at most two strategies per player. Finally, we show that
the hardness result for $\eps$-WSNE also holds for win-lose polymatrix games in
degree 7 bipartite games with at most two strategies per player (\cref{thm:win-lose-polymatrix}).

\subsubsection{A Simple Algorithm for \texorpdfstring{$\boldsymbol{1/3}$}{1/3}-WSNE in Two-Action Polymatrix Games}
\label{sec:WSNE-polymatrix-algo}

In this section we present a simple polynomial-time algorithm for computing a $1/3$-WSNE in a two-action polymatrix game. We begin with a description of the algorithm. The algorithm proceeds in two steps:

\paragraph{\bf Step 1: Eliminating players who don't need to mix.}
In the first step of the algorithm, we check if there exists some player $i$ such that one of its two actions is always a $1/3$-best response, no matter what actions the other players pick. More formally, we check if there exists a player $i$ and an action $k \in \{0,1\}$ such that $u_i(k,\vbsi) \geq \max\{u_i(0,\vbsi), u_i(1,\vbsi)\} - 1/3$ for all strategy profiles $\vbs$. If we find such a player $i$, then we fix their strategy to always play action $k$, and we remove them from the game, while also updating the payoffs of the other players to reflect the fact that player $i$ always plays action $k$. Note that no matter how we fix the strategies of the remaining players later in the algorithm, player $i$ is guaranteed to play a $1/3$-best response action in the original game.

After removing player $i$ from the game, and updating the payoffs, we again check if there exists some player $j$ with a $1/3$-best response action, and if so, fix their strategy and remove them from the game, as above. After at most $n$ iterations, we are left with a game where none of the remaining players has an action that is always a $1/3$-best response.

\paragraph{\bf Step 2: Letting the remaining players mix uniformly.}
In the second step of the algorithm, we simply fix the strategies of the remaining players so that they mix uniformly between their two actions, i.e., they play action $0$ with probability $1/2$, and action $1$ with probability $1/2$.

\bigskip

We note that the simple approach used by the algorithm is essentially the same as an algorithm given by Liu et al.~\cite{LiuLD21-sparse-winlose-polymatrix} for computation of exact equilibria in very sparse win-lose polymatrix games. We now prove the following.

\begin{theorem}\label{thm:polymatrix-algo}
The algorithm computes a $1/3$-WSNE in polynomial time in two-action polymatrix games.
\end{theorem}

\begin{proof}
The algorithm clearly runs in polynomial time. To prove its correctness, it suffices to show that for the remaining players (i.e., those that were not removed during Step 1) both their actions are $1/3$-best-responses, when all the other remaining players mix uniformly.

Consider some player $i$ that was not eliminated in the first step of the algorithm. Define the function $\phi: \Sigma_{-i} \to \mathbb{R}$ by letting $\phi(\vbsi) = u_i(0,\vbsi) - u_i(1,\vbsi)$, where $\Sigma_{-i} := (\Delta^1)^{n-1} \equiv [0,1]^{n-1}$, i.e., $\vbsi$ is a partial strategy profile (without the strategy of $i$). The crucial observation is that $\phi$ is a linear function, by definition of the utility function in a polymatrix game. Furthermore, by the payoff normalization it holds that $\phi(\vbsi) \in [-1,1]$ for all $\vbsi \in \Sigma_{-i}$, since $u_i(0,\vbsi), u_i(1,\vbsi) \in [0,1]$.

Since player $i$ was not eliminated in the first step of the algorithm, it must be that none of the two actions $0$ and $1$ is always a $1/3$-best response. In other words, the function $\phi$ satisfies $\max_{\vbsi} \phi(\vbsi) > 1/3$ and $\min_{\vbsi} \phi(\vbsi) < -1/3$. Now, since $\phi$ is a linear function, the maximum and minimum are achieved at two opposite corners of the domain $\Sigma_{-i} = (\Delta^1)^{n-1} \equiv [0,1]^{n-1}$. Denote these opposite corners $\ell, h \in \Sigma_{-i}$, such that $\phi(\ell) < -1/3$ and $\phi(h) > 1/3$.

Since $\ell$ and $h$ are opposite corners of the domain $\Sigma_{-i}$, it follows that $\ell/2 + h/2$ corresponds to the partial strategy profile where all players mix uniformly. As a result, it remains to show that $\phi(\ell/2+h/2) \in [-1/3,1/3]$, which exactly means that both actions $0$ and $1$ are $1/3$-best responses for player $i$, when all other players mix uniformly. By linearity of $\phi$ we can write
\begin{equation*}
\phi(\ell/2+h/2) = \phi(\ell)/2 + \phi(h)/2 \geq -1/2 +(1/3)/2 = -1/3
\end{equation*}
since $\phi(\ell) \geq -1$ and $\phi(h) > 1/3$. The inequality $\phi(\ell/2+h/2) \leq 1/3$ is similarly obtained by using $\phi(\ell) < -1/3$ and $\phi(h) \leq 1$. We have thus shown that $\phi(\ell/2+h/2) \in [-1/3,1/3]$ as desired.
\end{proof}

\subsubsection{Hardness for \texorpdfstring{$\boldsymbol{\eps}$}{ε}-WSNE}
\label{sec:WSNE-polymatrix}

In this section we prove that computing an $\eps$-WSNE is \ppad/-hard for any $\eps < 1/3$, even in two-action polymatrix games. In particular, this shows that the simple algorithm presented in the previous section is optimal.

\begin{theorem}\label{thm: polymatrix wsne}
It is \ppad/-hard to find an $\eps$-WSNE in a polymatrix game for all $\eps <
1/3$, even in degree 3 bipartite games with two strategies per player.
\end{theorem}

\begin{proof}
We reduce from the variant of \pcircuit that uses \NOT, \AND, and \PURE gates,
and we will specifically reduce from the hardness result given in \cref{cor:pcircuit-restricted}, since it will give us extra
properties for the hard polymatrix games.

Given an instance $(V, G)$ of \pcircuit, we will build a polymatrix game with
vertex set $V$, meaning each player in the game will simulate a variable from
\pcircuit. 
Each player $i$ will have exactly two strategies, which we will denote as
$\zero$ and $\one$. The strategy $s_i$ of player $i$ will encode the value
of the variable in the following way.
\begin{itemize}
\item If $s_i$ places all probability on $\zero$, then this corresponds to
setting variable $i$ to 0.
\item If $s_i$ places all probability on $\one$, then this corresponds to 
setting variable $i$ to 1.
\item If $s_i$ is a strictly mixed strategy, then this corresponds to setting
variable $i$ to $\garbo$.
\end{itemize}

The edges of the game will simulate the gates. The definition of \pcircuit
requires that each variable $v$ is the output of exactly one gate $g$.
An important property of our reduction is that the player representing $v$ only
receives non-zero payoffs from the inputs to the gate $g$, and receives zero
payoffs from any other gates that use $v$ as an input. 
This means that we can reason about the equilibrium condition of each of the
gates independently,
without worrying about where the values computed by those gates are used.

\paragraph{\bf \NOT gates.}

For a gate $g = (\NOT, u, v)$, the player $v$, who represents the output
variable, will play the
following bimatrix game against $u$, who represents the input variable. 

\begin{center}
\bimatrixgame{3.7mm}{2}{2}{$v$}{$u$}%
{{$\zero$}{$\one$}}%
{{$\zero$}{$\one$}}%
{
\payoffpairs{1}{{$0$}{$1$}}{{$0$}{$0$}}
\payoffpairs{2}{{$1$}{$0$}}{{$0$}{$0$}}
}
\end{center}
This depiction of a bimatrix game shows the matrices $A_{vu}$ and $A_{uv}$, with
the payoffs for $v$ being shown in the bottom-left of each cell, and the
payoffs for $u$ being shown in the top-right. 

We claim that this gate will work for all $\eps < 1$.
\begin{itemize}
\item If $u$ plays $\zero$ as a pure strategy, then $\one$ gives payoff $1$ to
$v$
and $\zero$ gives payoff $0$ to $v$. So in any $\eps$-WSNE with $\eps < 1$, only
$\one$ can be an $\eps$-best response for $v$, meaning that $v$ must play $\one$ as a pure
strategy, as required by the constraints of the \NOT gate.

\item Using identical reasoning, if 
$u$ plays $\one$ as a pure strategy, then $v$ must play $\zero$ as a pure
strategy in any $\eps$-WSNE.

\item If $u$ plays a strictly mixed strategy, then the \NOT gate places no
constraints on the output, so we need not prove anything about $v$'s strategy.
\end{itemize}

\paragraph{\bf \AND gates.}

For a gate $g = (\AND, u, v, w)$, the player $w$, who represents the output
variable, will play the
following matrix games against $u$ and $v$, who represent the input variables.
\begin{center}
\bimatrixgame{3.7mm}{2}{2}{$w$}{$u$}%
{{$\zero$}{$\one$}}%
{{$\zero$}{$\one$}}%
{
\payoffpairs{1}{{$\frac{1}{2}$}{$0$}}{{$0$}{$0$}}
\payoffpairs{2}{{$0$}{$\frac{1}{6}$}}{{$0$}{$0$}}
}
\hskip 1.5cm
\bimatrixgame{3.7mm}{2}{2}{$w$}{$v$}%
{{$\zero$}{$\one$}}%
{{$\zero$}{$\one$}}%
{
\payoffpairs{1}{{$\frac{1}{2}$}{$0$}}{{$0$}{$0$}}
\payoffpairs{2}{{$0$}{$\frac{1}{6}$}}{{$0$}{$0$}}
}
\end{center}
We claim that this gate will work for all $\eps < \frac{1}{3}$.

\begin{itemize}
\item If both $u$ and $v$ play $\one$ as a pure strategy, then the payoff to
$w$ for playing $\zero$ is $0$, and the payoff to $w$ for playing $\one$ is $1/6
+ 1/6 = 1/3$. So, if $\eps < \frac{1}{3}$, then in all $\eps$-WSNEs the only
$\eps$-best response for $w$ is $\one$, as required by the constraints
of the \AND gate.

\item If at least one of $u$ and $v$ play $\zero$ as a pure strategy, then the
payoff to $w$ for playing $\zero$ is at least $1/2$, while the payoff to $w$ for
playing $\one$ is at most $1/6$. So, if $\eps < \frac{1}{3}$, then in all
$\eps$-WSNEs the only $\eps$-best response for $w$ is $\zero$, as
required by the constraints of the \AND gate.

\item The \AND gate places no other constraints on the variable $w$, so we can
ignore all other cases, e.g., the case where both $u$ and $v$ play strictly mixed
strategies.
\end{itemize}

\paragraph{\bf \PURE gates.}

For a gate $g = (\PURE, u, v, w)$, the players $v$ and $w$, who represent the
output variables, play the following games against $u$, who represents the input
variable.
\begin{center}
\bimatrixgame{3.7mm}{2}{2}{$v$}{$u$}%
{{$\zero$}{$\one$}}%
{{$\zero$}{$\one$}}%
{
\payoffpairs{1}{{$\frac{1}{3}$}{$0$}}{{$0$}{$0$}}
\payoffpairs{2}{{$0$}{$1$}}{{$0$}{$0$}}
}
\hskip 1.5cm
\bimatrixgame{3.7mm}{2}{2}{$w$}{$u$}%
{{$\zero$}{$\one$}}%
{{$\zero$}{$\one$}}%
{
\payoffpairs{1}{{$1$}{$0$}}{{$0$}{$0$}}
\payoffpairs{2}{{$0$}{$\frac{1}{3}$}}{{$0$}{$0$}}
}
\end{center}
We claim that this gate works for any $\eps < 1/3$. We first consider player
$v$.

\begin{itemize}
\item If $u$ plays $\zero$ as a pure strategy, then strategy $\zero$ gives payoff $1/3$ to
$v$, while strategy $\one$ gives payoff $0$, so in any $\eps$-WSNE with $\eps < 1/3$ we
have that $\zero$ is the only $\eps$-best response for $v$. 

\item If $u$ places at least $1/2$ probability on $\one$, then strategy $\zero$
gives at most $1/6$ payoff to $v$, while strategy $\one$ gives at least $1/2$ payoff to $v$. So 
in any $\eps$-WSNE with $\eps < 1/3$ we
have that $\one$ is the only $\eps$-best response for $v$.
\end{itemize}
Symmetrically, for player $w$ we have the following.
\begin{itemize}
\item If $u$ places at most $1/2$ probability on $\one$, then strategy $\zero$
gives at least $1/2$ payoff to $w$, while strategy
$\one$ gives at most $1/6$ payoff to $w$. So in any $\eps$-WSNE with $\eps <
1/3$ we have that $\zero$ is the only $\eps$-best response for $w$.

\item If $u$ plays $\one$ as a pure strategy, then strategy $\one$ gives payoff $1/3$ to
$w$, while strategy $\zero$ gives payoff $0$, so in any $\eps$-WSNE with $\eps
< 1/3$ we have that $\one$ is the only $\eps$-best response for $w$.
\end{itemize}

\noindent
From these properties, we can verify that the constraints imposed by the \PURE
gate are enforced correctly.
\begin{itemize}
\item If $u$ plays $\zero$ as a pure strategy, then both $v$ and $w$ play
$\zero$ as a pure strategy.

\item If $u$ plays $\one$ as a pure strategy, then both $v$ and $w$ play
$\one$ as a pure strategy.

\item No matter how $u$ mixes, at least one of $v$ and $w$ plays a pure
strategy, with $v$ playing pure strategy $\zero$ whenever $u$ places at most
$1/2$ probability on $\one$, and $w$ playing pure strategy $\one$ whenever $u$
places at least $1/2$ probability on $\one$.
\end{itemize}

\paragraph{\bf Hardness result.}

As we have argued above, each of the gate constraints from the \pcircuit
instance are enforced correctly in any $\eps$-WSNE of the polymatrix game with $\eps < 1/3$. Thus,
given such an $\eps$-WSNE, we can produce a satisfying assignment to the
\pcircuit instance using the mapping that we defined at the start of the
reduction. All of the games that we have presented are already normalized so
that each player's maximum possible payoff is 1, and minimum possible payoff is
0, so no extra normalization step is necessary.

Since we are reducing from \cref{cor:pcircuit-restricted}, we get
some extra properties about the structure of the polymatrix game. Observe that
the interaction graph of the polymatrix game is exactly the interaction graph of
the \pcircuit instance, as defined in \cref{sec:pcircuitdef}. Thus 
\cref{cor:pcircuit-restricted} implies that the polymatrix game is
degree three and bipartite. Moreover, by construction, each player has exactly
two strategies.
\end{proof}

\subsubsection{Hardness for \texorpdfstring{$\boldsymbol{\eps}$}{ε}-NE}

We now show that computing an \nash{\eps} in polymatrix games is \ppad/-complete for any constant $\eps < 2\sqrt{73} - 17 \approx 0.088$ even in degree $3$ bipartite games with two strategies per player.
It is worth noting that, given our \ppad/-hardness result for \wn{\eps'} for all $\eps' < 1/3$ in polymatrix games (\cref{thm: polymatrix wsne}), one can effortlessly get \ppad/-hardness for \nash{\eps} for $\eps$ smaller than roughly $0.00136$ by Lemma 7.4 of \cite{Rubinstein18-Nash-inapproximability}. However, by doing a direct reduction from \pcircuit we get an upper bound for \eps that is more than 64 times greater.

\begin{theorem}\label{thm: polymatrix ne}
It is \ppad/-hard to find an \nash{\eps} in a polymatrix game for all $\eps < 2\sqrt{73} - 17 \approx 0.088$, even in degree 3 bipartite games with two strategies per player.
\end{theorem}

In the remainder of this section we prove the theorem.
Fix any $\eps < 2\sqrt{73} - 17$. While for $\eps$-WSNE, we used pure strategies to represent $0$ and
$1$ values in \pcircuit, for $\eps$-NE we also have to allow
mixed strategies to encode $0$s and $1$s.
Specifically, given a \emph{\discr} $\del \in (0,1/4)$ which will be specified later, we map the strategies of the players to
values in \pcircuit in the following way.

\begin{itemize}
\item If $s_i$ places probability in the range $[1-\del,1]$ on $\zero$, then this corresponds to
setting variable $i$ to 0.
\item If $s_i$ places probability in the range $[1-\del,1]$ on $\one$, then this corresponds to 
setting variable $i$ to 1.
\item Otherwise, this corresponds to setting
variable $i$ to $\garbo$.
\end{itemize}

Pick an arbitrary rational $\delta \in \left(\frac{1 - \sqrt{1 - 8 \varepsilon}}{4}, \frac{9 -\sqrt{73}}{4} \right)$. Note that this interval is non-empty since $\eps < 2\sqrt{73} - 17$. Furthermore, we have $\delta \in (0,1/4)$ and $\delta$ satisfies $\eps < \delta(1-2\delta) < 2\sqrt{73} - 17$.

\paragraph{\bf \NOT gates.}

For a gate $g = (\NOT, u, v)$, the player $v$, who represents the output variable, will play the same bimatrix game against $u$ as that in the \wn{\eps} case, namely:

\begin{center}
	\bimatrixgame{3.7mm}{2}{2}{$v$}{$u$}%
	{{$\zero$}{$\one$}}%
	{{$\zero$}{$\one$}}%
	{
		\payoffpairs{1}{{$0$}{$1$}}{{$0$}{$0$}}
		\payoffpairs{2}{{$1$}{$0$}}{{$0$}{$0$}}
	}
\end{center}

We claim that this gadget will correctly simulate the \NOT gate for the encoding we use in this reduction. Let player $u$'s strategy be $(1-q, q)$, where $q \in [0,1]$, i.e., $u$ plays \zero with probability $1-q$, and \one with probability $q$. Similarly, let $v$'s strategy be $(1-p, p)$, where $p \in [0,1]$. We need to show that in any \nash{\eps}:
\begin{enumerate}[(i)]
    \item if $q \leq \del$ then $p \geq 1-\del$;
    \item if $q \geq 1-\del$ then $p \leq \del$.
\end{enumerate}
The expected payoff of $v$ for playing pure strategy \zero is $U_{v}(0) = q$ and for pure strategy \one it is $U_{v}(1) = 1-q$. Her expected payoff when playing mixed strategy $(1-p, p)$ is $U_{v}(p) = (1-p) \cdot U_{v}(0) + p \cdot U_{v}(1) = U_{v}(0) - p \cdot (U_{v}(0) - U_{v}(1)) = U_{v}(1) - (1-p) \cdot (U_{v}(1) - U_{v}(0))$. Now, in any \nash{\eps}, we must have that $U_{v}(p) \geq U_{v}(1) - \eps$ and $U_{v}(p) \geq U_{v}(0) - \eps$, which implies that
\begin{equation}\label{eq:eps-NE-one}
(1-p) \cdot (U_{v}(1) - U_{v}(0)) \leq \eps
\end{equation}
and
\begin{equation}\label{eq:eps-NE-zero}
p \cdot (U_{v}(0) - U_{v}(1)) \leq \eps.
\end{equation}
In case (i), we have $q \leq \delta$, and thus $U_{v}(1) - U_{v}(0) = 1-2q \geq 1-2\delta$. By \eqref{eq:eps-NE-one}, it follows that $1-p \leq \eps/(1-2\delta)$. Since $\eps < \delta(1-2\delta)$, we obtain $1-p \leq \delta$, i.e., $p \geq 1 - \delta$, as desired.

In case (ii), $q \geq 1 - \delta$ implies that $U_{v}(0) - U_{v}(1) = 2q-1 \geq 1 - 2\delta$. By \eqref{eq:eps-NE-zero}, we must have $p \leq \eps/(1 - 2\delta) \leq \delta(1-2\delta)/(1-2\delta) = \delta$, since $\eps < \delta(1-2\delta)$.

\paragraph{\bf Sequence of \NOT gates.}

The following is a crucial observation: when constructing a gadget, we can assume that the input to the gadget is encoded using cutoff $\delta$ (and so the gap between encoding of $0$ and $1$ is least $1-2\delta$) and that the output of the gadget is encoded using cutoff $1/2-\delta$ (and so the gap between encoding of $0$ and $1$ is least $2\delta$). The reason for this is that we can use a sequence of consecutive \NOT gates to amplify the gap in the encoding of the output from $2\delta$ up to $1-2\delta$ (or, in other words, to decrease the cutoff from $1/2-\delta$ down to $\delta$). The following lemma formalizes this observation.

\begin{lemma}\label{lem:NOT-sequence}
There exists a constant $k \in \mathbb{N}$ (that only depends on $\eps$ and $\delta$) such that any sequence of $2k$ consecutive $\NOT$ gadgets $g_1, \dots, g_{2k}$ has the following properties at any $\eps$-NE:
\begin{itemize}
\item If player $v_1$, who is the input to gadget $g_1$, satisfies $U_{v_1}(0)-U_{v_1}(1) \geq 2\delta$, then player $v_{2k+1}$, who is the output of gadget $g_{2k}$, satisfies $U_{v_{2k+1}}(0)-U_{v_{2k+1}}(1) \geq 1 - 2\delta$. In particular, player $v_{2k+1}$ places probability at least $1- \delta$ on $\zero$.
\item If player $v_1$ satisfies $U_{v_1}(1)-U_{v_1}(0) \geq 2\delta$, then player $v_{2k+1}$ satisfies $U_{v_{2k+1}}(1)-U_{v_{2k+1}}(0) \geq 1 - 2\delta$. In particular, player $v_{2k+1}$ places probability at least $1- \delta$ on $\one$.
\end{itemize}
\end{lemma}

\begin{proof}
Consider any $i \in [2k]$ and let $v_i$ denote the player who is the input to gadget $g_i$ and $v_{i+1}$ the player who is the output. We will show that in any $\eps$-NE:
\begin{enumerate}
\item If $U_{v_i}(0)-U_{v_i}(1) \in [2\delta,1-2\delta]$, then $U_{v_{i+1}}(1)-U_{v_{i+1}}(0) \geq U_{v_i}(0)-U_{v_i}(1) + C(\eps,\delta)$.
\item If $U_{v_i}(0)-U_{v_i}(1) \geq 1 - 2\delta$, then $U_{v_{i+1}}(1)-U_{v_{i+1}}(0) \geq 1 - 2\delta$.
\item If $U_{v_i}(1)-U_{v_i}(0) \in [2\delta,1-2\delta]$, then $U_{v_{i+1}}(0)-U_{v_{i+1}}(1) \geq U_{v_i}(1)-U_{v_i}(0) + C(\eps,\delta)$.
\item If $U_{v_i}(1)-U_{v_i}(0) \geq 1 - 2\delta$, then $U_{v_{i+1}}(0)-U_{v_{i+1}}(1) \geq 1 - 2\delta$.
\end{enumerate}
Here $C(\eps,\delta) > 0$ is a quantity that only depends on $\eps$ and $\delta$. Note that the lemma follows from this, since $2k$ is even. Indeed, it suffices to pick $k \in \mathbb{N}$ such that $2k \geq (1-2\delta-2\delta)/C(\eps,\delta)$. Furthermore, using \eqref{eq:eps-NE-one} and \eqref{eq:eps-NE-zero}, it is easy to see that $U_{v_{2k+1}}(0)-U_{v_{2k+1}}(1) \geq 1 - 2\delta$ implies that player $v_{2k+1}$ places probability at least $1- \delta$ on $\zero$, and the analogous statement for the other case.

Let us begin by proving statement 1. Let $\gamma := U_{v_i}(0)-U_{v_i}(1) \in [2\delta,1-2\delta]$ and let $p$ denote the probability that player $v_i$ plays strategy $\one$. By \eqref{eq:eps-NE-zero} we obtain that $p \leq \eps/\gamma$, and thus $U_{v_{i+1}}(1)-U_{v_{i+1}}(0) = (1-p) - p = 1-2p \geq 1-2\eps/\gamma$. We can thus write
\begin{equation*}
\begin{split}
\big(U_{v_{i+1}}(1)-U_{v_{i+1}}(0)\big) - \big(U_{v_i}(0)-U_{v_i}(1)\big) &\geq 1-\frac{2\eps}{\gamma} - \gamma\\
&= \frac{-\gamma^2 +\gamma -2\eps}{\gamma}\\
&= \frac{-\gamma^2 +\gamma -2\delta(1-2\delta)}{\gamma} + \frac{2(\delta(1-2\delta)-\eps)}{\gamma}\\
&= \frac{(\gamma-2\delta)(1-2\delta-\gamma)}{\gamma} + \frac{2(\delta(1-2\delta)-\eps)}{\gamma}\\
&\geq \frac{2(\delta(1-2\delta)-\eps)}{1-2\delta} =: C(\eps,\delta) > 0
\end{split}
\end{equation*}
where we used $\gamma \in [2\delta,1-2\delta]$ and $\eps < \delta(1-2\delta)$. Statement 3 is proved using the same arguments.

Let us now prove statement 2. Letting $\gamma := U_{v_i}(0)-U_{v_i}(1) \in [1-2\delta,1]$, we again obtain that $p \leq \eps/\gamma$, and thus
$$U_{v_{i+1}}(1)-U_{v_{i+1}}(0) = (1-p) - p = 1-2p \geq 1-\frac{2\eps}{\gamma} \geq 1-\frac{2\delta(1-2\delta)}{1-2\delta} = 1-2\delta$$
as desired. The proof of statement 4 is completely analogous.
\end{proof}

\paragraph{\bf \AND gates.}

For a gate $g = (\AND, u, v, w)$, the player $w$, who represents the output
variable, will play the
following matrix games against $u$ and $v$, who represent the input variables.
\begin{center}
	\bimatrixgame{3.7mm}{2}{2}{$w$}{$u$}%
	{{$\zero$}{$\one$}}%
	{{$\zero$}{$\one$}}%
	{
		\payoffpairs{1}{{$\frac{1}{2}$}{$0$}}{{$0$}{$0$}}
		\payoffpairs{2}{{$0$}{$\,\,\frac{1+\delta}{6-2\delta}$}}{{$0$}{$0$}}
	}
	\hskip 1.5cm
	\bimatrixgame{3.7mm}{2}{2}{$w$}{$v$}%
	{{$\zero$}{$\one$}}%
	{{$\zero$}{$\one$}}%
	{
		\payoffpairs{1}{{$\frac{1}{2}$}{$0$}}{{$0$}{$0$}}
		\payoffpairs{2}{{$0$}{$\,\,\frac{1+\delta}{6-2\delta}$}}{{$0$}{$0$}}
	}
\end{center}
We claim that this gadget correctly simulates the \AND gate (together with a sequence of \NOT gates applied to the output). Let player $u$ play strategy $(1-q_1, q_1)$, and let $v$ play strategy $(1-q_2, q_2)$, where $q_1, q_2 \in [0,1]$. By \cref{lem:NOT-sequence} it suffices to show that
\begin{enumerate}[(i)]
	\item if $q_1 \leq \del$ or $q_2 \leq \del$, then $U_{w}(0) - U_{w}(1) \geq 2\delta$,
	\item if $q_1 \geq 1-\del$ and $q_2 \geq 1-\del$, then $U_{w}(1) - U_{w}(0) \geq 2\delta$.
\end{enumerate}
where $U_{w}(0)$ and $U_{w}(1)$ are the expected payoff of player $w$ for strategy $\zero$ and $\one$ respectively. Observe that $U_{w}(0) = \frac{1}{2}(1-q_1) + \frac{1}{2}(1-q_2)$ and $U_{w}(1) = \frac{1+\delta}{6-2\delta} q_1 + \frac{1+\delta}{6-2\delta} q_2$. Below, we use the fact that $(1-3\delta)/(3-\delta) \geq 2\delta$, which follows from $\delta < (9-\sqrt{73})/4$.

In case (i), $U_{w}(0) - U_{w}(1)$ is minimal when $q_1 = \delta$ and $q_2=1$ (and also when $q_1=1$ and $q_2=\delta$, which is symmetric). Thus, we necessarily have
$$U_{w}(0) - U_{w}(1) \geq \frac{1}{2} \cdot (1-\delta) - \frac{1+\delta}{6-2\delta}  \cdot \delta - \frac{1+\delta}{6-2\delta} = \frac{1-3\delta}{3-\delta} \geq 2\delta.$$
For case (ii), we have
$$U_{w}(1) - U_{w}(0) \geq \frac{1+\delta}{6-2\delta} \cdot  (1-\delta) + \frac{1+\delta}{6-2\delta}  \cdot 
(1-\delta) - \frac{1}{2} \cdot \delta - \frac{1}{2} \cdot \delta = \frac{1-3\delta}{3-\delta} \geq 2\delta$$
and thus the gate is correctly simulated.

\paragraph{\bf \PURE gates.}

For a gate $g = (\PURE, u, v, w)$, the players $v$ and $w$, who represent the
output variables, play the following games against $u$, who represents the input
variable.
\begin{center}
	\bimatrixgame{3.7mm}{2}{2}{$v$}{$u$}%
	{{$\zero$}{$\one$}}%
	{{$\zero$}{$\one$}}%
	{
		\payoffpairs{1}{{$\,\,\frac{1+2\delta}{3-2\delta}$}{$0$}}{{$0$}{$0$}}
		\payoffpairs{2}{{$0$}{$1$}}{{$0$}{$0$}}
	}
	\hskip 1.5cm
	\bimatrixgame{3.7mm}{2}{2}{$w$}{$u$}%
	{{$\zero$}{$\one$}}%
	{{$\zero$}{$\one$}}%
	{
		\payoffpairs{1}{{$1$}{$0$}}{{$0$}{$0$}}
		\payoffpairs{2}{{$0$}{$\,\,\frac{1+2\delta}{3-2\delta}$}}{{$0$}{$0$}}
	}
\end{center}

We claim that this gadget correctly simulates the \PURE gate (together with a sequence of \NOT gates applied to the output). Let player $u$ play strategy $(1-q, q)$ where $q \in [0,1]$. By \cref{lem:NOT-sequence} it suffices to show that
\begin{enumerate}[(i)]
    \item if $q \leq \del$, then $U_{v}(0) - U_{v}(1) \geq 2\delta$ and $U_{w}(0) - U_{w}(1) \geq 2\delta$,
    \item if $q \geq 1-\del$, then $U_{v}(1) - U_{v}(0) \geq 2\delta$ and $U_{w}(1) - U_{w}(0) \geq 2\delta$,
    \item if $q \in (\del, 1 - \del)$, then $|U_{v}(0) - U_{v}(1)| \geq 2\delta$ or $|U_{w}(0) - U_{w}(1)| \geq 2\delta$.
\end{enumerate}
Observe that $U_{v}(0) = \frac{1+2\delta}{3-2\delta}(1-q)$ and $U_{v}(1) = q$, as well as $U_{w}(0) = 1-q$ and $U_{w}(1) = \frac{1+2\delta}{3-2\delta} q$. Below, we use the fact that $(1-2\delta)/(3-2\delta) \geq 2\delta$, which follows from $(1-2\delta)/(3-2\delta) \geq (1-3\delta)/(3-\delta) \geq 2\delta$ as argued above.

In case (i), we have $q \leq \delta$ and thus
$$U_{v}(0) - U_{v}(1) \geq \frac{1+2\delta}{3-2\delta} \cdot (1-\delta) - \delta = \frac{1-2\delta}{3-2\delta} \geq 2\delta$$
and also
$$U_{w}(0) - U_{w}(1) \geq 1-\delta - \frac{1+2\delta}{3-2\delta} \cdot \delta = \frac{3-6\delta}{3-2\delta} \geq \frac{1-2\delta}{3-2\delta} \geq 2\delta.$$
In case (ii), we have $q \geq 1 -\delta$ and thus
$$U_{v}(1) - U_{v}(0) \geq 1-\delta - \frac{1+2\delta}{3-2\delta} \cdot \delta = \frac{3-6\delta}{3-2\delta} \geq \frac{1-2\delta}{3-2\delta} \geq 2\delta$$
as well as
$$U_{w}(1) - U_{w}(0) \geq \frac{1+2\delta}{3-2\delta} \cdot (1-\delta) - \delta = \frac{1-2\delta}{3-2\delta} \geq 2\delta.$$
For case (iii) we consider two subcases. If $q \geq 1/2$, then
$$U_{v}(1) - U_{v}(0) \geq \frac{1}{2} - \frac{1+2\delta}{3-2\delta} \cdot \frac{1}{2} = \frac{1-2\delta}{3-2\delta}  \geq 2\delta$$
as above. On the other hand, if $q \leq 1/2$, then
$$U_{w}(0) - U_{w}(1) \geq \frac{1}{2} - \frac{1+2\delta}{3-2\delta} \cdot \frac{1}{2} = \frac{1-2\delta}{3-2\delta}  \geq 2\delta.$$

\paragraph{\bf Hardness result.}

According to the case analysis above, all gate constraints from the
\pcircuit instance are enforced correctly in any \nash{\eps} of the polymatrix
game. Thus, given an \nash{\eps} solution, we can produce a
satisfying assignment to the \pcircuit instance using the mapping that we
defined at the start of the reduction. 

Furthermore, note that the polymatrix instance that we construct is normalized and that it is degree 3, bipartite, and has two strategies per player. In particular, note that the addition of a sequence of \NOT gates inside the gadgets for \AND gates and \PURE gates (according to \cref{lem:NOT-sequence}) does not increase the degree and does not destroy the bipartiteness of the graph, since the sequence has even length.

\subsubsection{Hardness for Win-Lose Polymatrix games}
A polymatrix game is {\em win-lose} if for every player $i \in [n]$ it holds
that the entries of every payoff matrix $A_{ij}$ belong to $\{0, a_i\}$, for
some $a_i \in (0,1]$. In other words, player $i$ either ``wins'' payoff $a_i$
from the interaction with another player, or ``loses'' and gets payoff 0. Prior
work has shown that computing an \eps-WSNE is \ppad/-complete for inverse-polynomial
$\eps$ via a reduction from \gcircuit~\cite{LiuLD21-sparse-winlose-polymatrix}. In this section we
will show that hardness holds for all $\eps < 1/3$, by modifying our hardness
result for $\eps$-WSNE given earlier.

\begin{theorem}\label{thm:win-lose-polymatrix}
It is \ppad/-hard to find an $\eps$-WSNE in a win-lose polymatrix game for all $\eps <
1/3$, even in degree 7 bipartite games with two strategies per player.
\end{theorem}

\begin{proof}
We begin by explaining how to modify the payoff matrices described in
\cref{sec:WSNE-polymatrix}. Observe that the payoff matrix for \NOT gates is already
win-lose, so we only need to convert \AND gates and \PURE gates into
win-lose payoff matrices. The high level idea for these two gates is to create a
number of intermediate players, who all copy the input player's strategy, and
whose payoffs to the output player sum to the same payoffs that we used in the
$\eps$-WSNE hardness construction. \cref{fig:win-lose-gadgets} depicts
the interaction graphs for these two gates.

\begin{figure}
	\centering
	\begin{tikzpicture}[roundnode/.style={circle, draw=black, inner sep=0, minimum size=2.5mm}]

\node[roundnode,label=above left:{$u$}] (ANDu) at (0,1.5) {};
\node[roundnode,label=below left:{$v$}] (ANDv) at (0,-0.5) {};
\node[roundnode,label={[label distance=-0.1cm]above right:{$u_1$}}] (ANDu1) at (1.5,2.2) {};
\node[roundnode,label={[label distance=-0.1cm]above right:{$u_2$}}] (ANDu2) at (1.5,1.5) {};
\node[roundnode,label={[label distance=-0.1cm]above right:{$u_3$}}] (ANDu3) at (1.5,0.8) {};
\node[roundnode,label={[label distance=-0.1cm]below right:{$v_1$}}] (ANDv1) at (1.5,0.2) {};
\node[roundnode,label={[label distance=-0.1cm]below right:{$v_2$}}] (ANDv2) at (1.5,-0.5) {};
\node[roundnode,label={[label distance=-0.1cm]below right:{$v_3$}}] (ANDv3) at (1.5,-1.2) {};
\node[roundnode,label={above right:{$w$}}] (ANDw) at (5,0.5) {};

\draw (ANDu) -- (ANDu1);
\draw (ANDu) -- (ANDu2);
\draw (ANDu) -- (ANDu3);
\draw (ANDv) -- (ANDv1);
\draw (ANDv) -- (ANDv2);
\draw (ANDv) -- (ANDv3);
\draw (ANDu1) -- (ANDw);
\draw (ANDu2) -- (ANDw);
\draw (ANDu3) -- (ANDw);
\draw (ANDv1) -- (ANDw);
\draw (ANDv2) -- (ANDw);
\draw (ANDv3) -- (ANDw);

\node[roundnode,label={above left:{$u$}}] (PUREu) at (8,0.5) {};
\node[roundnode,label={above left:{$u_1$}}] (PUREu1) at (9.5,2.2) {};
\node[roundnode,label={above left:{$u_2$}}] (PUREu2) at (9.5,0.5) {};
\node[roundnode,label={below left:{$u_3$}}] (PUREu3) at (9.5,-1.2) {};
\node[roundnode,label={above right:{$v$}}] (PUREv) at (11.5,1.5) {};
\node[roundnode,label={above right:{$w$}}] (PUREw) at (11.5,-0.5) {};

\draw (PUREu) -- (PUREu1);
\draw (PUREu) -- (PUREu2);
\draw (PUREu) -- (PUREu3);
\draw (PUREu1) -- (PUREv);
\draw (PUREu2) -- (PUREv);
\draw (PUREu3) -- (PUREv);
\draw (PUREu1) -- (PUREw);
\draw (PUREu2) -- (PUREw);
\draw (PUREu3) -- (PUREw);

\node at (2.2,-2.2) {\large \AND gadget};
\node at (9.8,-2.2) {\large \PURE gadget};

\end{tikzpicture}
	\caption{The gadgets used for win-lose polymatrix games.}
	\label{fig:win-lose-gadgets}
\end{figure}
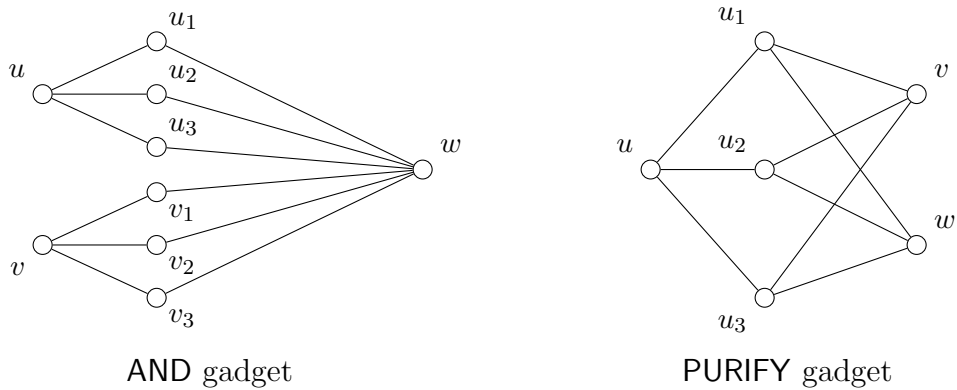

\paragraph{\bf \AND gates.}
For a gate $g = (\AND, u, v, w)$, we will create the \AND gadget from \cref{fig:win-lose-gadgets}. 
The player $w$, that represents the output variable, will play against the auxiliary players $u_1, u_2, u_3$, and $v_1, v_2, v_3$ who in turn will ``copy'' the values, represent the input variables. The payoff matrices for this gadget will be defined as follows, where $X \in \{u_1, v_1\}$, $Y \in \{u_2, u_3, v_2, v_3\}$, and when $Q=v$, then $P \in \{v_1, v_2, v_3\}$ and when $Q=u$, then $P \in \{u_1, u_2, u_3\}$.
\begin{center}
\bimatrixgame{3.7mm}{2}{2}{$w$}{$X$}%
{{$\zero$}{$\one$}}%
{{$\zero$}{$\one$}}%
{
\payoffpairs{1}{{$\frac{1}{6}$}{$0$}}{{$0$}{$0$}}
\payoffpairs{2}{{$0$}{$\frac{1}{6}$}}{{$0$}{$0$}}
}
\hskip 1.5cm
\bimatrixgame{3.7mm}{2}{2}{$w$}{$Y$}%
{{$\zero$}{$\one$}}%
{{$\zero$}{$\one$}}%
{
\payoffpairs{1}{{$\frac{1}{6}$}{$0$}}{{$0$}{$0$}}
\payoffpairs{2}{{$0$}{$0$}}{{$0$}{$0$}}
}
\hskip 1.5cm
\bimatrixgame{3.7mm}{2}{2}{$P$}{$Q$}%
{{$\zero$}{$\one$}}%
{{$\zero$}{$\one$}}%
{
\payoffpairs{1}{{$1$}{$0$}}{{$0$}{$0$}}
\payoffpairs{2}{{$0$}{$1$}}{{$0$}{$0$}}
}
\end{center}
We claim that this gate will work for all $\eps < 1/3$.
Firstly, it is not hard to observe that for the game played between $P$ and $Q$, in every \eps-WSNE of the game: if $Q$ plays $\zero$, then $P$ plays $\zero$; if $Q$ plays $\one$, then $P$ plays $\one$.
Hence, we have the following for the gadget. 
\begin{itemize}
\item If both $u$ and $v$ play $\one$ as a pure strategy, then every auxiliary player $u_1, u_2, u_3$ and $v_1, v_2, v_3$ plays $\one$. 
Hence, the payoff to $w$ for playing $\zero$ is $0$, and the payoff to $w$ for playing $\one$ is $1/6 + 1/6 = 1/3$. So, if $\eps < 1/3$, then in all $\eps$-WSNEs the only $\eps$-best response for $w$ is $\one$, as required by the constraints of the \AND gate.

\item If at least one of $u$ and $v$ play $\zero$ as a pure strategy, say that $u$ plays $\zero$, then the auxiliary players $u_1, u_2, u_3$ will play $\zero$ as well. 
Thus, the payoff to $w$ for playing $\zero$ is at least $1/2$, while the payoff to $w$ for playing $\one$ is at most $1/6$. 
So, if $\eps < 1/3$, then in all $\eps$-WSNEs the only $\eps$-best response for $w$ is $\zero$, as required by the constraints of the \AND gate.

\item The \AND gate places no other constraints on the variable $w$, so we can ignore all other cases, e.g., the case where both $u$ and $v$ play strictly mixed
strategies.
\end{itemize}

\paragraph{\bf \PURE gates.}
For a gate $g = (\PURE, u, v, w)$, we will create a \PURE gadget from \cref{fig:win-lose-gadgets}. 
Players $v$ and $w$, that represent the output variables, will play against the auxiliary players $u_1, u_2, u_3$, who in turn will ``copy'' the values of the input variable that corresponds to player $u$. The payoff matrices for this gadget will be defined as follows, where $P \in \{u_1, u_2, u_3\}$, $Q \in \{v, w\}$, and $X \in \{u_1, u_2\}$.
\begin{center}
\hskip 0.2cm
\bimatrixgame{3.4mm}{2}{2}{$v$}{$X$}%
{{$\zero$}{$\one$}}%
{{$\zero$}{$\one$}}%
{
\payoffpairs{1}{{$0$}{$0$}}{{$0$}{$0$}}
\payoffpairs{2}{{$0$}{$\frac{1}{3}$}}{{$0$}{$0$}}
}
\hskip 0.5cm
\bimatrixgame{3.4mm}{2}{2}{$w$}{$X$}%
{{$\zero$}{$\one$}}%
{{$\zero$}{$\one$}}%
{
\payoffpairs{1}{{$\frac{1}{3}$}{$0$}}{{$0$}{$0$}}
\payoffpairs{2}{{$0$}{$0$}}{{$0$}{$0$}}
}
\hskip 0.5cm
\bimatrixgame{3.4mm}{2}{2}{$Q$}{$u_3$}%
{{$\zero$}{$\one$}}%
{{$\zero$}{$\one$}}%
{
\payoffpairs{1}{{$\frac{1}{3}$}{$0$}}{{$0$}{$0$}}
\payoffpairs{2}{{$0$}{$\frac{1}{3}$}}{{$0$}{$0$}}
}
\hskip 0.5cm
\bimatrixgame{3.4mm}{2}{2}{$P$}{$u$}%
{{$\zero$}{$\one$}}%
{{$\zero$}{$\one$}}%
{
\payoffpairs{1}{{$1$}{$0$}}{{$0$}{$0$}}
\payoffpairs{2}{{$0$}{$1$}}{{$0$}{$0$}}
}
\end{center}
We claim that this gate will work for all $\eps < 1/3$. Similarly to the \AND gadget, it is not hard to observe that for the game played between $u$ and $P$, in every \eps-WSNE of the game: if $u$ plays $\zero$, then $P$ plays $\zero$; if $u$ plays $\one$, then $P$ plays $\one$.

\begin{itemize}
\item If $u$ plays $\zero$ as a pure strategy, then auxiliary players $u_1, u_2, u_3$ play $\zero$ as well.
Thus, strategy $\zero$ gives payoff $1/3$ to both $v$ and $1$ to $w$, while strategy $\one$ gives payoff $0$ to both $v$ and $w$, so in any $\eps$-WSNE with 
$\eps < 1/3$ we have that $\zero$ is the only $\eps$-best response for both $v$ and $w$.

\item If $u$ plays $\one$ as a pure strategy, then auxiliary players $u_1, u_2, u_3$ play $\one$ as well.
Thus, strategy $\one$ gives payoff $1$ to $v$ and  and payoff $1/3$ to $w$, while strategy $\zero$ gives payoff $0$ to both $v$ and $w$. So, in any $\eps$-WSNE 
with $\eps < 1/3$ we have that $\one$ is the only $\eps$-best response for $v$ and $w$.

\item Assume now that $u$ mixes between pure strategies $\zero$ and $\one$. Then, the auxiliary players $u_1, u_2, u_3$ might mix between their pure strategies as well; for $i \in \{1, 2, 3\}$ let $p_i$ be the probability player $u_i$ places on pure strategy $\zero$. 
Then, we have that the following hold. 
\begin{itemize}
    \item The payoff for $v$ is $\frac{1}{3}\cdot p_3$ for playing $\zero$, and $1-\frac{1}{3}\cdot (p_1+p_2+p_3)$ for playing $\one$.
    \item The payoff for $w$ is $\frac{1}{3}\cdot (p_1+p_2+p_3)$ for playing $\zero$, and $\frac{1}{3}-\frac{1}{3}\cdot p_3$ for playing $\one$.
\end{itemize}
Now assume that in an \eps-WSNE, with $\eps < 1/3$, $v$ mixes between $\zero$ and $\one$. Then, it must hold that the difference of the payoffs between the two pure strategies is bounded by 1/3, i.e., 
\begin{align*}
    \left|\frac{1}{3}\cdot p_3 - \left(1-\frac{1}{3}\cdot (p_1+p_2+p_3)\right) \right| <  \frac{1}{3}.
\end{align*}
Hence, we get that
\begin{align*}
    \frac{2}{3} - \frac{1}{3}\cdot p_3 <  \frac{1}{3}\cdot (p_1+p_2+p_3).
\end{align*}
We claim that in this case $w$ will have to play $\zero$ as a pure strategy.
Indeed, observe that the payoff of $w$ from $\zero$ is $\frac{1}{3}\cdot (p_1+p_2+p_3) > \frac{2}{3} - \frac{1}{3}\cdot p_3$, where the inequality follows from above.
On the other hand, the payoff from $\one$ is $\frac{1}{3} - \frac{1}{3}\cdot p_3$. Hence, the difference between the payoffs of the two pure strategies is larger than 1/3. Thus, in any \eps-WSNE, with $\eps < 1/3$, if $v$ plays a mixed strategy, then \zero is the only \eps-best response for $w$.

For the second case, assume that in an \eps-WSNE, with $\eps < 1/3$, $w$ mixes between $\zero$ and $\one$. Then, it must hold that the difference of the payoffs between the two pure strategies is bounded by 1/3, i.e., 
\begin{align*}
     \left|\frac{1}{3}\cdot (p_1+p_2+p_3) - \left(\frac{1}{3}-\frac{1}{3}\cdot p_3\right) \right| <  \frac{1}{3}.
\end{align*}
Hence, we get that
\begin{align*}
     \frac{1}{3}\cdot p_3 < \frac{2}{3} -  \frac{1}{3}\cdot (p_1+p_2+p_3).
\end{align*}
We claim that in this case $v$ will have to play $\one$ as a pure strategy.
Indeed, observe that the payoff of $v$ from $\zero$ is $\frac{1}{3}\cdot p_3 < \frac{2}{3} - \frac{1}{3}\cdot (p_1+p_2+p_3)$, where the inequality follows from above.
On the other hand, the payoff from $\one$ is $1-\frac{1}{3}\cdot (p_1+p_2+p_3)$. Hence, the difference between the payoffs of the two pure strategies is larger than 1/3. Thus, in any \eps-WSNE, with $\eps < 1/3$, if $w$ plays a mixed strategy, then \one is the only \eps-best response for $v$.
\end{itemize}

\paragraph{\bf Hardness result.}

As we have argued above, each of the gate constraints from the \pcircuit
instance is enforced correctly in any $\eps$-WSNE of the  win-lose polymatrix game with $\eps < 1/3$. 
Thus, given such an $\eps$-WSNE, we can produce a satisfying assignment to the
\pcircuit instance using the same mapping as we have used before.
In addition, all of the games that we have presented are win-lose games and are already normalized, and in addition every player has two strategies.

Now, observe that the reduction described above does not immediately yield a bipartite graph and in addition the maximum degree of the constructed graph is at most nine (this can happen if the vertex that represents the output of an \AND gadget, is the input to a different \AND gadget as well). 
However, it is not difficult to tweak the construction and get these two properties too.

Firstly, we perform the following preprocessing step on the \pcircuit instance, which results in having a graph with maximum degree seven in the win-lose game constructed here. We replace every node $v$ of the \pcircuit instance with three nodes $v_1, v_2, v_3$ where: $v_3$ is the output of a \NOT gate with input $v_2$; $v_2$ in turn is the output of a \NOT gate with input $v_1$; $v_3$ is the input of the gate(s) that had as input the vertex $v$; $v_1$ is the output of the gate that had as output $v$. It is not difficult to see that this transformation does not change the solutions of the original instance. In addition, observe that under this preprocessing step and \cref{cor:pcircuit-restricted}, the maximum degree of the constructed win-lose game is now seven.

Next, in order to get a bipartite graph we modify the gadget for \NOT gates as follows. For a gate $g = (\NOT, u, v)$ we create an auxiliary player $u'$, where $u'$ ``copies'' the strategy of $u$ and $v$ ``negates'' the strategy of $u'$. So, overall $v$ ``negates'' the strategy of $u$. To achieve this, the game played between $u'$ and $v$ corresponds to the game constructed for the \NOT gates for general polymatrix games and the game between $u$ and $u'$ corresponds to the rightmost game constructed for the \AND gate for win-lose polymatrix games. Observe now that the final graph is bipartite: one side contains only the vertices from the preprocessed \pcircuit instance and the other side contains only the auxiliary vertices we have created for the gadgets.
\end{proof}

\subsection{Graphical Games}
\label{sec:graphical}

In this section we give results for graphical games (defined in
Section~\ref{sec:gamedefs}) for the two well-studied
notions of approximate equilibrium: $\eps$-NEs and $\eps$-WSNEs. 

We start with results for $\eps$-WSNEs, where we first give a straightforward
algorithm for finding a \wn{\left(1 - \frac{2}{2^d + 1}\right)} in polynomial
time in a two-action graphical game with maximum in-degree
$d$~(\cref{thm:graphical-wsne-UB}), and we then show that finding an \wn{\eps}
in two-action graphical games with maximum in-degree $d \geq 2$ is
\ppad/-complete for any $\eps < 1 - \frac{2}{2^d + 1}$ (\cref{thm:wsne-graph}).
Combined, these two results give a complete characterisation of the computational
complexity of finding an $\eps$-WSNE in a graphical game with maximum in-degree
$d$.

We then move on to $\eps$-NEs, where we first give a straightforward algorithm
that finds a $1/2$-NE in a two-action graphical game
(\cref{thm:eacy-graphical-ne}).
For the lower bound, we then show that it is \ppad/-hard to find an $\eps$-NE in a two-action graphical game for any constant $\eps < 1/2$
(\cref{thm:graphical-wsne-hardness}).
Combined, these results give a complete characterisation of the computational
complexity of finding an $\eps$-NE in a two-action graphical game.

\subsubsection{Upper and Lower Bounds for Well-Supported Nash Equilibria}

We now show results for $\eps$-WSNEs in graphical games.

\paragraph{An easy upper bound.}

We present a polynomial time algorithm to compute a \wn{\left(1 - \frac{2}{2^d
+ 1}\right)} in any two-action graphical game with maximum in-degree $d \geq
2$. The algorithm relies on a simple and natural approach that has been used
for similar problems by Liu et al.~\cite{LiuLD21-sparse-winlose-polymatrix}.

The algorithm proceeds in two steps. In the first step, it iteratively checks for a player that has an action that is {\em \eps-dominant}; an action which, if played with probability 1, will satisfy the \wn{\eps} conditions no matter what strategies the in-neighbours play.
Here, we will set $\eps =  1 - \frac{2}{2^d + 1}$, where $d$ is the maximum in-degree of the graph.
If such a player with an \eps-dominant action exists, the algorithm fixes the strategy of that player, updates the game accordingly, and iterates until there is no such player left.
In the second step, the algorithm lets all remaining players mix uniformly, i.e., every player without an \eps-dominant action plays each of its two actions with probability $1/2$.

\begin{theorem}
	\label{thm:graphical-wsne-UB}
	There is a polynomial-time algorithm that finds a \wn{\left(1 - \frac{2}{2^d + 1}\right)} in a two-action graphical game with maximum in-degree $d$.
\end{theorem}
\begin{proof}
	It is easy to see that the algorithm described above runs in polynomial time. In particular, we can check if a player has an \eps-dominant action by simply going over all possible action profiles of its in-neighbours. We will prove it computes an \wn{\eps} for $\eps = 1 - \frac{2}{2^d + 1}$.
	By definition of \eps-dominance, the players whose actions were fixed in the first step satisfy the constraints of \eps-WSNE.
	Notice that after fixing any such player to play an \eps-dominant action, we get a smaller graphical game. Thus, it suffices to consider the graphical game we obtain after the end of the first step, and to show that if all (remaining) players mix uniformly, this is an \wn{\eps}.
	
	So, consider a player $i$ in the graphical game we obtain after the end of the first step, and denote its in-degree by $k := |\innei{i}| \leq d$. Since all in-neighbours of player $i$ are mixing uniformly, the expected payoff of player $i$ for playing action $0$ is $\sum_{\vba \in A(\innei{i})} R_{i}(0 ; \vba)/2^k$, and for playing action $1$, it is $\sum_{\vba \in A(\innei{i})} R_{i}(1 ; \vba)/2^k$. The constraints of \wn{\eps} for player $i$ are satisfied if both actions are $\eps$-best responses, i.e., if
	\begin{align*}
	    \left|\sum_{\vba \in A(\innei{i})} R_{i}(0 ; \vba)/2^k - \sum_{\vba \in A(\innei{i})} R_{i}(1 ; \vba)/2^k\right| \leq \eps.
	\end{align*}
	This can be rewritten as $ \frac{1}{2^k} \left| \sum_{\vba \in A(\innei{i})} f_{i}(\vba) \right| \leq \eps $, where, for any action profile $\vba \in A(\innei{i})$, we let
	\begin{align*}
		f_{i}(\vba) := R_{i}(0 ; \vba) - R_{i}(1 ; \vba).
	\end{align*}
	Note that since all payoffs lie in $[0,1]$, we always have $f_{i}(\vba) \in [-1,1]$.
	
	Let $M := \sum_{\vba \in A(\innei{i})} f_{i}(\vba)$. To prove the correctness of the algorithm, it suffices to prove that $\left| M \right| \leq 2^k - 1 - \eps$; since then 
	$ \frac{1}{2^k} \cdot |M| \leq 1- \frac{1+\eps}{2^k} \leq 1- \frac{1+\eps}{2^d} = \eps$. To see why this is indeed the case, observe the following. Since player $i$ does not have an \eps-dominant action (otherwise, it would have been removed in the first step), it means that there exist $\vba, \vba' \in A(\innei{i})$  such that 
	\begin{align}
		\label{eq:graphical-ub}
		f_{i}(\vba) > \eps \quad \text{and} \quad  f_{i}(\vba') < -\eps.
	\end{align}
	Given that $M$ is the sum of $2^k$ terms, each of them upper bounded by $1$, and at least one of them upper bounded by $-\eps$ by \eqref{eq:graphical-ub}, it follows that $M \leq 2^k-1 -\eps$. Similarly, since each term is also lower bounded by $-1$, and one of them is lower bounded by $\eps$, we also obtain that $M \geq -2^k-1 + \eps$. Thus, $\left| M \right| \leq 2^k - 1 - \eps$, as desired, and this completes the proof of correctness.
\end{proof}

\paragraph{\bf The lower bound.}
In this section we prove a matching lower bound for \cref{thm:graphical-wsne-UB},
which essentially proves that computing an $\eps$-WSNE in two-action graphical games is \ppad/-complete for every constant $\eps \in (0,1)$.

We will prove our result by a reduction from \pcircuit. For the remainder of this section, we fix $\eps < 1 - \frac{2}{2^d+1}$.
Given a \pcircuit instance, we build a two-action graphical game, where the two actions will be named $\zero$ and $\one$. For any given $d \geq 2$, the game will have in-degree at most $d$.
Each node $v$ of the \pcircuit instance will correspond to a player in the 
game -- the game will have some additional auxiliary players too -- whose strategy in any \eps-WSNE will encode a solution to the \pcircuit problem as follows. 
Given a strategy $s_v$ for the player that corresponds to node $v$, we define the assignment $\valonly$ for \pcircuit such that:
\begin{itemize}
	\item if $s_{v}(\zero) = 1$, then $\val{v} = 0$;
	\item if $s_{v}(\one) = 1$, then $\val{v} = 1$;
	\item otherwise, $\val{v} = \garbo$.
\end{itemize}

We now give implementations for \NOT, \AND, and \PURE gates. We note that, in
all three cases, the payoff received by player $v$ is only affected by the
actions chosen by the players representing the inputs to the (unique) gate $g$
that outputs to $v$. Thus, we can argue about the equilibrium condition at $v$
by only considering the players involved in gate $g$, and we can ignore all
other gates while doing this.

\paragraph{\bf \NOT gates.}
For a gate $g = (\NOT, u, v)$ -- where recall that $u$ is the input variable and $v$ is the output variable -- we create a gadget involving players $u$ and $v$, where player $v$ has a unique incoming edge from $u$. The payoffs of $v$ are defined as follows.
\begin{itemize}
	\item If $u$ plays $\zero$, then $v$ gets payoff 0 from playing $\zero$ and payoff
	1 from playing  $\one$.
	\item If $u$ plays $\one$, then $v$ gets 1 from $\zero$ and 0 from $\one$.
\end{itemize}
We claim that this gadget works correctly.
\begin{itemize}
	\item[-] If $s_{u}(\zero) = 1$, i.e., $u$ encodes 0, observe that for player $v$ action $\zero$ yields payoff 0, while action $\one$ yields payoff 1. Hence, by the constraints imposed by \eps-WSNE it must hold that	$s_v(\one)=1$, and thus $v$ encodes $1$.
	\item[-] Using identical reasoning, we can prove that if $s_{u}(\one) = 1$, then $s_{v}(\one) = 0$ in any $\eps$-WSNE.
\end{itemize}

\paragraph{\bf \AND gates.}
For a gate $g = (\AND, u, v, w)$ we create the following gadget with players $u, v$ and $w$, where $u$ and $v$ are the in-neighbors of $w$. The payoffs of $w$ are as follows.
\begin{itemize}
	\item If $s_u(\one) = 1$ and $s_v(\one) = 1$, then $w$ gets payoff 0 from playing $\zero$ and payoff 1 from playing $\one$.
	\item For any other action profile of $u$ and $v$, player $w$ gets 1 from $\zero$ and 0 from $\one$.
\end{itemize}
Next we argue that this gadget works correctly.
\begin{itemize}
	\item[-] If $s_u(\one) = 1$ and $s_v(\one) = 1$, i.e. both $u$ and $v$ encode 1, observe that for player $w$ action \zero yields payoff 0, while action \one yields payoff 1. Hence, by the constraints imposed by \eps-WSNE it must hold that	$s_w(\one)=1$, and thus $w$ encodes $1$. 
	\item[-] If at least one of $u$ or $v$ encodes 0, then for player $w$ action \zero yields expected payoff 1 while action \one yields expected payoff 0. Hence, by the constraints imposed by \eps-WSNE it must hold that	$s_w(\zero)=1$, and thus $w$ encodes $0$.
\end{itemize}

\paragraph{\bf \PURE gates.}
For a gate $g = (\PURE, u, v, w)$ we create the following gadget with $d+3$ players.
We introduce auxiliary players $u_1, u_2, \ldots, u_d$. Each player $u_i$ has a unique incoming edge from $u$. The idea is that in any \eps-WSNE, every player $u_i$ ``copies'' the strategy of player $u$.
\begin{itemize}
	\item If $u$ plays \zero, then $u_i$ gets 1 from \zero and 0 from \one.
	\item If $u$ plays \one, then $u_i$ gets 0 from zero and 1 from \one.
\end{itemize}

\begin{lemma}\label{lem:graphical-wsne-aux}
	At any \eps-WSNE the following hold for every $i \in [d]$: if $s_u(\zero)=1$, then $s_{u_i}(\zero)=1$;
	if $s_u(\one)=1$, then $s_{u_i}(\one)=1$.
\end{lemma}
\begin{proof}
	If $s_u(\zero)=1$, then for player $u_i$ action \zero yields payoff 1, while action \one yields payoff 0. Thus, the constraints of \eps-WSNE dictate that $s_{u_i}(\zero)=1$.
	If $s_u(\one)=1$, then for player $u_i$ action \zero yields payoff 0, while action \one yields payoff 1. Thus, the constraints of \eps-WSNE dictate that $s_{u_i}(\one)=1$.
\end{proof}
Next, we describe the payoff tensors of players $v$ and $w$; each one of them has in-degree $d$ with edges from all $u_1, u_2, \ldots, u_d$. In what follows, fix $\lambda := 1- \frac{2}{2^d+1}$. The payoffs of player $v$ are as follows.
\begin{itemize}
	\item If $v$ plays \zero and at least one of $u_1,\ldots, u_d$ plays \zero, then the payoff for $v$ is 1.
	\item If $v$ plays \zero and every one of $u_1,\ldots, u_d$ plays \one, then the payoff for $v$ is 0.
	\item If $v$ plays \one and at least one of $u_1,\ldots, u_d$ plays \zero, then the payoff for $v$ is 0.
	\item If $v$ plays \one and every one of $u_1,\ldots, u_d$ plays \one, then the payoff for $v$ is $\lambda$.
\end{itemize}
The payoffs of player $w$ are as follows.
\begin{itemize}
	\item If $w$ plays \zero and every one of $u_1,\ldots, u_d$ plays \zero, then the payoff for $w$ is $\lambda$.
	\item If $w$ plays \zero and at least one of $u_1,\ldots, u_d$ plays \one, then the payoff for $w$ is 0.
	\item If $w$ plays \one and every one of $u_1,\ldots, u_d$ plays \zero, then the payoff for $w$ is 0.
	\item If $w$ plays \one and at least one of $u_1,\ldots, u_d$ plays \one, then the payoff for $w$ is $1$.
\end{itemize}
We are now ready to prove that this construction correctly simulates a \PURE gate. We consider the different cases that arise depending on the value encoded by $u$.
\begin{itemize}
	\item[{\bf --}] $s_u(\zero)=1$, i.e., $u$ encodes 0. From \cref{lem:graphical-wsne-aux} we know that $s_{u_i}(\zero)=1$ for every $i \in [d]$. Then, we have the following for players $v$ and $w$.
	\begin{itemize}
		\item[-] Player $v$ gets payoff 1 from action \zero and payoff 0 from action \one. Hence, in an \eps-WSNE we get that $s_v(\zero)=1$, and thus $v$ encodes 0.
		\item[-] Player $w$ gets payoff $\lambda$ from action \zero and payoff 0 from action \one. Hence, since $\eps < \lambda$, in an \eps-WSNE we get that $s_w(\zero)=1$, and thus $w$ encodes 0.
	\end{itemize}
	\item[{\bf --}] $s_u(\one)=1$, i.e., $u$ encodes 1. From \cref{lem:graphical-wsne-aux} we know that $s_{u_i}(\one)=1$ for every $i \in [d]$. Then, we have the following for players $v$ and $w$.
	\begin{itemize}
		\item[-] Player $v$ gets payoff 0 from action \zero and payoff $\lambda$ from action \one. Hence, since $\eps < \lambda$, in an \eps-WSNE we get that $s_v(\one)=1$, and thus $v$ encodes 1.
		\item[-] Player $w$ gets payoff 0 from action \zero and payoff 1 from action \one. Hence, in an \eps-WSNE we get that $s_w(\one)=1$, and thus $w$ encodes 1.
	\end{itemize}
	\item[{\bf --}] $s_u(\one) \in (0,1)$, i.e., $u$ encodes \garbo. Then each one of the auxiliary players $u_1, \ldots, u_d$ can play a different strategy. For each $i \in [d]$, denote $s_{u_i}(\one) = p_i$, i.e., $p_i \in [0,1]$ is the probability player $u_i$ assigns on action \one. Let $P:= \prod_{i \in [d]} p_i$ and $Q := \prod_{i \in [d]} (1-p_i)$.
	Then, we have the following two cases.
	\begin{itemize}
		\item[-] $P \leq 2^{-d}$. Then we focus on player $v$: action \zero yields
		expected payoff $1-P \geq 1-2^{-d}$, while action \one yields expected payoff
		$\lambda \cdot P \leq \lambda \cdot 2^{-d}$. Then, since $\eps < \lambda$, we
		have
        \begin{align*}
        \lambda \cdot 2^{-d} & < \lambda \cdot 2^{-d} + \lambda - \eps  \\
        &= 1-2^{-d} - \eps,
        \end{align*}
        where the equality on the second line is true because our choice of
        $\lambda = 1 - \frac{2}{2^d + 1}$ satisfies the equation
        $\lambda \cdot 2^{-d}  + \lambda = 1 - 2^{-d}$. Thus, the expected payoff
        of \zero is strictly more than the expected payoff of \one plus \eps.
        This means that in any \eps-WSNE, $s_v(\one) = 0$, or equivalently,
        $s_v(\zero) = 1$, i.e., $v$ encodes~0.
        
		\item[-] $P > 2^{-d}$. Then, it holds that $Q < 2^{-d}$, because
		$P\cdot Q = \prod_{i \in [d]} p_i \cdot (1-p_i) \leq (1/4)^d$.
		In this case we focus on player $w$: action \zero yields payoff $\lambda \cdot
		Q < \lambda \cdot 2^{-d}$, while action \one yields payoff $1-Q > 1-2^{-d}$.
        We can then use the same argument as the previous case to obtain 
		  $\lambda \cdot 2^{-d} < 1-2^{-d} - \eps$.
		Therefore, the expected payoff
        of \one is strictly more than the expected payoff of \zero plus \eps. This implies
		that in any \eps-WSNE, $s_w(\zero) = 0$, or equivalently, $s_w(\one) = 1$,
		i.e., $w$ encodes 1.
	\end{itemize}
    Hence, either player $v$ encodes $0$, or $w$ encodes $1$. So when $s_u(\one) \in (0,1)$, we cannot have that both $v$ and
    $w$ encode $\garbo$, as required. 
\end{itemize}

From the above analysis, we have that in any \wn{\eps} of the
graphical game, with $\eps < 1 - \frac{2}{2^d + 1}$, the players correctly encode a solution to the \pcircuit instance.
\begin{theorem}
	\label{thm:wsne-graph}
	Computing an \wn{\eps} in two-action graphical games with maximum in-degree $d \geq 2$ is \ppad/-complete for any $\eps < 1 - \frac{2}{2^d + 1}$.
\end{theorem}

We can see that the constructed game is not win-lose since there is a payoff $\lambda \notin \{0,1\}$ in the gadget that simulates \PURE gates. However, if we set $\lambda = 1$, and use verbatim the analysis from above, we will get \ppad/-hardness for \eps-WSNE with  $\eps < 1 - \frac{1}{2^{d-1}}$.

\begin{theorem}
	\label{thm:wsne-graph-win-lose}
	Computing an \wn{\eps} in two-action  win-lose graphical games with maximum in-degree $d \geq 2$ is \ppad/-complete for any $\eps < 1 - \frac{1}{2^{d-1}}$.
\end{theorem}

Since for every constant $\eps < 1$ we can find a constant $d$ such that \cref{thm:wsne-graph-win-lose} holds, we get the following corollary.

\begin{corollary}
	\label{cor:wsne-PPAD-any-eps}
	For any constant $\eps < 1$, computing an \eps-WSNE in two-action win-lose graphical games is \ppad/-complete.
\end{corollary}

\subsubsection{Upper and Lower Bounds for Approximate Nash Equilibria}

We now show results for $\eps$-NEs in graphical games. 

\paragraph{\bf A straightforward upper bound.}

We first show that a $1/2$-NE can easily be 
found in any two-action graphical game. 

\begin{theorem}
    \label{thm:eacy-graphical-ne}
	There is a polynomial-time algorithm that finds a $1/2$-NE in a two-action
	graphical game.
\end{theorem}
\begin{proof}
	Let $\vbs$ be the strategy profile in
	which all players mix uniformly over their two actions.
	Then, for each player $i$ we have
	\begin{align*}
		u_i(\vbs) &\ge 1/2 \cdot \br_i(\vbsi) \\
        &= \br_i(\vbsi) - 1/2 \cdot \br_i(\vbsi)     \\
		&\ge \br_i(\vbsi) - 1/2,
	\end{align*}
	where the final inequality used the fact that $\br_i(\vbsi) \in [0, 1]$. 
	Thus, $\vbs$ is a $1/2$-NE.
\end{proof}

\paragraph{\bf The lower bound.}

We will show that computing an $\eps$-NE of a graphical game is \ppad/-hard
for any constant $\eps < 1/2$ by a reduction from \pcircuit. 

Fix any $\eps < 1/2$ and let $\delta \in (0,1/2]$ be such that $1/2-\delta = \eps$. Given a \pcircuit instance, 
we build a two-action graphical game, where the two actions will be named 
$\zero$ and $\one$.
Each node $v$ of the \pcircuit instance will be represented by a set of $k$
players in the game, named $v_1, v_2, \dots, v_k$, where we fix $k$ to be an odd 
number satisfying $k \ge \ln (3/\delta) \cdot 18/\delta^2$. 
Since $\delta$ is constant, we have that $k$ is also constant.

The strategies of these players will
encode a solution to the \pcircuit problem in the following way. Given a
strategy profile \vbs, we define the assignment $\valonly$ such that
\begin{itemize}
	\item If $s_{v_i}(\zero) \ge 1/2 + \delta/3$ for all $i$, then $\val{v} = 0$.
	\item If $s_{v_i}(\one) \ge 1/2 + \delta/3$ for all $i$, then $\val{v} = 1$.
	\item In all other cases, $\val{v} = \garbo$.
\end{itemize}

We now give implementations for \NOT, \AND, and \PURE gates. We note that, in
all three cases, the payoff received by player $v_i$ is only affected by the
actions chosen by the players representing the inputs to the (unique) gate $g$
that outputs to $v$. Thus, we can argue about the equilibrium condition at $v_i$
by only considering the players involved in gate $g$, and we can ignore all
other gates while doing this.

\paragraph{\bf \NOT gates.}

For a gate $g = (\NOT, u, v)$, we use the following construction, which
specifies the games that will be played between the set of players that represent
$u$ and the set of players that represent $v$. 

Each player $v_i$ has incoming edges from all players $u_1, u_2, \dots, u_k$, and $v_i$'s payoff tensor is set as follows. 
\begin{itemize}
	\item If strictly more than\footnote{Since $k$ is odd, it is not possible for
		exactly $k/2$ players to play $\zero$.} $k/2$ of the players $u_1$ through $u_k$ play
	$\zero$, then $v_i$ receives payoff $0$ for strategy $\zero$ and payoff 1 for
	strategy $\one$.
	\item If strictly less than $k/2$ of the players $u_1$ through $u_k$ play
	$\zero$, then $v_i$ receives payoff $1$ for strategy $\zero$ and payoff 0 for
	strategy $\one$.
\end{itemize}

We now show the correctness of this construction. We start with a technical
lemma that we will use repeatedly throughout the construction.

\begin{lemma}
	\label{lem:payoff2prob}
	Suppose that player $p$ has two actions $a$ and~$b$. 
	In any $(1/2 - \delta)$-NE,
	if the payoff of
	action $a$ is at most $\delta/3$, and the payoff of action $b$ is at least $1
	- \delta/3$, then player~$p$ must play action $b$ with probability at least $1/2 + \delta/3$. 
\end{lemma}
\begin{proof}
	Since action $b$ has payoff at least $1 - \delta/3$, we have that the best
	response payoff to $p$ is also at least $1 - \delta/3$. Hence, in any strategy
	profile \vbs that is an $(1/2 -
	\delta)$-NE, we have
	\begin{align*}
		u_{p}(\vbs) &\ge \br_{p}(\vbs) - 1/2 + \delta \\
		&\ge 1 - \delta/3 - 1/2 + \delta \\
		&= 1/2 + 2\delta/3.
	\end{align*}
	Since the payoff of $a$ is bounded by $\delta/3$, and the payoff of $b$ is bounded
	by $1$ (since all payoffs are in the range $[0, 1]$), we obtain
	\begin{align*}
		u_{p}(\vbs) &\le s_{p}(a) \cdot \delta/3 + s_{p}(b) \cdot 1 \\
		&=  (1 - s_{p}(b)) \cdot \delta/3 + s_{p}(b) \cdot 1 \\
		&= s_{p}(b) (1 - \delta/3) + \delta/3 \\
		&\le s_{p}(b) + \delta/3.
	\end{align*}
	Joining the two previous inequalities gives 
	$s_{p}(b) + \delta/3 \ge 1/2 + 2\delta/3$,
	and therefore 
	$s_{p}(b) \ge 1/2 + \delta/3$. 
\end{proof}

The following bound will also be used repeatedly in our proofs.

\begin{lemma}\label{prop:hoeffding}
    If $k \geq \ln (6/\delta) \cdot 9/2\delta^2$ and $N$ is a random variable
    distributed binomially according to $B(k, 1/2 + \delta/3)$, then $\Pr(N \le
    k/2) \leq \delta/6$. 
\end{lemma}

\begin{proof}
    We can use the standard Hoeffding bound~\cite{hoeffding1994probability} for the binomial distribution. This gives
    \begin{align*}
        \Pr(N \le k/2) &\leq \exp\left(-2k\left(1/2 + \delta/3 - \frac{k/2}{k}\right)^2\right) \\
        &= \exp(-\frac{2}{9} k \delta^2 ) \\
        &\leq \exp(-\ln(6/\delta) ) \\
        &= \delta/6.
    \end{align*}
\end{proof}

We can now prove that the \NOT gadget operates correctly. 

\begin{lemma}
	In every $(1/2-\delta)$-NE, the following properties hold.
	\begin{itemize}
		\item If the players representing $u$ encode $0$, then the players
		representing $v$ encode $1$.
		\item If the players representing $u$ encode $1$, then the players
		representing $v$ encode $0$.
	\end{itemize}
\end{lemma}
\begin{proof}
	Let $\vbs$ be a $(1/2-\delta)$-NE. 
	We begin with the first claim. Since the players representing $u$ encode a $0$,
	we have that $s_{u_j}(\zero) \ge 1/2 + \delta/3$ for all~$j \in [k]$. 
	
	We start by proving an upper bound on the payoff of action \zero for player
	$v_i$. The payoff of this action increases as the players $u_j$ place
	less probability on action \zero, so we can assume 
	$s_{u_j}(\zero) = 1/2 + \delta/3$, since this minimizes the payoff of \zero to
	$v_i$. 
	
	Under this assumption, let $N$ be the number of players $u_j$ that play action \zero, and observe that $N \sim B(k, 1/2 + \delta/3)$. 
	Since $k \geq \ln(3/\delta) \cdot 18/\delta^2 \geq \ln (6/\delta) \cdot 9/2\delta^2$, from \cref{prop:hoeffding} we have $\Pr(N \le k/2) \leq \delta/6 \leq \delta/3$.
	Hence the payoff of action $\zero$ to player $v_i$ is at most
	$\delta/3$, and therefore the payoff of action $\one$ to $v_i$ is at least $1 -
	\delta/3$.

	So we can apply \cref{lem:payoff2prob} to argue that $s_{v_i}(\one) \ge 1/2
	+ \delta/3$. Since this holds for all~$i$, we have that the players representing
	$v$ encode the value $1$ in the \pcircuit instance, as required.
	
	The second case can be proved in an entirely symmetric manner.
\end{proof}

\paragraph{\bf \AND gates.} 

For a gate $g = (\AND, u, v, w)$ we use the following construction. Each player
$w_i$ has in-degree $2k$ and has incoming edges from all of the players $u_1, u_2, \dots, u_k$, and all of the players $v_1, v_2, \dots, v_k$. 
The payoff tensor of $w_i$ is as follows.

\begin{itemize}
	\item If strictly more than $k/2$ of the players $u_1$ through $u_k$ play \one,
	and strictly more than $k/2$ of the players $v_1$ through $v_k$ play \one, 
	then $w_i$ receives payoff $0$ from action $\zero$ and payoff 1 from
	action $\one$.
	
	\item If this is not the case, then 
	$w_i$ receives payoff $1$ from action $\zero$ and payoff 0 from action $\one$. 
\end{itemize}

The correctness of this construction is shown in the following pair of lemmas.

\begin{lemma}
	In every $(1/2-\delta)$-NE of the game, if the players representing $u$ encode value
	$1$, and the players representing $v$ encode value $1$, then the players
	representing $w$ will encode value~$1$.
\end{lemma}
\begin{proof}
	Let $\vbs$ be a $(1/2-\delta)$-NE. From the assumptions about $u$ and $v$, we have that
	$s_{u_j}(\one) \ge 1/2 + \delta/3$ for all $j$, and $s_{v_j}(\one) \ge 1/2 +
	\delta/3$ for all $j \in [k]$.
	
	We start by proving an upper bound on the payoff of $\zero$ to $w_i$. Since this payoff
	decreases as the players $u_j$ and $v_j$ place more probability on $\one$, we
	can assume that 
	$s_{u_j}(\one) = 1/2
	+ \delta/3$ for all $j$, and $s_{v_j}(\one) = 1/2 + \delta/3$ for all $j$, since
	this maximizes the payoff of $\zero$ to $w_i$. 
	
	Let $N$ be the number of players $u_j$ that play \one, and observe that $N \sim B(k, 1/2 + \delta/3)$. Likewise, the number of players $v_j$ that
	play \one are distributed according to the same distribution. Using \cref{prop:hoeffding} and the fact that $k
	\geq \ln(3/\delta) \cdot 18/\delta^2 \geq \ln (6/\delta) \cdot 9/2\delta^2$, we get
	$\Pr(N \le k/2) \leq \delta/6$. Then, using the union bound, we can derive an
	upper bound of $\delta/3$ for the probability that strictly less than $k/2$ of
	the players $u_1$ through $u_k$ play \one or strictly less than $k/2$ of the
	players $v_1$ through $v_k$ play \one.
	
	Hence, the payoff of $\zero$ to player $w_i$ is at most $\delta/3$, and so the
	payoff of $\one$ to player $w_i$ is at least $1 - \delta/3$. We can then apply
	\cref{lem:payoff2prob} to argue that $w_i$ must play $\one$ with
	probability at least $1/2 + \delta/3$. Since this holds for all $i$, we have that
	$w_1, w_2, \ldots w_k$ correctly encode value $1$.
\end{proof}

\begin{lemma}
	In every $(1/2-\delta)$-NE of the game, if the players representing $u$ encode value
	$0$, or the players representing $v$ encode value $0$, then the players
	representing $w$ will encode value $0$.
\end{lemma}

\begin{proof}
	We will provide a proof for the case where the players representing $u$ encode
	value $0$, since the other case is entirely symmetric. 
	
	Let $\vbs$ be an $(1/2-\delta)$-NE. By assumption we have that $s_{u_j}(\zero) \ge 1/2 +
	\delta/3$ for all $j$.
	We start by proving an upper bound on the payoff of $\one$ to $w_i$. Since the
	payoff of this strategy decreases as $u_j$ places more probability on $\zero$,
	we can assume that 
	$s_{u_j}(\zero) = 1/2 + \delta/3$ for all $j$, since this minimizes the payoff of
	$\one$ to $w_i$. 
	
	Let $N$ be the number of players $u_j$ that play \one and observe that $N \sim B(k, 1/2 + \delta/3)$. By \cref{prop:hoeffding}, given that $k \geq \ln(3/\delta) \cdot 18/\delta^2 \geq \ln (6/\delta) \cdot 9/2\delta^2$, we get $\Pr(N \le k/2) \leq \delta/3$.
	Hence, the payoff of $\one$ to player $w_i$ is at most $\delta/3$, and so the
	payoff of $\zero$ to player $w_i$ is at least $1 - \delta/3$. Thus, we can apply
	\cref{lem:payoff2prob} to argue that $w_i$ must play $\zero$ with
	probability at least $1/2 + \delta/3$. Since this holds for all $i$, we have that
	$w_1, w_2, \ldots w_k$ correctly encode value~$0$.
\end{proof}

\paragraph{\bf \PURE gates.}

For a gate $g = (\PURE, u, v, w)$, we use the following construction. Each
player $v_i$ will have incoming edges from all players $u_1, u_2, \dots, u_k$, with 
their payoff tensor set as follows.
\begin{itemize}
	\item If strictly more than $(1/2 - \delta/6) \cdot k$ of the players $u_1$
	through $u_k$ play $\one$, then $v_i$ receives payoff $0$ from
	action $\zero$ and payoff $1$ from action $\one$. 
	
	\item If this is not the case, then $v_i$ receives payoff $1$ from action
	$\zero$ and payoff $0$ from action $\one$.
	
\end{itemize}
Each player $w_i$ 
will have incoming edges from all players $u_1, u_2, \dots, u_k$, with their payoff tensor set as follows.
\begin{itemize}
	\item If strictly more than $(1/2 + \delta/6) \cdot k$ of the players $u_1$
	through $u_k$ play $\one$, then $w_i$ receives payoff $0$ from 
	action $\zero$ and payoff $1$ from action $\one$. 
	
	\item If this is not the case, then $w_i$ receives payoff $1$ from action
	$\zero$ and payoff $0$ from action $\one$.
	
\end{itemize}
The following lemma analyzes the behavior of the players representing $v$.

\begin{lemma}
	\label{lem:nepurev}
	Let $\vbs$ be a $(1/2 - \delta)$-NE, and let $N$ be a random variable that denotes the number of
	players $u_i$ playing strategy $\one$ under $\vbs$. 
	\begin{itemize}
		\item If $E[N] \le (1/2 - \delta/3) \cdot k$, then the players representing $v$
		will encode value $0$. 
		
		\item If $E[N] \ge 1/2 \cdot k$, then the players representing $v$ will encode
		value $1$. 
	\end{itemize}
\end{lemma}
\begin{proof}
	For the first part of the statement, notice that $N$ follows the binomial distribution with $E[N] \le (1/2 - \delta/3) \cdot k$, by applying Hoeffding's
	inequality~\cite{hoeffding1994probability}, we get
	\begin{align*}
		\Pr(N \ge (1/2 - \delta/6) \cdot k) 
		&\le \Pr(N - E[N] \ge \delta/6 \cdot k) \\
		& \le \exp\left( \frac{-2 (k \cdot \delta/6)^2}{k} \right) \\
		& = \exp \left( -k \cdot \delta^2/18 \right) \\
		& \le \delta/3,
	\end{align*}
    where the last inequality holds due to the fact that $k \ge \ln (3/\delta) \cdot 18/\delta^2$.
	Therefore, the payoff of strategy $\one$ to $v_i$ is at most $\delta/3$, and so
	the payoff of strategy $\zero$ to $v_i$ is at least $1 - \delta/3$. Thus we can
	apply \cref{lem:payoff2prob} to argue that 
	$s_{v_i}(\zero) \ge 1/2 + \delta/3$. Since this holds for all $i$, we have that
	the players representing $v$ encode value $0$, as required.

    For the second part of the statement, we get 
    \begin{align*}
		\Pr(N < (1/2 - \delta/6) \cdot k) 
		&\le \Pr(N - E[N] < - \delta/6 \cdot k) \\
		& \le \exp\left( \frac{-2 (k \cdot \delta/6)^2}{k} \right) \\
		& \le \delta/3,
	\end{align*}
    where the argument is the same as in the previous case.
    Then, also similarly to the previous case, from \cref{lem:payoff2prob} we can deduce that for every $i$, $s_{v_i}(\one) \ge 1/2 + \delta/3$. Therefore, the players representing $v$ encode value $1$, as required.
\end{proof}

\noindent
The next lemma analyzes the behavior of the players representing $w$; its proof is omitted since it is entirely symmetric to the proof of \cref{lem:nepurev}.

\begin{lemma}
	Let $\vbs$ be a $(1/2 - \delta)$-NE, and let $N$ be a random variable that denotes the number of
	players $u_i$ playing strategy $\one$ under $\vbs$. 
	\begin{itemize}
		\item If $E[N] \le 1/2 \cdot k$, then the players representing $w$
		will encode value $0$. 
		
		\item If $E[N] \ge (1/2 + \delta/3) \cdot k$, then the players representing $w$ will encode
		value $1$. 
	\end{itemize}
\end{lemma}

\noindent
Combining the two previous lemmas, we can see that the construction correctly
simulates a \PURE gate.

\begin{itemize}
	\item If the players representing $u$ encode value $0$, then $E[N] \le 1/2 -
	\delta/3$, and so both the players representing $v$ and those representing $w$ encode value $0$. 
	\item If the players representing $u$ encode value $1$, then $E[N] \ge 1/2 +
	\delta/3$, and so both the players representing $v$ and those representing $w$ encode value $1$. 
	\item In all other cases we can verify that either the players representing $v$ or the players representing $w$ encode
	a $0$ or a $1$. Specifically, if $E[N] \le 1/2 \cdot k$ then the players representing $w$ encode value $0$,
	while if $E[N] \ge 1/2 \cdot k$, then the players representing $v$ encode value $1$.
\end{itemize}

\paragraph{\bf The hardness result.}

From the arguments given above, we have that in an $(1/2 - \delta)$-NE of the
graphical game, the players correctly encode a solution to the \pcircuit
instance.
Note also that, since \cref{cor:pcircuit-restricted} gives hardness for
\pcircuit even when the total degree of each node is 3, the graphical game that we
have built has total degree at most $3k$. Thus, the game can be built in polynomial time.

\begin{theorem}
    \label{thm:graphical-wsne-hardness}
	It is \ppad/-hard to find an $\eps$-NE in a 
	two-action graphical game for any constant $\eps < 1/2$.
\end{theorem}

In fact, the payoff entries in all gadgets are 0 or 1. Thus, our \ppad/-hardness result holds for win-lose games too.

\begin{corollary}
	For any constant $\eps < 1/2$, it is \ppad/-hard to find an $\eps$-NE in a 
	two-action win-lose graphical game.
\end{corollary}

\subsection{Threshold Games}
{\em Threshold games} were introduced by Papadimitriou and Peng~\cite{PapadimitriouP21-threshold-games} as an intermediate problem that was used to prove \ppad/-hardness for public good games on directed networks. Since then, threshold games have been used to prove \ppad/-hardness for throttling equilibria in auction markets~\cite{ChenKK21-throttling} and for the famous Hylland-Zeckhauser scheme~\cite{ChenCPY22-HZ-hardness}. 

\paragraph{\bf Threshold Game.} A threshold game $\mathcal{G}(V,E)$ is defined on a directed graph $G=(V,E)$. 
Every node $v \in V$ represents a player with strategy space $\val{v} \in
[0,1]$. For every $v\in V$, we use $N_v: = \{u \in V: (u,v) \in E \}$ to denote the set of nodes with outgoing edges towards $v$. Then, we say that $\valonly:=(\val{v} : v \in V) \in [0,1]^{|V|}$ is an $\eps$-approximate equilibrium if every $\val{v}$ satisfies
\begin{align*}
    \val{v} \in 
    \begin{cases} 
    [0, \eps], & \text{if}~ \sum_{u \in N_v} \val{u} > 1/2 + \eps; \\
    [1- \eps, 1], & \text{if}~ \sum_{u \in N_v} \val{u} < 1/2 - \eps;\\
    [0, 1], & \text{if}~ \sum_{u \in N_v} \val{u} \in [1/2 - \eps, 1/2 + \eps].
    \end{cases}
\end{align*}

\begin{definition}[\eps-\threshold]
Let $\eps \in [0,1]$. An instance of the \eps-\hspace{0pt}\threshold problem consists of a threshold game $\mathcal{G}(V,E)$. The task is to find an \eps-approximate equilibrium for $\mathcal{G}(V,E)$.
\end{definition}

Papadimitriou and Peng~\cite{PapadimitriouP21-threshold-games} proved that
there exists a constant $\eps' > 0$ such that $\eps'$-\threshold is
\ppad/-complete, via a reduction from \eps-\gcircuit, where $\eps' < \eps/5$.
Hence, from our hardness result for \gcircuit (\cref{thm:gcircuithard}) we
obtain that \eps-\threshold is \ppad/-complete for some \eps upper bounded by
$1/50$.

Using a direct reduction from \pcircuit we can show that the problem is in fact \ppad/-complete for any $\eps < 1/6$. Furthermore, this is tight, since we provide a simple algorithm solving the problem in polynomial time for $\eps = 1/6$. We first present the algorithm achieving the upper bound, and then the improved lower bound through a direct reduction from \pcircuit.

\subsubsection{Computing a \texorpdfstring{$\boldsymbol{1/6}$}{1/6}-Approximate Equilibrium}

\begin{theorem}
The $1/6$-\threshold problem can be solved in polynomial time.
\end{theorem}

\begin{proof}
Let $G=(V,E)$ be the threshold game graph. The algorithm proceeds as follows.
\begin{enumerate}
    \item For every node $v \in V$ with in-degree at least two, i.e., $|N_v| \geq 2$, set $\val{v} := 1/6$. Remove all incoming edges of $v$ from the graph. Note that every node in the graph now has in-degree at most one. 
    \item Repeat until there are no more directed cycles: Pick a directed cycle, and for each node $v$ on the cycle, set $\val{v} := 1/2$. Remove all the edges on the cycle from the graph.
    \item Repeat until all nodes have been assigned a value: Since the graph is now acyclic, there exists a node $v \in V$ with unassigned value such that $v$ has in-degree zero, or such that its incoming neighbour $w \in N_v$ has already been assigned a value. If $\sum_{u \in N_v} \val{u} > 1/2 + 1/6$, set $\val{v} := 1/6$. Otherwise, set $\val{v} := 1$.
\end{enumerate}
The algorithm clearly runs in polynomial time. It remains to prove that all nodes satisfy the $\eps$-approximate equilibrium condition for $\eps = 1/6$. Note that when we assign a value to a node, all of its original incoming edges are still present in the graph. Furthermore, we only assign a value to a node once. We proceed by considering the nodes in each step separately.

For nodes which get assigned in Step 3, the $1/6$-approximate equilibrium condition is immediately satisfied because of the way in which we pick the value to assign. For a node which gets assigned in Step 2, note that it must have in-degree one, and thus its unique incoming neighbour also lies on the same directed cycle, and also gets assigned value $1/2$, which implies that the equilibrium condition is satisfied. Finally, note that the algorithm only assigns values $\geq 1/6$ to nodes. As a result, for any node $v$ which gets assigned in Step 1, we have $\sum_{u \in N_v} \val{u} \geq |N_v|/6 \geq 1/3$, since $|N_v| \geq 2$. In particular, assigning value $1/6$ to $v$ satisfies the $1/6$-approximate equilibrium condition, because $1/3 \geq 1/2-1/6$.
\end{proof}

\subsubsection{Hardness of \texorpdfstring{$\boldsymbol{\eps}$}{ε}-Approximate Equilibrium for any \texorpdfstring{$\boldsymbol{\eps < 1/6}$}{ε < 1/6}}

\begin{theorem}
\label{thm:threshold-hard}
\eps-\threshold is \ppad/-complete for every $\eps < 1/6$, even when the in- and out-degree of each node is at most two.
\end{theorem}

\begin{proof}
In order to prove that \eps-\threshold is \ppad/-hard for $\eps < 1/6$, we will reduce from \pcircuit that uses the gates \NOT, \NOR, and \PURE (we include the \NOT gate because we will later make use of \cref{cor:pcircuit-restricted} to argue about the degrees of nodes). We will encode $0$ values in the \pcircuit problem as values in the range $[0,\eps]$ in \threshold, while $1$ values will be encoded as values in the range $[1-\eps, 1]$. Then, each gate of \pcircuit will be simulated by the corresponding gadget depicted in \cref{fig:threshold-gadgets}.

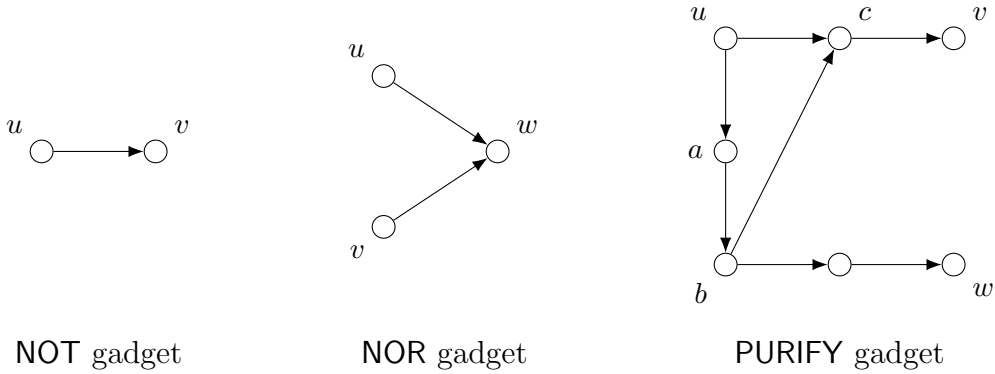
\begin{figure}
	\centering
	\begin{tikzpicture}[roundnode/.style={circle, draw=black, inner sep=0, minimum size=3mm}]

\node[roundnode,label=above left:{$u$}] (NOTu) at (-3.5,0.5) {};
\node[roundnode,label=above right:{$v$}] (NOTv) at (-2,0.5) {};

\draw[-{Latex[length=2mm]}] (NOTu) -- (NOTv);

\node[roundnode,label=above left:{$u$}] (NORu) at (1,1.5) {};
\node[roundnode,label=below left:{$v$}] (NORv) at (1,-0.5) {};
\node[roundnode,label=above right:{$w$}] (NORw) at (2.5,0.5) {};

\draw[-{Latex[length=2mm]}] (NORu) -- (NORw);
\draw[-{Latex[length=2mm]}] (NORv) -- (NORw);

\node[roundnode,label=above left:{$u$}] (PUREu) at (5.5,2) {};

\node[roundnode,label=left:{$a$}] (PUREa) at (5.5,0.5) {};
\node[roundnode,label=below left:{$b$}] (PUREb) at (5.5,-1) {};
\node[roundnode] (PUREbw) at (7,-1) {};
\node[roundnode,label=below right:{$w$}] (PUREw) at (8.5,-1) {};

\node[roundnode,label=above right:{$c$}] (PUREc) at (7,2) {};
\node[roundnode,label=above right:{$v$}] (PUREv) at (8.5,2) {};

\draw[-{Latex[length=2mm]}] (PUREu) -- (PUREa);
\draw[-{Latex[length=2mm]}] (PUREa) -- (PUREb);
\draw[-{Latex[length=2mm]}] (PUREb) -- (PUREbw);
\draw[-{Latex[length=2mm]}] (PUREbw) -- (PUREw);

\draw[-{Latex[length=2mm]}] (PUREu) -- (PUREc);
\draw[-{Latex[length=2mm]}] (PUREb) -- (PUREc);
\draw[-{Latex[length=2mm]}] (PUREc) -- (PUREv);

\node at (-2.75,-2.2) {\large \NOT gadget};
\node at (1.8,-2.2) {\large \NOR gadget};
\node at (7,-2.2) {\large \PURE gadget};

\end{tikzpicture}
	\caption{The gadgets used for \threshold.}\label{fig:threshold-gadgets}
\end{figure}

\paragraph{\bf \NOT gates.} A $(\NOT, u, v)$ gate will be simulated by the \NOT
gadget depicted in \cref{fig:threshold-gadgets}. We argue that the gadget works
for any $\eps < 1/4$, and thus in particular for any $\eps < 1/6$.
\begin{itemize}
    \item If $\val{u} \leq \eps < 1/2 - \eps$, then according to the definition of threshold games in any \eps-approximate equilibrium it must hold that $\val{v} \geq 1-\eps$. Hence, when the input encodes a $0$, the output value will encode a $1$.
    \item If $\val{u} \geq 1 - \eps > 1/2 + \eps$, then according to the definition of threshold games in any \eps-approximate equilibrium it must hold that $\val{v} \leq \eps$. Hence, when the input encodes a $1$, the output value will encode a $0$.
\end{itemize}

\paragraph{\bf \NOR gates.} A $(\NOR, u, v, w)$ gate will be simulated by the \NOR gadget depicted in \cref{fig:threshold-gadgets}. We claim that the gadget works for any $\eps < 1/6$.
\begin{itemize}
    \item If both $\val{u}$ and $\val{v}$ are in $[0, \eps]$, then $\val{u} + \val{v} \leq 2 \eps < 1/2 - \eps$. Hence, in any \eps-approximate equilibrium it must hold that $\val{w} \geq 1- \eps$. In other words, when both input values encode $0$s, then the output value of the \NOR gadget will encode $1$.
    \item If at least one of $\val{u}$ and $\val{v}$ is in $[1-\eps, 1]$, then $\val{u} + \val{v} \geq 1-\eps > 1/2 + \eps$. Hence, in any \eps-approximate equilibrium it must hold that $\val{w} \leq \eps$. In other words, when at least one input value encodes a $1$, then the output value of the \NOR gadget will encode $0$.
\end{itemize}

\paragraph{\bf \PURE gates.} A $(\PURE, u, v, w)$ gate will be simulated by the
\PURE gadget depicted in \cref{fig:threshold-gadgets}. We argue that the gadget works for any $\eps < 1/6$.
\begin{itemize}
    \item If $\val{u} \leq \eps$, then, in particular, $\val{u} < 1/2 - \eps$, and thus $\val{a} \geq 1-\eps$. Then, since $\val{a} \geq 1-\eps > 1/2+\eps$, it follows that $\val{b} \leq \eps$. By applying the same arguments again, we obtain that $\val{w} \leq \eps$. Furthermore, since $\val{u}+\val{b} \leq 2\eps < 1/2-\eps$, it follows that $\val{c} \geq 1-\eps$, and thus $\val{v} \leq \eps$. In other words, when the input value encodes a $0$, then both output values of the \PURE gadget will encode $0$s.
    \item If $\val{u} \geq 1-\eps$, then it is not hard to see that $\val{b} \geq 1-\eps$, and thus $\val{w} \geq 1-\eps$. In addition, since $\val{u} + \val{b} \geq 2(1-\eps) > 1/2+\eps$, it follows that $\val{c} \leq \eps$, and thus $\val{v} \geq 1-\eps$. In other words, when the input value encodes a $1$, then both output values of the \PURE gadget will encode $1$s.
    
    \item Finally, it remains to prove that no matter what the input $u$ is, at least one of the two output values $v, w$ of the \PURE gadget will encode a pure bit value. Assume that $w$ does not encode a pure bit value, i.e., $\val{w} \in (\eps, 1-\eps)$. We will show that in that case it must be that $\val{v} \notin (\eps, 1-\eps)$. Since $\val{w} \in (\eps, 1-\eps)$, it follows that $\val{b} \in [1/2-\eps,1/2+\eps]$. Similarly, since $\val{b} \in [1/2-\eps,1/2+\eps] \subseteq (\eps,1-\eps)$, it also follows that $\val{u} \in [1/2-\eps,1/2+\eps]$. As a result, we obtain that $\val{u} + \val{b} \geq 2(1/2-\eps) > 1/2+\eps$, which implies that $\val{c} \leq \eps$, and thus $\val{v} \geq 1-\eps$, as desired.
\end{itemize}

\paragraph{\bf The reduction.}
Given a \pcircuit instance over variables $V$, we produce a \threshold instance by replacing each gate with the corresponding gadget depicted in \cref{fig:threshold-gadgets}. Then, given a solution $\valonly$ to the \eps-\threshold instance with $\eps < 1/6$, we can produce a solution $\valtwoonly$ for \pcircuit as follows. For every $v \in V$:
\begin{itemize}
    \item if $\val{v} \leq \eps$, then we set $\valtwo{v} = 0$;
    \item if $\val{v} \geq 1-\eps$, then we set $\valtwo{v} = 1$;
    \item else we set $\valtwo{v} = \garbo$.
\end{itemize}
The correctness of the reduction follows from the arguments above. Finally, by using \cref{cor:pcircuit-restricted}, the constructed instances of \threshold will all satisfy that the in- and out-degree of every node is at most two. Indeed, the new ``auxiliary'' nodes introduced by the \PURE gadget all satisfy this, and the contribution of each gadget to the in- and out-degree of ``original'' nodes is exactly the same as the contribution of the corresponding \pcircuit gates in the interaction graph.
\end{proof}

\appendix

\section{Further Details on \pcircuit}

\subsection{On the Definition of \pcircuit}\label{app:sec:pcircuit-discussion}

In this section, we explore different ways to weaken the definition of \pcircuit, and show how, in each case, the problem is no longer \ppad/-hard.

\paragraph{\bf No \PURE gate.}
If we allow all gates, except the \PURE gate (so only \NOT, \COPY, \OR, \AND, \NOR, \NAND), then the problem becomes polynomial-time solvable. Indeed, it suffices to assign value $\garbo$ to all the nodes.

\paragraph{\bf No Negation.}
If we allow all gates, except the ones that perform some kind of negation (so only \PURE, \COPY, \OR, \AND), then the problem becomes polynomial-time solvable. Indeed, assigning the value $1$ to each node (or, alternatively, the value $0$ to each node) always yields a solution. More generally, we can make the following observation: for any set of gates that can be implemented by monotone functions, the problem lies in the class \pls/, and is thus unlikely to be \ppad/-complete. Indeed, as already mentioned in \cref{sec:pcircuitdef}, we can view any \pcircuit instance with $n$ nodes as a function $F: [0,1]^n \to [0,1]^n$, where each gate is replaced by a continuous function that is consistent with the gate-constraint. Then any fixed point of $F$ yields a solution to the \pcircuit instance. If each of the gates can be replaced by a continuous \emph{monotone} function, then the problem of finding a fixed point of $F$ is an instance of Tarski's fixed point theorem, which is known to lie in \pls/~\cite{EtessamiPRY20-Tarski}.

\paragraph{\bf No robustness.}
If we allow all gates (\PURE, \NOT, \COPY, \OR, \AND, \NOR, \NAND), but we drop the robustness requirement from the logical gates \OR, \AND, \NOR, \NAND, then the problem can be solved in polynomial time.

Construct the interaction graph $G$ of the \pcircuit instance, as defined in \cref{sec:pcircuit-more-gates}. In the first stage of the algorithm, as long as there exists a directed cycle in graph $G$, we do the following:
\begin{enumerate}
    \item Pick an arbitrary directed simple cycle of $G$.
    \item For each node $u$ on the simple cycle $C$, assign value $\garbo$ to it, i.e., $\val{u} := \garbo$.
    \item Remove all nodes on the simple cycle $C$ from the graph $G$, including all their incident edges.
\end{enumerate}
At the end of this procedure, $G$ no longer contains any cycles. In the second stage of the algorithm we then repeat the following, until $G$ is empty:
\begin{enumerate}
    \item Pick any source $u$ of $G$ (which must exist, since $G$ contains no cycles).
    \item Let $g$ be the (unique) gate that has $u$ as output. Since $u$ does not have incoming edges in $G$, all inputs of $g$ have already been assigned a value. If $u$ is the only output of $g$, then assign a value to $u$ that satisfies the gate, and remove $u$ and its edges from $G$. If the gate $g$ also has another output $v$, then $g$ is a \PURE gate, and there are two cases:
    \begin{itemize}
        \item If $v$ has not been assigned a value yet, then assign values to both $u$ and $v$ such that the gate is satisfied. Remove $u$ and $v$ and their edges from $G$.
        \item If $v$ has already been assigned a value, then this happened in the first stage of the algorithm, and both $v$ and the input to $g$ were assigned value $\garbo$ (if $v$ lies on a simple cycle $C$, then so does the input of gate $g$, because $v$ has a single incoming edge in the original interaction graph). In that case, we assign value $0$ or $1$ to $u$ and $g$ is satisfied. Remove $u$ and its edges from $G$ ($v$ was already removed in the first stage).
    \end{itemize}
\end{enumerate}

Since a node is removed from $G$ only when it is assigned a value, all nodes have been assigned a value at the end of the algorithm. We argue that all gates are satisfied. Clearly, any gate that has an output node that is still present in $G$ after the end of the first stage will be satisfied (by construction of the second stage). Thus, it remains to consider any gate $g$ such that all its output nodes are removed in the first stage. There are three cases:
\begin{itemize}
    \item $g$ is a \NOT or \COPY gate: if the output lies on a simple cycle $C$, then so does the input. Both are thus assigned value $\garbo$ and the gate is satisfied.
    \item $g$ is a (non-robust) \OR, \AND, \NOR, or \NAND gate: if the output lies on a simple cycle $C$, then so does at least one of its inputs. Thus, at least one input is also assigned value $\garbo$, and the gate is satisfied. Note that we crucially used the non-robustness of the gate here.
    \item $g$ is a \PURE gate: if an output $v$ of $g$ lies on a simple cycle $C$, then the input $u$ of $g$ also lies on $C$, and so both are removed at the same time from $G$. However, the other output $w$ of $g$ cannot lie on that same simple cycle $C$, and after $u$ is removed, $w$ does not have an incoming edge anymore and will thus not be removed in the first stage. Thus, $g$ cannot be a \PURE gate.
\end{itemize}

\paragraph{\bf Adding Constant Gates.}
It is a natural idea to try to add constant gates in an attempt to make the problem hard for a set of gates for which the problem is not \ppad/-hard. By constant gates, we mean a \textsf{0}-gate which has one output and no input, and which enforces that the value of its output node always be $0$, and a \textsf{1}-gate defined analogously. No matter which subset of gates $S \subseteq \{\PURE, \NOT, \COPY, \OR, \AND, \NOR, \NAND\}$ we use, adding the constant gates \textsf{0} and \textsf{1} does not change the complexity of the problem. Indeed, it is easy to see that the constants can be ``propagated'' through the circuit. If a constant is an input to a \PURE, \NOT, or \COPY gate, then we can replace the output(s) of that gate by constants that satisfy the gate-constraint. If a constant is an input to an \textsf{X} gate, where $\textsf{X} \in \{\OR, \AND, \NOR, \NAND\}$, then there are two cases: either we can replace the output by a constant, or we can replace the gate by an \textsf{X} gate with the same input twice (namely, the other input). In order to create a copy of the other input, we can either use the \COPY gate, or the \NOT gate, or the \PURE gate. If none of these three gates lies in $S$, then \pcircuit with gates $S \cup \{\textsf{0}, \textsf{1}\}$ is polynomial-time solvable, since it is polynomial-time solvable with gates $S \cup \{\NOT, \textsf{0}, \textsf{1}\}$ (by the propagation argument, and the lack of \PURE gate).

\subsection{More Structure: Proof of Corollary~\ref*{cor:pcircuit-restricted}}\label{sec:app:proof-pcircuit-restricted}

Let $\X$ and $\Y$ be as required by the statement of \cref{cor:pcircuit-restricted}. By \cref{cor:pcircuit-gates} it immediately follows that \pcircuit with gates $\{\PURE, \X, \Y\}$ is \ppad/-complete. Now consider an instance of \pcircuit with gates $\{\PURE, \X, \Y\}$. We will explain how to turn it into an instance that satisfies the three conditions of \cref{cor:pcircuit-restricted}. We proceed in three steps, where each step adds more structure to the instance, without destroying any of the structure introduced in a previous step.

\textbf{Step 1.} First of all, for every node $u$ that is used as an input to $k$ gates $g_1, \dots, g_k$, we construct a binary tree consisting of \PURE gates that is rooted at $u$ and has leaves $u_1, \dots, u_k$. Then, we modify each gate $g_i$ so that it uses $u_i$ as input instead of $u$. As a result, the instance now satisfies that every node is the input of at most one gate. Note that for any solution $\valonly$, if $\val{u} \in \{0,1\}$, then $\val{u_i} = \val{u}$ for all $i$. On the other hand, if $\val{u} = \garbo$, then the $\val{u_i}$ could have different values. Despite this, we argue that this is a valid reduction, namely that restricting $\valonly$ to the original nodes yields a solution to the original instance. To see this, it suffices to notice that, no matter what the type $T$ of gate $g_i$ is, if $u$ satisfies the condition on the left of the ``$\implies$'' sign in the definition of $T$, then the condition is still satisfied if we replace $u$ by $u_i$.

\textbf{Step 2.} After the first step, the instance now satisfies that every node is the input of at most one gate, in particular $(d_{in},d_{out}) \in \{1,2\} \times \{0,1,2\}$. In the next step, we will further enforce that $(d_{in},d_{out}) \neq (2,2)$ and make the interaction graph bipartite. To do this, we replace every node $u$ by a small gadget that consists of nodes $u_a$, $u_b$, and some other auxiliary nodes. If $u$ was the output of some gate $g$, then $g$ now has output $u_a$ instead. Similarly, if $u$ was the input of some gate $g'$, then $g'$ has input $u_b$ instead. The gadget will ensure that $\val{u_a} \in \{0,1\} \implies \val{u_b} = \val{u_a}$. Thus, by the same argument as in the previous paragraph, restricting $\valonly$ to $\{u_a: u \in V\}$ yields a solution to the original instance. In the case $\X = \COPY$, the gadget is simply implemented by using a \COPY gate with input $u_a$ and output $u_b$, i.e., a $\COPY(u_a,u_b)$ gate. Applying this transformation to every node $u$ achieves two things: (i) the interaction graph only contains nodes with $(d_{in},d_{out}) \in (\{1,2\} \times \{0,1,2\}) \setminus \{(2,2)\}$, (ii) the interaction graph is bipartite ($A = \{u_a: u \in V\}$, $B = \{u_b: u \in V\}$). In the case $\X = \NOT$, the gadget uses auxiliary (new) nodes $v, w, w'$ and consists of the gates $\NOT(u_a,v)$, $\PURE(v,w,w')$, and $\NOT(w,u_b)$. It is easy to check that this gadget satisfies the same properties as before, in particular (i) and (ii). For (ii), note that we add nodes $w,w'$ to set $A$ and node $v$ to set $B$ to make the graph bipartite.

\textbf{Step 3.} It remains to enforce that every node is the input of \emph{exactly} one gate. Currently, this only holds with ``at most'' instead of ``exactly''. We begin by introducing a new node $v^*$. Note that $v^*$ is not used as an input of any gate, and it is also the only node that is not used as an output by any gate. Let $V_{sink}$ be the set of all nodes that are not used as an input by any gate. In particular, $v^* \in V_{sink}$. Using a binary tree of $\Y$-gates which takes the nodes in $V_{sink}$ as inputs, we can ensure that the root $u_{sink}$ of the binary tree is the only node that is not used as input by any gate. By using additional $\X$-gates at the leaves (when needed) we can ensure that the graph remains bipartite. Note that $v^*$ is still the only node that is not used as an output of any gate. We finish the construction by adding an \X-gate with input $u_{sink}$ and output $v^*$. If that makes the graph non-bipartite, then we can instead introduce a new node $v$ and gates $\X(u_{sink},v)$ and $\X(v,v^*)$ instead. The instance now satisfies all three conditions of the statement.

\section{Bimatrix Games}\label{app:sec:bimatrix}

In this section we show lower bounds for finding \emph{relative} approximate
well-supported equilibria in bimatrix games. Unlike the \emph{additive}
approximations we have studied so far, in a relative approximate equilibrium
$\eps$ is not measuring the difference between the current strategy's payoff and a deviation's payoff, but instead, the \emph{ratio} of that difference over the deviation's payoff. 

Daskalakis~\cite{Daskalakis13-approximate-Nash} proved that computing a \relwn{\eps} in
bimatrix games with payoffs in $[-1,1]$ is \ppad/-complete for any constant
$\eps \in [0,1)$. This left open the question of inapproximability of bimatrix
games with non-negative payoffs\footnote{Relative \eps-WSNE (as well as relative
\eps-NE) are scale invariant. Therefore, multiplying all payoffs by a positive
constant does not affect the equilibria for a given \eps.}. It
is known that for bimatrix games with payoffs in
$[0,1]$ there is a polynomial time algorithm for finding a \relwn{1/2}~\cite{FederNS07-Nash-small-supp}. The best
hardness result for games with payoffs in $[0, 1]$ was given by
Rubinstein~\cite{Rubinstein18-Nash-inapproximability},
who showed that there exists a constant $\eps$ such that \relwn{\eps}
is \ppad/-hard.

Here we show that the problem is hard for any $\eps \leq 1/57$. 
A bimatrix game is a special case of a polymatrix game (see \cref{sec:polymatrix}) where we only have two players connected by an edge. Hence, we use the same notation as that of the aforementioned section.
Before presenting this section's main result, let us define the equilibrium notion that will be studied in this section.

\begin{definition}[\relwn{\eps}]\label{def:rel wsne}
    The strategy profile $\vb{s}$ is a relative \eps-well-supported Nash equilibrium (\relwn{\eps}) if for every player, each action in the support of her strategy is at least as good as any other action up to a fraction \eps of the latter action's absolute expected utility. In particular, for non-negative utilities we have 
    \begin{align*}
\forall i \in [n], \forall k \in \supp(s_i), \quad u_i(k,\vbsi) \geq (1-\eps) \cdot \br_{i}(\vbsi).
    \end{align*}
\end{definition}

\begin{theorem}\label{thm:bimatrix-rel-WSNE}
    Computing a \relwn{\eps} in a bimatrix game with non-negative payoffs is \ppad/-complete for any $\eps \leq 1/57$.
\end{theorem}

The remainder of this section is devoted to the proof of this theorem.

\paragraph{\bf Hardness.}

We will reduce from the \ppad/-complete problem of computing an (additive) \wn{\eps'} in a polymatrix game with $2n$ players (nodes) and two actions per player for $\eps' < 1/3$ (as shown in \cref{thm: polymatrix wsne}) to the problem of computing a \relwn{\eps} in a $2n \times 2n$ bimatrix game for $\eps = \eps'/\beta$, where $\beta = 18.9$. Our reduction closely follows that of Rubinstein's \cite{Rubinstein18-Nash-inapproximability}, but instead of reducing from (additive) $\eps'$-NE in polymatrix games, we reduce from their \wn{\eps'} counterpart for which we obtained hardness for significantly higher $\eps'$. 
    
The construction of the bimatrix game is identical to that of the aforementioned paper, but we optimize the involved parameters so that we get the highest \eps possible. In particular, since again we reduce from a \emph{bipartite} polymatrix game (see \cref{thm: polymatrix wsne}) we consider the two sides of the graph, namely the $R$-side and the $C$-side. Without loss of generality, we consider that each side of the graph has $n$ nodes, since we can add ``dummy'' nodes to the smallest of the two sides. We then consider two super-players, namely the \emph{row player} $R$ and the \emph{column player} $C$ for the two sets of nodes. Each super-player ``represents'' her nodes in the new bimatrix game, meaning that her action set consists of the union of the individual nodes' action sets. We embed the polymatrix game into the bimatrix game, and multiply all the corresponding payoffs with a positive constant $\lambda < 1$ (to be defined later) so that all payoffs are normalized in $[0,\lambda]$. Then, to this game we add an ``imitation game'' as follows. To the payoff matrix of the row player we add a ``block identity matrix'' $\mathbb{1}_{b(2 \times 2)}$ which is a $2n \times 2n$ matrix with its diagonal $2 \times 2$ blocks having 1's and the rest 0's. To the payoff matrix of the column player we add a modified ``block identity matrix'' which is like $\mathbb{1}_{b(2 \times 2)}$ but with its entries shifted by two columns to the right (modulo $2n$). 

For simplicity, given a strategy of a player in the bimatrix game, we refer to the total probability mass of the two actions corresponding to the polymatrix node $i \in [n]$ as the \emph{mass of node $i$} and we denote it as $x(i), y(i)$ for the row and column player, respectively. The actions 0, 1 of node $i$ and their probability masses are denoted by $(i:a)$, and $x(i:a)$, $y(i:a)$, respectively, for $a \in \{0,1\}$. For ease of presentation in what follows, we will refer to the players of the polymatrix game as ``nodes'' and reserve the word ``players'' for the bimatrix game's participants. We denote by $N_{R}(i)$ the neighbourhood of node $i$ that belongs to the $R$-side of the polymatrix game (and respectively $N_{C}(i)$ for a node $i$ on the $C$-side).
    
First, we present modified versions of the respective results of \cite{Rubinstein18-Nash-inapproximability} tailored to the needs of our reduction. Since their proofs are almost identical to those of the aforementioned paper, we have preserved most of their notation.

\begin{lemma}[Modified Lemma 9.3 of \cite{Rubinstein18-Nash-inapproximability}]
    In every $(x,y)$ \relwn{\eps}, for $\lambda < \frac{1-\eps}{2}$, $x(i), y(i) \in \left[ \frac{1 - \eps - 2 \lambda}{n}, \frac{(1 - \eps - 2 \lambda)^{-1}}{n} \right]$.
\end{lemma}

\begin{proof}
      Let $t_R = \max_{i} x(i)$ and $i^{*}_R = \arg \max_{i} x(i)$, and let us call $U_{R}((i:a), y)$ the (expected) payoff that the row player gets from playing action $(i:a)$, $a \in \{0,1\}$ (respectively, $U_{C}(x, (i:a))$ for the column player). We will show that, for any $j \in [n]$, if $x(j) < (1 - \eps - 2 \lambda) t_R$ then $y(j+1) = 0$. For any action $(j+1:a)$ that the column player picks, her payoff is at most $U_{C}(x, (j+1:a)) \leq x(j) + 2 \cdot \lambda \cdot t_R$. That is due to the imitation game that gives her payoff at most $x(j)$ and the payoffs induced due to her neighbours $k \in N(j+1)$ in the polymatrix game that give her positive payoff\footnote{From the \ppad/-hard \wn{\eps} instances we construct for polymatrix games in \cref{sec:WSNE-polymatrix}, observe that at most two neighbours (out of three) give positive payoff to the node. Namely, only its parents can give it positive payoff while its children always give it 0 payoff.}, combined with the fact that $x(k) \leq t_R$. Therefore, $U_{C}(x, (j+1:a)) < (1 - \eps) t_R$, but the column player can guarantee a payoff of at least $t_R$ by playing action $(i^{*}_R + 1 : a)$ for any $a \in \{ 0, 1 \}$. This would contradict the \wn{\eps} condition, therefore $y(j+1) = 0$. Similarly, for any $j \in [n]$, if $y(j) < (1 - \eps - 2 \lambda) t_C$ then $x(j) = 0$, where $t_C = \max_{i} y(i)$.
      
      Now we claim that for every $i \in [n]$, $x(i), y(i) > 0$. For the sake of contradiction assume that $x(i) = 0$ without loss of generality. Then, since $\lambda < (1-\eps)/2$ and $t_{R} \geq 1/n > 0$, we get $x(i) < (1 - \eps - 2 \lambda) t_R$, and therefore $y(i+1) = 0$. With a similar argument for $y(i+1)$ we deduce from the previous paragraph that $x(i+1) = 0$, and this inductively yields that for all $i \in [n]$, $x(i) = y(i) = 0$, a contradiction. Therefore, for all $i \in [n]$, $x(i), y(i) > 0$. This, combined with the above paragraph yields that for all $i \in [n]$, $x(i) \geq (1 - \eps - 2 \lambda) t_R \geq (1 - \eps - 2 \lambda) /n$.
      
      The upper bound is deduced by the fact that if $t_R > (1 - \eps - 2 \lambda)^{-1}/{n}$. Since there exists a node $i$ with $x(i) \leq 1/n$ (otherwise $\sum_{j \in [n]} x(j) > 1$), we have that $x(i) < (1 - \eps - 2 \lambda) t_R$, which cannot happen as shown in the first paragraph of the proof. Therefore, for any $i \in [n]$, $x(i) \leq t_R \leq (1 - \eps - 2 \lambda)^{-1}/{n}$. A similar argument holds for $t_C$ which yields the upper bound for $y(i)$.
\end{proof}

\begin{corollary}[Modified Corollary 9.4 of \cite{Rubinstein18-Nash-inapproximability}]\label{cor:util-bounds-relWSNE}
    In every $(x,y)$ \relwn{\eps}, for $\lambda < \frac{1-\eps}{2}$, the expected utilities of both players for playing any action against the strategy of the other player are in $\left[ \frac{1 - \eps - 2 \lambda}{n}, \frac{(1 + 2\lambda)(1 - \eps - 2 \lambda)^{-1}}{n} \right]$.
\end{corollary}

\begin{proof}
    The lower bound comes from the lower bound of the above lemma and the fact that the row player from any action $(i:a)$, $a \in \{ 0, 1 \}$ against $y$, gets at least the payoff of the imitation game, multiplied by $y(i:a) + y(i:1-a) = y(i)$. The upper bound comes from the upper bound of the above lemma and the fact that the row player from any action $(i:a)$, $a \in \{ 0, 1 \}$ against $y$ gets from each entry $((i:a), j)$, $j \in [n]$ of the matrix at most payoff $1 \cdot y(i)$ from the imitation game, and additional payoff of at most $\lambda \cdot y(i)$ from each of two neighbouring nodes in the polymatrix game (note again that only two out of maximum three neighbours can give her positive payoff). A symmetric argument holds for the column player's bounds.
\end{proof}

Fix $\beta = 18.9$. Now we will show that, given an $\eps'<1/3$, for any $(x,y)$ which is a \relwn{\eps'/\beta} in the constructed bimatrix game, the marginal distributions of the actions $(i:a)$, $a \in \{0,1\}$ for each node $i \in [n]$ constitute an (additive) $\eps'$-WSNE of the original polymatrix game. In particular, we will prove that the strategy profile where each node $i$ in the polymatrix game belonging to the $R$-side of the bipartite graph plays $(x(i:0)/x(i), x(i:1)/x(i))$, and similarly each node $j$ of the $C$ side of the bipartite graph plays $(y(j:0)/y(j), y(j:1)/y(j))$ is an \wn{\eps'}. 

For the sake of contradiction, assume that the marginal distributions of the node strategies are not such an \wn{\eps'}. Then, given the marginals induced by $(x,y)$ there is a node $i$ on the left side of the bipartite graph (without loss of generality) for whom playing action $(i:a)$ for some $a \in \{0,1\}$ against $y$ gives her additive $\eps'$ more expected utility than what the action $(i:1-a)$ gives. This discrepancy, translated in the bimatrix game (where we have multiplied all payoffs of the polymatrix game by $\lambda$) means that the difference in payoff would be at least $\eps' \cdot \lambda \cdot \frac{1 - \eps - 2 \lambda}{n}$ due to \cref{cor:util-bounds-relWSNE}, where $\eps = \eps'/\beta$. Then, again by \cref{cor:util-bounds-relWSNE} which bounds her maximum utility from playing any action $(i:a)$ against the mixed strategy $y$, we conclude that her relative increase in expected payoff is at least
\begin{align*}
    \frac{ u_{R}((i:1-a), y) + \eps' \cdot \lambda \cdot \frac{1 - \eps - 2 \lambda}{n} }{u_{R}((i:1-a), y)} &\geq 1 + \frac{ \eps' \cdot \lambda \cdot \frac{1 - \eps - 2 \lambda}{n} }{\frac{(1 + 2\lambda)(1 - \eps - 2 \lambda)^{-1}}{n}}  \quad \text{(by \cref{cor:util-bounds-relWSNE})} \\
  &= 1 + \eps' \cdot \frac{\lambda}{1 + 2\lambda} \cdot (1 - \eps - 2 \lambda)^{2}\\
  &= 1 + \beta \cdot \eps \cdot \frac{\lambda}{1 + 2\lambda} \cdot (1 - \eps - 2 \lambda)^{2} \quad \text{(since $\eps = \eps'/\beta$)} .
\end{align*}
We now find the value for $\lambda$ that maximizes the above expression, namely, $\lambda = \frac{-3 + \sqrt{17-8\eps}}{8}$. For this value of $\lambda$ and for every $\eps = \eps'/\beta \leq 1/57$ we have that 
\begin{align*}
    1 + \beta \cdot \eps \cdot \frac{\lambda}{1 + 2\lambda} \cdot (1 - \eps - 2 \lambda)^{2} > \frac{1}{1-\eps}
\end{align*}
Note that the optimum value of $\lambda$ is potentially irrational, and therefore, not suitable as input to the algorithm of our reduction. Nevertheless, the proof goes through for $\lambda = 0.1383$.
The above strict inequality shows that the action that the node deviated to gives the bimatrix player more than relative $1/(1-\eps)$ expected utility, and therefore by definition our initial profile $(x,y)$ is
not a \relwn{\eps} (for $\eps = \eps'/\beta$), a contradiction. This completes the proof of \cref{thm:bimatrix-rel-WSNE}.

\begin{remark}
    In fact, \cref{thm:bimatrix-rel-WSNE} holds for any $\eps<\frac{1}{3 \beta^*}$, where $\beta^*$ is the optimum value from the above proof. This $\beta^*$ is in $(18.86, 18.87)$, but we have picked $18.9$ for ease of presentation.
\end{remark}

\bigskip
\subsubsection*{Acknowledgements}
We thank the anonymous reviewers for comments and suggestions that helped improve the presentation of the paper. We are also very grateful to Steffen Schuldenzucker for feedback on an earlier version of the manuscript, and to Christian Ikenmeyer for pointing out the connection to hazard-free circuits.
The first author was supported by EPSRC Grant EP/X039862/1 ``NAfANE: New Approaches for Approximate Nash Equilibria''. The second author was supported by EPSRC grant EP/W014750/1 ``New Techniques for Resolving Boundary Problems in Total Search''.

\bibliographystyle{alphaurl}
\bibliography{references}

\end{document}